\documentclass[aop]{imsart}

\RequirePackage{amsthm,amsmath,amsfonts,amssymb}
\RequirePackage[numbers]{natbib}

\startlocaldefs

\usepackage{amsmath, amsthm, latexsym, amssymb}
\usepackage[unicode]{hyperref}
\usepackage{srcltx}
\usepackage{bbm}
\usepackage{xcolor}
\usepackage{graphicx}

\usepackage{enumitem}

\def\@rmrk#1#2{\refstepcounter
    {#1}\@ifnextchar[{\@yrmrk{#1}{#2}}{\@xrmrk{#1}{#2}}}

\usepackage{mathtools}
\mathtoolsset{showonlyrefs}

\makeatletter\@addtoreset{equation}{section}\makeatother

 \newfont{\bfit}{cmbxti10 scaled 1200}

\renewcommand{\d}{{\rm d}}
 \newcommand{\e}{{\rm e} }

 \newcommand{\eps}{\varepsilon}

 \newcommand{\R}{\mathbb{R}}
 \newcommand{\N}{\mathbb{N}}
 \newcommand{\Z}{\mathbb{Z}}

 \newcommand{\E}{\mathbb{E}}
 \renewcommand{\P}{\mathbb{P}}
 \def\1{{\mathchoice {1\mskip-4mu\mathrm l} 
{1\mskip-4mu\mathrm l}
{1\mskip-4.5mu\mathrm l} {1\mskip-5mu\mathrm l}}}

 \newcommand{\Bcal}{{\mathcal B}}
 \newcommand{\Mcal}{{\mathcal M}}

\newcommand{\heap}[2]{\genfrac{}{}{0pt}{}{#1}{#2}}

\newcommand{\ssup}[1] {{\scriptscriptstyle{({#1}})}}

\usepackage{mathrsfs}

\makeatletter
\def\blfootnote{\xdef\@thefnmark{}\@footnotetext}
\makeatother

\newcommand{\bea}{\begin{eqnarray}}
\newcommand{\eea}{\end{eqnarray}}
\newcommand{\<}{\langle}
\renewcommand{\>}{\rangle}

\newcommand{\wt}{\widetilde}

\newcommand{\wh}{\widehat}

\def\eps{{\varepsilon}}

\def\bP{{\boldsymbol{P}}}

\def\cF{{\mathcal F}}

\def\cX{{\mathcal X}}

\def\Z{{\mathbb Z}}

\def\good{\textbf{\texttt{Good}}}
\def\verygood{\textbf{\texttt{VG}}}
\def\std{\footnotesize\textbf{\texttt{std}}}
\def\sstd{\footnotesize\textbf{\texttt{sstd}}}

\def\loong{\footnotesize\textbf{\texttt{long}}}

\def\R{{\mathbb R}}

\def\Z{{\mathbb Z}}

\def\de{{\rm d}}

\def\<{\langle}
\def\>{\rangle}

\def\cL{{\mathcal L}}

\def\P{\mathbb{P}}
\def\cE{{\mathcal E}}

\def\b0{{\boldsymbol{0}}}

\def\oq{\overline{q}}

\def\cI{{\mathcal I}}
\def\cJ{{\mathcal J}}

\def\oC{\overline{C}}

\renewcommand{\b}{\mathbf{b}}

\def\lt{\left}
\def\rt{\right}

\def\eps{\varepsilon}

\def\bbE{{\mathbb{E}}}

\def\bbP{{\mathbb{P}}}

\def\bP{\boldsymbol{P}}
\def\bbR{{\mathbb{R}}}

\def\bbZ{{\mathbb{Z}}}

\def\cF{{\mathcal{F}}}

\def\cP{{\mathcal{P}}}

\def\bP{\mathbf{P}}

\newcommand{\tmtextit}[1]{\text{{\itshape{#1}}}}

\usepackage{color-edits}
\addauthor{rev}{red}   
\addauthor{ms}{blue}

\renewcommand{\subsection}{\secdef \subsct\sbsect}
\newcommand{\subsct}[2][default]{\refstepcounter{subsection}
\vspace{0.15cm}
{\flushleft\bf \arabic{section}.\arabic{subsection}~\bf #1  }
\nopagebreak\nopagebreak}
\newcommand{\sbsect}[1]{\vspace{0.1cm}\noindent
{\bf #1}\vspace{0.1cm}}

{\nopagebreak {\hfill\rule{2mm}{2mm}}\\ }

\newtheorem{theorem}{Theorem}[section]
\newtheorem{lemma}[theorem]{Lemma}
\newtheorem{cor}[theorem]{Corollary}
\newtheorem{prop}[theorem]{Proposition}

\newtheorem{definition}[theorem]{Definition}

\newtheoremstyle{thm}{1.5ex}{1.5ex}{\itshape\rmfamily}{}
{\bfseries\rmfamily}{}{2ex}{}

\newtheoremstyle{rem}{1.3ex}{1.3ex}{\rmfamily}{}
{\itshape\rmfamily}{}{1.5ex}{}
\theoremstyle{rem}
\newtheorem{remark}{{\slshape\sffamily Remark}}[]

\refstepcounter{subsubsection}

\def\thebibliography#1{\section*{References}
  \list%
  {\arabic{enumi}.}
    {\settowidth\labelwidth{[#1]}\leftmargin\labelwidth
    \advance\leftmargin\labelsep
    \parsep0pt\itemsep0pt
    \usecounter{enumi}}
    \def\newblock{\hskip .11em plus .33em minus .07em}
    \sloppy                   
    \sfcode`\.=1000\relax}

\begin{document}

\begin{frontmatter}

\title{Effective mass of the Fr\"ohlich Polaron and the Landau-Pekar-Spohn conjecture}



\author[A]{\fnms{Rodrigo}~\snm{Bazaes}\ead[label=e1]{rbazaes@uni-muenster.de}},
\author[A]{\fnms{Chiranjib}~\snm{Mukherjee}\ead[label=e2]{chiranjib.mukherjee@uni-muenster.de}},
\author[C]{\fnms{Mark}~\snm{Sellke}\ead[label=e3]{msellke@fas.harvard.edu}}
\and
\author[D]{\fnms{S. R. S. }~\snm{Varadhan}\ead[label=e4]{varadhan@cims.nyu.edu}}

\address[A]{
University of M\"unster\printead[presep={,\ }]{e1,e2}}

\address[C]{
Harvard University\printead[presep={,\ }]{e3}}

\address[D]{
Courant Institute of Mathematical Sciences\printead[presep={,\ }]{e4}}

\begin{keyword}[class=MSC]
\kwd[Primary ]{60J65}
\kwd{60F10}
\kwd{81S40}
\kwd{60G55}
\end{keyword}

\begin{keyword}
\kwd{Fr\"ohlich polaron}
\kwd{Effective Mass}
\kwd{Landau-Pekar theory}
\kwd{Spohn's conjecture}
\kwd{Strong coupling}
\kwd{Pekar variational formula}
\kwd{Pekar process}
\kwd{Point processes}
\kwd{Large deviations}
\end{keyword}

\renewcommand*{\thefootnote}{\arabic{footnote}}

\begin{abstract}
 We prove that there is a constant $\overline C\in (0,\infty)$ such that the effective mass $m(\alpha)$ of the Fr\"ohlich Polaron satisfies $m(\alpha) \geq \overline C \alpha^4$, which is sharp according to a long-standing prediction of Landau-Pekar \cite{LP49} from 1948 and of Spohn \cite{S87} from 1987. 
 
 The method of proof, which demonstrates how the sharp quartic divergence rate of $m(\alpha)$ appears in a natural way, is based on analyzing the Gaussian representation of the Polaron measure and that of the associated tilted Poisson point process developed in \cite{MV18a}. Additionally, our technique here leads to accompanying results including, 1) an explicit identification of local interval process from \cite{MV18a} in the strong coupling limit  in terms of functionals of the {\it Pekar process} \cite{MV18b}, 2) strict monotonicity of the effective mass $m(\alpha)$ for all $\alpha>0$ and 3) the quartic divergence of $m(\alpha)$ for a generalized class of Polaron type interactions.
\end{abstract}

\end{frontmatter}



\section{Introduction and main result}

The  {\it{Polaron problem}}  in quantum statistical mechanics is inspired by studying the slow movement 
of a charged particle, e.g., an electron, in a crystal whose lattice sites are polarized by this slow motion. The electron then drags around it a cloud of polarized lattice 
points which influences and determines the {\it{effective behavior}} of the electron. 
A key quantity is given by the bottom of the spectrum 
$E_\alpha(P)= \inf \,\mathrm{spec}(H_P)$
of the (fiber) Hamiltonian of the Fr\"ohlich Polaron. 
It is known that $E_\alpha(\cdot)$ is rotationally symmetric and is analytic when $P\approx 0$. Then the central objects of interest  are the {\it ground state energy}  
\begin{equation}\label{def-g-alpha}
g(\alpha)=- \min_P E_\alpha(P)
\end{equation}
 as well as the {\it effective mass} $m(\alpha)$ of the Fr\"ohlich Polaron, defined as the inverse of the curvature:
 \begin{equation}\label{def-m-alpha}
 m(\alpha)= \bigg[\frac{\partial^2}{\partial P^2}E_\alpha(P)\big|_{P=0}\bigg]^{-1}.
 \end{equation} 
See \cite{S87,DS19}. Physically relevant questions concern the {\it strong-coupling behavior} of $g(\alpha)$ and $m(\alpha)$. Indeed, the ground state energy in this regime was studied by Pekar \cite{P49}, who also conjectured that the limit 
\begin{equation}\label{Pekar-conj}
\begin{aligned}
&g_0:= \lim_{\alpha\to\infty}\frac{g(\alpha)}{\alpha^2} \qquad\mbox{exists, and }\\
&g_0= \sup_{\heap{\psi\in H^1(\R^3)}{|\psi|_2=1}} \bigg[\iint_{\R^3\times \R^3} \frac{\psi^2(x) \psi^2(y)}{|x-y|} \d x \d y - \frac 12 \|\nabla \psi\|_2^2\bigg].
\end{aligned}
\end{equation}
By a well-known result of E. Lieb \cite{L76}, the above variational formula $g_0$ admits a rotationally symmetric, smooth and centered maximizer $\psi_0\in H^1(\R^3)$ with $\|\psi_0\|_2=1$ which is unique except for spatial translations. 
One can also obtain a probabilistic representation for $g(\alpha)$. Indeed, Feynman's path integral formulation \cite{F72} leads to $g(\alpha)= \lim_{T\to\infty} \frac 1 T \log\langle \Psi | \e^{-TH}|\Psi\rangle$ with $\Psi$ being chosen such that its spectral resolution contains the ground state energy or low energy spectrum of $H$, but is otherwise arbitrary. Then the Feynman-Kac formula for the semigroup $\e^{-TH}$ implies that the last expression can be rewritten further as 
 \begin{equation}\label{g}
 g(\alpha)=\lim_{T\to\infty} \frac 1 T\log \E_0\bigg[\exp\bigg\{\alpha \int_0^T\int_0^T \d s \d t\,\, \frac{\e^{-{|t-s|}}}{|\omega(t)-\omega(s)|}\bigg\}\bigg],
 \end{equation}
 with $\E_0$ denoting expectation w.r.t. the law of a three-dimensional Brownian path starting at $0$. Starting with this expression and using large deviation theory from \cite{DV83-4}, Pekar's conjecture \eqref{Pekar-conj} was proved in \cite{DV83}. Later, a different proof was given by E. Lieb and L. Thomas \cite{LT97} using a functional analytic approach which also provided quantitative error bounds. 
 
 \medskip

As for the effective mass defined in \eqref{def-m-alpha}, according to a long-standing conjecture by Landau-Pekar \cite{LP49} and by H. Spohn \cite{S87}, $m(\alpha)$ should diverge like $\alpha^4$ with a pre-factor given by the centered solution $\psi_0$ 
of the Pekar variational problem \eqref{Pekar-conj}  
 in the strong coupling limit $\alpha\to\infty$. With this background, the main result of this article is to show the following theorem:
 \begin{theorem}\label{thm}
 There exists a constant $\overline C\in (0,\infty)$ such that for all $\alpha\geq 1$, 
 \begin{equation}\label{thm-eq}
 \frac{m(\alpha)}{\alpha^4} \geq \overline C.
  \end{equation}
 \end{theorem} 
 As a consequence of the method developed to prove Theorem \ref{thm}, we will also obtain auxiliary results including
 \begin{itemize}
     \item an explicit description of the the point process method for Polaron measure from \cite{MV18a} in terms of the coupling parameter $\alpha\to\infty$ and the {\it Pekar process} \cite{MV18b} (Corollary \ref{cor:consequences-number-intervals}), 
     
     \item {\it strict} monotonicity of the effective mass $m(\alpha)$ for all $\alpha>0$ (Corollary \ref{cor:m-eff-monotone}) and 
     
     \item Theorem \ref{thm} for a generalized class of Polaron interactions (Remark \ref{rem:general-p} \& Appendix \ref{subsec:universality}). 
     \end{itemize}
     In addition to the proof method, these results will be discussed in Section \ref{sec-proof-outline}. Before that, let us provide some background. 

 \subsection{Background: Polaron path measure.}
 In 1987, H. Spohn \cite{S87} established a link between the effective mass $m(\alpha)$ and 
 the actual {\it path behavior} under the {\it Polaron measure}. Indeed, the exponential weight on the right hand side in \eqref{g} defines a tilted measure on the path space of the Brownian motion, or rather, on the space of increments of Brownian paths. 
More precisely, let $\P=\P_T$ be the law of the Brownian increments $\{\omega(t)-\omega(s)\}_{-T\leq s < t \leq T}$ for three dimensional Brownian motion. Then the {\it Polaron measure} is defined as the transformed measure 
\begin{equation}\label{def-polaronmeas}
\widehat\P_{\alpha,T}(\d\omega)=\frac 1 {Z_{\alpha,T}} \exp\bigg(\frac\alpha2\int_{-T}^T\int_{-T}^T\frac{\e^{-|t-s|}}{|\omega(t)-\omega(s)|} \d t\d s\bigg) \P(\d\omega),
\end{equation}
where 
$$
Z_{\alpha,T}=\E^\P\bigg[\exp\bigg(\frac\alpha2\int_{-T}^T\int_{-T}^T\frac{\e^{-|t-s|}}{|\omega(t)-\omega(s)|} \d t\d s\bigg)\bigg] 
$$
is the total mass of the exponential weight, or the {\it partition function}. 

\smallskip

It was conjectured by Spohn in \cite{S87} that for any fixed coupling $\alpha>0$ and as $T\to\infty$,
 the distribution of the diffusively rescaled Brownian path under the Polaron measure should be asymptotically Gaussian with zero mean and variance $\sigma^2(\alpha)>0$.
The following results were shown in \cite{MV18a,MV21}: 
\begin{enumerate}
\item [(1)] For any $\alpha>0$, 
the infinite-volume Polaron measure 
$$
\widehat\P_\alpha= \lim_{T\to\infty} \widehat\P_{\alpha,T}
$$
 exists, and it is an explicit mixture of Gaussian measures, see Sec. \ref{sec-duality} for details. 
 \item [(2)] 
 For any $\alpha>0$, the distribution of the rescaled Brownian increments $\frac{\omega(T)-\omega(-T)}{\sqrt{2T}}$ under $\widehat\P_{\alpha,T}$ and $\widehat\P_\alpha$ satisfies a central limit theorem, i.e., for any $\alpha>0$, 
\begin{equation}\label{eq-CLT}
\begin{aligned}
\lim_{T\to\infty} \widehat\P_{\alpha,T}\bigg[ \frac{\omega(T)-\omega(-T)}{\sqrt{2T}} \in \cdot\bigg]&= \lim_{T\to\infty} \widehat\P_{\alpha}\bigg[ \frac{\omega(T)-\omega(-T)}{\sqrt{2T}} \in \cdot\bigg] \\
&= \mathbf N(0,\sigma^2(\alpha)\mathbf I_{3\times 3}),\quad\mbox{where}
\end{aligned}
\end{equation} 
\begin{equation}\label{CLTvariance}
\begin{aligned}
\sigma^2(\alpha)&= \lim_{T\to\infty} \frac 1 {2T} \E^{\widehat\P_{\alpha,T}}\big[\big|\omega(T)- \omega(-T)|^2\big] \\
&=\lim_{T\to\infty} \frac 1 {2T} \E^{\widehat\P_{\alpha}}\big[\big|\omega(T)- \omega(-T)|^2\big]\in (0,1). 
\end{aligned}
\end{equation}
See Sec. \ref{sec-duality}. 
\end{enumerate}
In the above display, $\mathbf N(0,\sigma^2(\alpha)\mathbf I_{3\times 3})$ denotes the law of a three-dimensional Gaussian vector with mean zero and covariance matrix $\sigma^2(\alpha) \mathbf I_{3\times 3}$. We refer to \cite{BP21} for an extension of these results using the method from \cite{MV18a} to other polaron-type interactions, and to \cite{M17} for a different method for treating translation-invariant interactions that are either slowly decaying in time and bounded in space, or short-range in time and singular in space. 

Returning to \eqref{eq-CLT}-\eqref{CLTvariance}, we note that the strict bound $\sigma^2(\alpha)< 1$ in \eqref{CLTvariance} for any coupling $\alpha>0$ reflects the attractive nature of the interaction defined in \eqref{def-polaronmeas}.
Assuming the validity of the above CLT \eqref{eq-CLT}, already in \cite{S87} Spohn proved a simple relation between the effective mass $m(\alpha)$ and the CLT variance 
$\sigma^2(\alpha)$:
\begin{equation}\label{eq-mass-variance}
m(\alpha)^{-1} =\sigma^2(\alpha) \qquad\mbox{for any }\alpha>0,
\end{equation}
see also Dybalski-Spohn \cite{DS19} for a recent proof of the above relation using \eqref{eq-CLT}. In \cite{S87}, Spohn also conjectured that the
the strong coupling behavior of the infinite-volume limit $\lim_{\alpha\to\infty}\lim_{T\to\infty} \, \widehat\P_{\alpha,T}=\lim_{\alpha\to\infty}\P_\alpha$, suitably rescaled, 
should converge to the so-called {\it Pekar process}, which is a diffusion process with generator 
$$
\frac 12 \Delta+ \frac{\nabla \psi}{\psi}\cdot\nabla, 
$$
where $\psi$ is any solution of the variational problem \eqref{g}. This conjecture was proved in \cite{MV18b} showing that, the distribution of the rescaled process 
 $\big(\alpha \big(\omega(\frac{t}{\alpha^2}) - \omega(\frac{s}{\alpha^2})\big)\big)_{s\in A, t\in B}$ under $\widehat\P_\alpha$ 
converges as $\alpha\to\infty$ to a unique limit which is the stationary version of the {\it increments} of the Pekar process.\footnote{Here $A$ and $B$ are any intervals of fixed length. We also note that, by Brownian scaling, the convergence of the abovementioned rescaled process under $\widehat\P_\alpha$ is equivalent to the convergence of the distribution of $(\omega_t - \omega_s)_{s\in A, t\in B}$ under $\widehat\P^{\mathrm{Kac}}_\alpha$ as $\alpha\to\infty$ to the same limit. Here $\widehat\P^{\mathrm{Kac}}_\alpha= \lim_{T\to\infty} \widehat\P^{\mathrm{Kac}}_{\alpha,T}$, with the latter defined as in \eqref{def-polaronmeas}, but with interaction strength $\alpha^{-2}\e^{-\alpha^{-2}|t-s|}$ instead of $\alpha \e^{-|t-s|}$ and 
$\varepsilon:= \alpha^{-2}$ is referred to as the {\it Kac parameter}.} We note that the Pekar process was also earlier constructed 
in \cite{MV14,KM15,BKM15} as the infinite-volume limit of the {\it mean-field Polaron}  -- convergence of the latter towards the Pekar process was also conjectured by Spohn in \cite{S87}.

   Based on the path behavior of $\widehat\P_\alpha$,  in \cite{S87}  the decay rate of the CLT diffusion constant $\sigma^2(\alpha) \sim \alpha^{-4}$ as $\alpha\to\infty$ was also derived heuristically -- note that, given the relation \eqref{eq-mass-variance}, 
this decay rate would be equivalent to the divergence rate $m(\alpha)\sim \alpha^4$, conjectured by Landau and Pekar \cite{LP49}. 
Using a functional analytic route from \cite{LT97}, it was shown in \cite{LS20} that $\lim_{\alpha\to\infty} m(\alpha)=\infty$.  
By means of probabilistic techniques from \cite{MV18a,MV21}, it has been recently shown in \cite{BP22} that $\sigma^2(\alpha) \leq c \alpha^{-2/5}$ for some $c<\infty$. Very recently, 
using the probabilistic representation of the Polaron measure \eqref{def-polaronmeas} but invoking Gaussian correlation inequalities orthogonally to the current method, 
it has been shown in \cite{S22} that $m(\alpha) \geq C \alpha^{4}/(\log \alpha)^6$. 
For the corresponding upper bound $m(\alpha) \leq C^* \alpha^4$, we refer to the very recent article 
\cite{BS22} that used a functional analytic route, combined with a concavity result from \cite{Pol22} which used the probabilistic route and point process method from \cite{MV18a}. The method currently developed for obtaining Theorem \ref{thm} for the Fr\"ohlich Polaron is quite different from the ones found in the literature for showing previous bounds on $m(\alpha)$. We will outline this approach below and explain along the lines how the $\alpha^4$ divergence rate of $m(\alpha)$ appears in a natural way.

\subsection{An outline of the proof and constituent results.}\label{sec-proof-outline}
 The starting point is the method developed in \cite{MV18a}, where by writing the Coulomb potential $\frac 1 {|x|}=\sqrt{\frac 2\pi} \int_0^\infty \e^{-\frac{u^2 x^2}2} \d u$ and expanding the exponential weight in \eqref{def-polaronmeas} in a power series
for any $\alpha>0$ and $T>0$, the Polaron measure 
\begin{equation}\label{Gaussian0}
\widehat\P_{\alpha,T}(\d\omega)= \int \mathbf P_{\hat\xi,\hat u}(\d\omega) \widehat\Theta_{\alpha,T}(\d\hat\xi\d\hat u)
\end{equation} was represented as a mixture of centered Gaussian measures 
$\mathbf P_{\hat\xi,\hat u}$ with variance 
\begin{equation}\label{var0}
\mathrm{Var}^{\mathbf{P}_{\hat{\xi},\hat{u}}}\Big[\frac{\omega(T)-\omega(-T)}{\sqrt{2T}}\Big]=3\sup_{f\in H_T}\bigg[2~\frac{f(T)-f(-T)}{\sqrt{2T}}-\int_{-T}^T f'(t)^2\d t-\sum_{i=1}^{n_T(\hat\xi)} u_i^2|f(t_i)-f(s_i)|^2\bigg]. 
\end{equation} 
Here, $H_T$ denotes all absolutely continuous functions on $[-T,T]$ with square integrable derivatives (see \cite[Eq. (3.3)-(3.4)]{MV18a} and Section \ref{sec-duality} for a detailed review). In \eqref{Gaussian0}, $\widehat\Theta_{\alpha,T}(\d\hat\xi\d\hat u)$ represents the law of a tilted Poisson point process taking values on the space of
$(\hat\xi,\hat u)$, with $\hat\xi=\{[s_1,t_1],\dots,[s_n,t_n]\}_{n\geq 0}$ denoting a collection for (possibly overlapping) intervals contained in $[-T,T]$ and $\hat u=(u_1,\dots, u_n)\in (0,\infty)^n$ denoting a string of positive numbers, with 
each $u_i$ being linked to the interval $[s_i,t_i]$. For any fixed $\alpha>0$ and as $T\to\infty$, 
the limit $\widehat\Theta_\alpha=\lim_{T\to\infty} \widehat\Theta_{\alpha,T}$ exists, can be identified explicitly and is stationary. Consequently, the infinite-volume limit $\widehat\P_\alpha=\lim_{T\to\infty} \widehat\P_{\alpha,T}$ 
 also admits a Gaussian representation 
 \begin{equation}\label{Gaussian1}
 \widehat\P_\alpha(\cdot)=\int \mathbf P_{\hat\xi,\hat u}(\cdot) \widehat\Theta_{\alpha}(\d\hat\xi\d\hat u)
 \end{equation} analogous to \eqref{Gaussian0}, and  
 for any $\alpha>0$, the distributions of the rescaled increments $\frac{\omega(T)-\omega(-T)}{\sqrt{2T}}$, both under $\widehat\P_{\alpha,T}$ and under $\widehat\P_\alpha$, converge for any $\alpha>0$ and as $T\to\infty$ 
 to a $3d$ centered Gaussian law $N(0,\sigma^2(\alpha))$ with variance given by the $L^1(\widehat\Theta_\alpha)$ and $\widehat\Theta_\alpha$-a.s. limit
  \begin{equation}\label{variance1}
 \sigma^2(\alpha) =\lim_{T\to\infty} 3\sup_{f\in H_T}\bigg[2~\frac{f(T)-f(-T)}{\sqrt{2T}}-\int_{-T}^T f'(t)^2\d t-\sum_{i=1}^{n_T(\hat\xi)} u_i^2|f(t_i)-f(s_i)|^2\bigg]. 
 \end{equation}
 We refer to Section \ref{sec-duality} for a more detailed review of these arguments from \cite{MV18a}. To show Theorem \ref{thm}, we will show that as $\alpha\to\infty$ the averaged supremum in \eqref{variance1} under $\wh\Theta_\alpha$ is bounded above by a constant times $\alpha^{-4}$. We will outline the proof of this bound now, along with a discussion of the auxiliary results which we will prove on the way. This task will be split now into four main steps.

{\it Step 1 (Duality):} In the identity \eqref{Gaussian1}, given any realization of the point process $(\hat\xi,\hat u)$ sampled according to $\wh\Theta_\alpha$, $\mathbf P_{\hat\xi,\hat u}$ is a centered Gaussian measure. The first step is to develop this {\it duality} between $\widehat\Theta_{\alpha}$ (resp. $\widehat\Theta_{\alpha,T}$) and $\widehat\P_{\alpha}$ 
 (resp. $\widehat\P_{\alpha,T}$) further. This was done in the follow up work \cite{MV21}, where a simple but a very useful identity was introduced (see \cite[Eq. (1.11), p.1647]{MV21})  -- namely, for any $\alpha>0$ and any interval $[-A,A]$, conditional on the Brownian increments $\{\omega(t)- \omega(s)\}_{-A\leq s < t \leq A}$ sampled according to $\wh\P_{\alpha}$, the ``quenched" law of the point process
$\{(s_i,t_i,u_i): -A \leq s_i < t_i \leq A\}$ under $\wh\Theta_{\alpha}$, denoted by $\wh\Theta_{\alpha,\omega}$, is a stationary ergodic Poisson point process with random intensity 
\begin{equation}
\Lambda(\alpha,\omega, \, \d s \d t\d u):= \alpha \sqrt{\frac 2\pi} \e^{-(t-s)} \e^{-\frac{u^2 |\omega(t) - \omega(s)|^2}{2}} \,\, \1_{-A \leq s < t \leq A} \, \d s \d t\, \d u.
\end{equation}
Consequently, the quenched distribution of {\it any} function 
 $$
 \mathrm f(\hat\xi,\hat u)= \sum_i \mathrm f(s_i,t_i,u_i)\qquad s_i, t_i \in [-A,A] 
 $$
 under $\widehat\Theta_{\alpha,\omega}$ is itself {\it Poisson distributed} with a {\it random} intensity determined by $\Lambda(\alpha,\omega)$. For instance,  
 the quenched distribution of the point process $\{(s_i,t_i,u_i): -A \leq s_i < t_i \leq A, u \geq 0\}$ is Poissonian with random intensity (using again the identity $1/|x|=\sqrt{2/\pi}\int_0^\infty \e^{-u^2|x|^2} \d u$)
 \begin{equation}\label{eq shoot up 0}
 \alpha  \frac {\e^{-(t-s)}} {|\omega(t)- \omega(s)|} \1_{-A \leq s < t \leq A} \d s \d t= \alpha^2   \frac {\e^{-(t-s)}} {\alpha |\omega(t)- \omega(s)|} \1_{-A \leq s < t \leq A} \d s \d t\,.
 \end{equation}
Likewise, for any $C, C^\prime>0$, the above conditional law of $\{(s_i,t_i): -A \leq s_i < t_i \leq A, C\alpha \leq u_i \leq C^\prime \alpha\}$ is Poissonian with random intensity 
\begin{equation}\label{eq shoot up 1}
\alpha^2 \sqrt{\frac 2\pi} \e^{-(t-s)} \bigg(\int_{C}^{C^\prime} \d u\,\, \e^{-\frac{u^2 \alpha^2 |\omega(t) - \omega(s)|^2}{2}}\bigg) \,\, \1_{-A \leq s < t \leq A} \, \d s \d t\,.
\end{equation}
That is, the quenched intensity of the point process $\{(s_i,t_i,u_i): u_i\geq 0\}$ under $\wh\Theta_{\alpha}$ shoots up to  
\begin{equation}\label{eq shoot up}
\alpha^2 \int_{-A\leq s < t \leq A} \d s \d t \e^{-(t-s)} V(\alpha|\omega(t)- \omega(s)|)
    \end{equation} 
for a suitable function $V$ defined on the {\it rescaled} increments $(\alpha |\omega(t)- \omega(s))_{-A \leq s < t \leq A}$ sampled according to $\wh\P_\alpha$. 
The same statement holds for the distribution of the point process $\{(s_i,t_i,u_i): C\alpha \leq u_i\leq C^\prime\alpha\}$ under $\wh\Theta_{\alpha}$  
with {\it rescaled} $u_i \in [C\alpha,C^\prime \alpha]$. We refer to Theorem \ref{thm duality} in Section \ref{sec-duality-2} for details.

 {\it Step 2 (Random intensities in strong coupling and Pekar process):} The next task is to determine the behavior of intensities $\Lambda(\alpha,\cdot)$ of the above form under $\widehat\P_\alpha$ in the {\it strong coupling limit }
 $\alpha\to\infty$. Indeed, in Theorem \ref{thm-strong-coupling} we show that for a large class of functions $V$ (including continuous bounded functions, and $V(x)=\frac 1{|x|}$, and $V(|x|)=|x|$ etc.), 
 \begin{equation}\label{strong-coupling}
 \begin{aligned}
 &\lim_{\alpha\to\infty}\E^{\widehat\P_\alpha}\bigg[\iint_{-A\leq s < t \leq A} \d s \d t \e^{-(t-s)} V(\alpha|\omega(t)- \omega(s)|)\bigg] \\
 &= \bigg(\int_{-A \leq s < t \leq A} \d s \d t \e^{-(t-s)}\bigg) \bigg(\int_{\R^3\times \R^3} V(|x-y|) \psi_0^2(x) \psi_0^2(y)\bigg) \in (0,\infty),
 \end{aligned}
 \end{equation}
 where $\psi_0$ denotes the centered solution of the Pekar variational formula $g_0$ (recall \eqref{Pekar-conj}) -- we refer to Theorem \ref{thm-strong-coupling} for details. Consequently, we have the convergence of the (averaged) distributions 
 $$
 \widehat\P_\alpha\big[\alpha|\omega(t) - \omega(s)|\in \cdot\big] \Longrightarrow  
  (\psi_0^2\otimes \psi_0^2)\big[|x-y| \in \cdot\big] \qquad \alpha\to\infty, \quad s, t \in [-A,A],
  $$
 see Corollary \ref{cor-tightness} for a precise statement.\footnote{We remark that, as mentioned previously, in \cite{MV18b}, 
 the distribution of the rescaled process 
 $(\alpha|\omega(\frac{t}{\alpha^2}) - \omega(\frac{s}{\alpha^2})|)_{s\in A, t\in B}$ under $\widehat\P_\alpha$ was shown to converge to the stationary version of the
 increments of the Pekar process. That is, the distribution of the processes on time scales of order $\frac 1 {\alpha^2}$ was considered. Currently, we are considering the distributions 
 of the rescaled increments  $\alpha|\omega(t) - \omega(s)|$ under $\widehat\P_\alpha$ with $s,t\in [-A,A]$, i.e.\ on time scales of order one.}
 In particular, the integral on the LHS in \eqref{eq shoot up} remain uniformly bounded away from zero under $\widehat\P_\alpha$, and as an upshot, we get that the Poisson intensity in Step 1 remains for $\alpha$ large, 
 on average under $\widehat\P_\alpha$, of order $C \alpha^2$ with an explicit constant $C \in (0,\infty)$ depending on the Pekar solution $\psi_0$. 
 
 The above weak convergence, in particular, yields tightness of the rescaled increments $\alpha|\omega(t)- \omega(s)|$ under $\wh\P_\alpha$ (see also Corollary \ref{cor-tightness}). For technical reasons (to be explained in Step 5 below), we will also strengthen this tightness to a quantified version -- namely, in Proposition \ref{cor:polaron-many-intervals} we will show that the tightness of $\alpha|\omega(t)- \omega(s)|$ under $\wh\P_\alpha$ fails with probability at most $\e^{-\tilde C(\eps)\alpha^2}\leq \e^{-\alpha}$ for $\alpha$ large.

 {\it Step 3: (Functionals of $(\hat\xi,\hat u)$ under $\widehat\Theta_\alpha$ and Pekar process):} We now apply the above duality to particular choices of $\mathrm f(\hat\xi,\hat u)$ and
 combine Step 1 and Step 2 above. Concretely, in Corollary \ref{cor:consequences-number-intervals} we will show that, with high probability and on average under $\widehat\Theta_\alpha$ and as $T\to\infty$, followed by $\alpha\to\infty$ the following statements hold:
 \begin{enumerate}
     \item Under $\widehat\Theta_\alpha$, the total number of intervals in $[-T,T]$ grows like $2\alpha^2 T \|\nabla \psi_0\|_2^2$, i.e.,  \begin{equation}
 \begin{aligned}
&\frac{n_T(\hat\xi)}{2\alpha^2 T} \simeq 2g_0= \int_{\R^3}|\nabla \psi_0(x)|^2 \d x= \int_{\R^3\times\R^3} \frac{\psi_0^2(x)\psi_0^2(y)\d x \d y}{|x-y|},\footnotemark \qquad\mbox{where}\\
&\qquad n_T(\hat\xi):=\#\big\{i: [s_i,t_i]\subset [-T,T]\big\},
\end{aligned}
\end{equation}
\item Under $\widehat\Theta_\alpha$, the lengths of these intervals remain exponentially distributed, i.e., 
\begin{equation}
 \begin{aligned}
& \frac{n_T^{\ssup a}(\hat\xi)}{2\alpha^2 T} \simeq 2g_0(1-\e^{-a}), \quad\mbox{where}\\
&\qquad\qquad n_T^{\ssup a}(\hat\xi):= \#\big\{i: [s_i,t_i]\subset [-T,T], t_i - s_i \leq a\big\},
\qquad a>0, 
\end{aligned}
\end{equation}
\item Under $\widehat\Theta_\alpha$, the size of $u_i$s attached with each interval $[s_i,t_i]$ grow like $\alpha$, i.e.,
\begin{equation}
 \begin{aligned}
& \frac {n_T^{\ssup{A,B}}(\hat\xi)}{2\alpha^2T} \simeq  \sqrt{\frac 2 \pi} \int_a^b \d z \iint_{\R^3\times\R^3} \d x \d y \psi_0^2(x)\psi_0^2(y)   \e^{-\frac{z^2|x-y|^2}2}, \quad\mbox{where}\\
&\qquad\qquad\qquad\qquad n_T^{\ssup{A,B}}(\hat\xi):= \#\big\{i: [s_i,t_i]\subset [-T,T], A\alpha \leq u_i \leq B\alpha\big\}, \qquad A, B>0.
\end{aligned}
 \end{equation}
 \end{enumerate}
 \footnotetext{Using a simple scaling argument, in Lemma \ref{lemma-Pekar-energy} it will be shown that $2g_0= \int_{\R^3}|\nabla \psi_0(x)|^2 \d x= \int_{\R^3\times\R^3} \frac{\psi_0^2(x)\psi_0^2(y)\d x \d y}{|x-y|}$.}
We refer to Corollary \ref{cor:consequences-number-intervals} for a precise version of the above statements. The first statement underlines that, under $\widehat\Theta_\alpha$, the total number of intervals in $[-T,T]$ grows like $2T\alpha^2 \int_{\R^3}|\nabla \psi_0(x)|^2 \d x $ as $T\to\infty$, followed by $\alpha\to\infty$. That is, the tilting 
 in $\widehat\Theta_\alpha$ increases the density of intervals from $\alpha$ to $\alpha^2$ when $\alpha$ becomes large. The second statement shows that the tilting in $\widehat\Theta_\alpha$ {\it does not change the distribution of lengths of intervals} -- the sizes of all intervals in $[-T,T]$ 
 remain exponential with parameter $1$. The third statement underlines that 
under $\widehat\Theta_\alpha$, the average size of $u$ is of order $\alpha$. 
In fact, again using the duality between $\widehat\Theta_\alpha$ and $\widehat\P_\alpha$, which is stationary and ergodic, and invoking the resulting ergodic theorem under $\widehat\Theta_\alpha$, we can strengthen the above 
facts to almost sure statements under $\widehat\Theta_\alpha$ using the corresponding framework of Palm measures, see Corollary \ref{lemma-xi-prime}. As we will see below, these facts will be reflected in the $\alpha^4$ divergence of the effective mass in the variational formula \eqref{variance1}. 

In Section \ref{subsec FKG} and Section \ref{subsec FKG comparison pot}
we consider an FKG inequality on point processes for a model similar to Polaron and apply the aforementioned duality between the point process and the corresponding path measure also in this context.  Incidentally, this FKG inequality also yields {\it strict} monotonicity of the effective mass $m(\alpha)$ (see Corollary \ref{cor:m-eff-monotone}) and sub-additivity of the variance (see Corollary \ref{cor:subadditivity}). Along these lines, in Remark \ref{rem:general-p} and in Section \ref{subsec:universality}, we will provide a generalization of these ideas.


{\it Step 4 (Estimating $\sigma^2(\alpha)$):} With the above recipe, we turn to the variational formula \eqref{variance1} to prove that $\sigma^2(\alpha)\leq \frac 1{\overline C\alpha^4}$. A simple but important first step, done in Lemma \ref{lemma-square-root}, is to show that the required bound follows if we can prove the following relation between linear and the quadratic parts of the supremum in \eqref{var0}:
\begin{equation}\label{lin quad}
2~\frac{f(T)-f(-T)}{\sqrt{2T}} \leq \frac 1{\sqrt{\overline C}\alpha^2} \sqrt{\int_{-T}^T f'(t)^2\d t+\sum_{i=1}^{n_T(\hat\xi)} u_i^2|f(t_i)-f(s_i)|^2}.
\end{equation}
In Section \ref{sec sufficient}, we provide a sufficient condition for the validity of the above bound, the idea for which is the following: given the increments $\omega$ (on some time interval $[0,T_0]$ with $T_0=T_0(\alpha)\gg \alpha$) we have roughly $2\alpha^2 T_0 \int_{\R^3} |\nabla \psi_0(x)|^2 \d x$ many intervals, length of which are bounded in $\alpha$ and the corresponding $u$'s are of size approximately $\alpha$ (as mentioned in Step 3 above). We call these intervals {\it standard} and restrict our attention to these intervals (for some fixed $\omega$ on some suitable set). In Proposition \ref{prop-second-mom} we will show that there is a set $E_\alpha^c$ of large probability $\wh\P_\alpha[E_\alpha^c] \geq 1- \e^{-\alpha}$ such that for any $\omega\in E_\alpha^c$ it is possible to choose $\delta \alpha^2$ {\it steps} (or {\it paths}) of standard intervals, with some $\delta>0$, such that the following hold:
\begin{enumerate}
    \item For any $L=1,\dots, \delta \alpha^2$, step $L$ consists of $T_0$-many {\it disjoint} intervals $\eta_L:=\big\{[s_1^{\ssup L}, t_1^{\ssup L}]$, \dots, $[s_{T_0}^{\ssup L}, t_{T_0}^{\ssup L}]\big\}$, whose lengths are approximately one and $u^{\ssup L}_k \approx \alpha$ for all $k=1,\dots, T_0$, 
    \item In each step $L=1,\dots, \delta \alpha^2$ we consider a {\it new} set of standard intervals as above (i.e., $\eta_\ell\cap \eta_{\ell^\prime}=\emptyset$) and 
    \item For any $x\in [0,T_0]$, if 
    $$
    U_{\delta\alpha^2}(x):= \sum_{L=1}^{\delta\alpha^2} \sum_{k=1}^{T_0} \1_{[t_k^{\ssup L}, s_{k+1}^{\ssup L}]}(x)
    $$
    is the number of times $x$ is caught in a ``vacant period" of step $L$, then (for $\omega\in E_\alpha^c$), under the quenched measure $\wh\Theta_{\alpha,\omega}$, we have 
    \[
    \sup_\alpha \frac 1 {T_0} \E^{\wh\Theta_{\alpha,\omega}}\Big[\int_0^{T_0} U_{\delta \alpha^2}(x)^2 \d x \Big] \leq \overline A< \infty.
    \]
\end{enumerate}
A key argument in Corollary \ref{cor-second-mom} using the above three conditions verifies \eqref{lin quad} on the time interval $[0,T_0]$, which, combined with the aforementioned sub-additivity of the variance implies the bound $\sigma^2(\alpha)\leq \frac 1{\overline C\alpha^4}$ and therefore Theorem \ref{thm}.

{\it Step 5:} The final step is to prove Proposition \ref{prop-second-mom}, which is shown in two further steps. As remarked in the last paragraph of Step 2, in Proposition \ref{cor:polaron-many-intervals}, we will show that, for any $\eps>0$, the $\wh\P_\alpha$-probability that $\eps$-fraction of the pairs $(s,t)$ (in some fixed interval $[0,3]$) where $\alpha|\omega(t)- \omega(s)|$ exceeds a large constant (depending on $\eps$) is at most $\e^{-\tilde C(\eps)\alpha^2}\leq \e^{-\alpha}$ for $\alpha$ large.  To do this, we compare via FKG inequality to the point process of a similar model where the only permitted intervals have $u\leq O(\alpha)$ and use duality between $\wh\Theta_\alpha$ and $\wh\P_\alpha$ again to obtain the desired bound. In particular, this implies that the local interval density (on $[0,3]$) is $ C\alpha^2$ in a quantified sense. 

Next, in Section \ref{sec long paths}, we condition on the Polaron path so the interval process is conditionally Poissonian with random intensity. Given the path, we say $x\in [0,3]$ is ``good'' if it has both intensity at least $(1-\eps)\alpha^2$ to open intervals (with $u\approx \alpha$), and the same intensity to close intervals. Using the previous paragraph, we show the complementary ``bad'' points have density at most $\eps$ with probability $1-\e^{-c\alpha}$. Using boundedness of the Hardy--Littlewood maximal operator, this implies that most points $x$ are ``very good'', in the sense that only a small fraction of $[x-r,x+r]$ is bad, \textit{uniformly in the scale} $r$. The details for this part can be found in Section \ref{sec-good-verygood}.

Finally, in Sections \ref{section-long-paths}-\ref{sec proof prop second moment} we use the above idea to construct the steps of intervals (i)-(iii) from Step 4. The idea is that waiting times to open an interval have uniformly exponential tails when started from a closing-end-point which is ``very good''. Gaps created by previous steps of interval matching are simply added to the ``bad'' set over time. This shrinks the ``very good'' set, but the two sets remain respectively small and large, which suffices to prove Proposition \ref{prop-second-mom}.

\noindent{\bf Organization of the rest of the article:}  In Section \ref{sec-strong-coupling}, we will provide the results outlined in Step 2 above. There, the necessary 
properties of the Pekar variational problem $g_0$ will be deduced in Section \ref{sec-Pekar} and Theorem \ref{thm-strong-coupling} will be proved in Section \ref{subsec-proof-thm-strong-coupling} (a technical part of its proof concerning unbounded potential $V$ has been deferred to Appendix \ref{sec unbounded V}).\footnote{Strictly speaking, Section \ref{sec-strong-coupling} is not used in the remaining sections, but as described above, it conceptually plays an important part in the arguments of the following sections, apart from identifying the correct scale and the local picture of the point process in the strong coupling limit in terms of the Pekar process.} Section \ref{sec duality main} contains the results mentioned in Step 1 and Step 3 above. Section \ref{sec est variance} contains the results mentioned in Step 4, and Sections \ref{sec high density}-\ref{sec long paths} contain the arguments outlined in Step 5. In Appendix \ref{subsec:universality}, we will provide the generalization mentioned towards the end of Step 3, and in Appendix \ref{sec-pointprocess}, we will provide some background on Palm measures and Poisson processes with random intensities and their ergodic properties. 

\section{Strong coupling limits and the Pekar process}\label{sec-strong-coupling} 
The goal of this section is to prove the following two results: 

\begin{theorem}\label{thm-strong-coupling}
Let $V:[0,\infty)\to [0,\infty)$
be any continuous and bounded function. Then for any $\theta>0$, we have 
\begin{equation}\label{main}
\lim_{\alpha\to\infty}\E^{\widehat\P_\alpha}\bigg[   \int_0^\infty\theta \e^{-\theta t}V(\alpha|\omega(t)-\omega(0)|) \d t \bigg]=   \iint\psi_{0}^2(x)\psi_{0}^2(y)V(|x-y|)\d x\d y.      
\end{equation}
Moreover, for any  symmetric, continuous and integrable function $g:(0,\infty)^2\to [0,\infty)$, 
\begin{equation}\label{main2}
\begin{aligned}
&\lim_{\alpha\to\infty}\E^{\widehat\P_{\alpha}}\bigg[ \int_0^\infty\int_0^\infty g(s,t) V(\alpha|\omega(t)-\omega(s)|)\d s\d t\bigg] \\
&=\bigg[ \int_0^\infty \int_0^\infty g(s,t)\d s\d t\bigg]\bigg[  \iint_{\R^3\times\R^3} V(|x-y|)  \psi_{0}^2(x)\psi_{0}^2(y)\d x\d y \bigg].
\end{aligned}
\end{equation}
Moreover, both \eqref{main}-\eqref{main2} hold for $V(|x|)=\frac 1 {|x|}$ in $\R^3$ and also for any continuous function $V:[0,\infty)\to [0,\infty)$ with $|V(x)| \leq C(1+|x|)$ for some $C<\infty$. 
\end{theorem}

\begin{cor}\label{cor-tightness}
Fix any $-\infty < a_i < b_i < \infty$ for $i=1,2$. For any $s \in [a_1,b_1]$ and $t \in [a_2,b_2]$, let $\mu_{\alpha}(s,t,\cdot)= \widehat\P_\alpha[\alpha|\omega(t)-\omega(s)| \in \cdot]$ be the distribution of $\alpha|\omega(t)-\omega(s)|$ under $\widehat\P_\alpha$, while $\widehat\mu_\alpha(\cdot)$ denotes its average
$$
\begin{aligned}
&\widehat\mu_\alpha(B)=\frac{1}{Z} \int_{a_1}^{b_1}\int_{a_2}^{b_2} \e^{-|s-t|} \mu_{\alpha}(s,t,B) \d s\d t 
\qquad \forall B\subset [0,\infty); \\
&Z= \int_{a_1}^{b_1}\int_{a_2}^{b_2} \d s \d t\,\, \e^{-|t-s|}.
\end{aligned}
$$
If $\widehat\mu(\psi_0,\cdot)$ denotes the distribution of $|x-y|$ under $\psi^2_{0}(x)\otimes \psi^2_{0}(y) \d x\d y$ on $[0,\infty)$, 
then $\widehat\mu_\alpha(\cdot)$ converges weakly to $\widehat\mu(\psi_0,\cdot)$ as $\alpha\to\infty$.  In particular, $\wh\mu_\alpha$ is uniformly tight, meaning 
\[
\lim_{M\to\infty}\limsup_{\alpha\to\infty} \wh\mu_\alpha((M,\infty))=0.
\]
\end{cor}

The rest of the section is devoted to the proofs of the above two results.

\subsection{Properties of the Pekar variational problem.}\label{sec-Pekar}
For the proof of Theorem \ref{thm-strong-coupling}, we will need some properties of the Pekar variational problem, which we will deduce in the next four lemmas. Recall that the supremum in 
\begin{equation}\label{g0}
g_0=\sup_{\heap{\psi\in H^1(\R^3)}{|\psi|_2=1}}\bigg[\iint_{\R^3\times \R^3} \frac{\psi^2(x)\psi^2(y)\d x \d y}{|x-y|}-\frac{1}{2}\int_{\R^3} |\nabla \psi(x)|^2 \d x \bigg]
\end{equation}
is attained at some $\psi_{0}$ which unique modulo spatial translations and can be chosen to be centered at $0$ and is a radially symmetric function \cite{L76}.
Moreover, we have 
\begin{lemma}\label{lemma-Pekar-energy}
Let $\psi_0$ be the centered radially symmetric maximizer of \eqref{g0}. Then 
\begin{equation}\label{lemma-Pekar1}
\begin{aligned}
&\iint_{\R^3\times \R^3} \frac{\psi_{0}^2( x)\psi_{0}^2( y)}{|x-y|} \d x\d y= \int_{\R^3} |\nabla \psi_0(x)|^2 \d x, \qquad\mbox{and \,\, therefore}\\
&g_0=\frac 12 \int_{\R^3} |\nabla \psi_0(x)|^2 \d x= \frac 12 \iint_{\R^3\times \R^3} \frac{\psi_{0}^2( x)\psi_{0}^2( y)}{|x-y|} \d x\d y. 
\end{aligned}
\end{equation}
\end{lemma}
\begin{proof} 
Consider the family 
$$\psi^{\ssup\lambda}(x):=\lambda^{\frac{3}{2}}\psi_{0}(\lambda x)
$$ 
Then by rescaling 
\begin{align*}
\lambda^6&\iint_{\R^3\times \R^3} \frac{\psi_{0}^2(\lambda x)\psi_{0}^2(\lambda y)}{|x-y|} \d x\d y-\frac{\lambda^5}{2}\int_{\R^3} |\nabla \psi_{0}(\lambda x)|^2 \d x\\
=\lambda &\iint_{\R^3\times \R^3} \frac{\psi_{0}^2( x)\psi_{0}^2(y)}{|x-y|} \d x\d y-\frac{\lambda^2}{2}\int_{\R^3} |\nabla \psi_0 (x)|^2 \d x
\end{align*}
has a maximum at $\lambda=1$, providing 
$$
\iint_{\R^3\times \R^3} \frac{\psi_{0}^2( x)\psi_{0}^2(y)}{|x-y|} \d x\d y=\int_{\R^3} |\nabla \psi_{0}( x)|^2 \d x.
$$
It follows that
\[
\iint_{\R^3\times \R^3} \frac{\psi_{0}^2(x)\psi_{0}^2(y)}{|x-y|} \d x\d y=2g_0, \quad\mbox{and}\quad
\int_{\R^3} |\nabla \psi_{0}( x) |^2 \d x=2 g_0.
\qedhere
\]
\end{proof}

 \begin{lemma}\label{lemma-appendix-Pekar}
 Let $V:[0,\infty)\to \R$ be a non-negative continuous function and 
\begin{equation}\label{tilde_g_eps}
 g_\eta
= \sup_{\| \psi\|_2=1}\bigg[\int_{\R^3}\int_{\R^3} \d x \d y \psi^2(x)\psi^2 (y) \bigg(\frac 1 {|x-y|}+ \eta V(|x-y|)\bigg)-\frac{1}{2} |\nabla \psi |^2\bigg].
\end{equation}
 Then 
 \begin{equation}\label{eq-lemma-appendix-Pekar}
 \lim_{\eta\to 0} \frac{ g_\eta- g_0}\eta=\iint_{\R^3\times \R^3}\psi_0^2(x) \psi_0^2(y) V(|x-y|) \d x \d y.
 \end{equation}
 \end{lemma}
 This result will follow from Lemma \ref{lemma-appendix} below.

 \begin{lemma}
 \label{lemma-appendix}
Let $F(\cdot)$ be a non-negative real-valued function on an arbitrary metric space such that $a_0:=\inf_y F(y)$ is attained at $x_0$. 
For $U_{\delta}(x_0)$ the open $\delta$-ball around $x_0$, suppose that for any $\delta>0$, 
\[
c(\delta):=\inf_{y\notin U_{\delta}(x_0)}[F(y)-a_0]>0.
\]
Let $G$ be a continuous, nonnegative function on the same space such that $G(x_0)<\infty$. If
\[
 a_\eta: =\inf_y[F(y)+\eta G(y)],
\]
  then 
\[
\lim_{\eta\to 0} a_\eta= a_0,\qquad\mbox{and}\qquad  \lim_{\eta\to 0}{a_\eta-a_0\over \eta}= G(x_0). 
\]
 \end{lemma}
 
 \begin{proof}
 Since $G\geq 0$, it holds that $a_{\eta}\geq a_0$ for each $\eta\geq 0$. Thus, \begin{equation*}
  a_0\leq a_{\eta}\leq F(x_0)+\eta G(x_0)=a_0+\eta G(x_0).
 \end{equation*}
 Letting $\eta\to 0$ and using that $G(x_0)<\infty$ leads to the first assertion. To prove the second one, the previous display implies that \begin{equation*}
  \limsup_{\eta\to 0}\frac{a_{\eta}-a_0}{\eta}\leq G(x_0).
 \end{equation*}
 To prove the converse inequality, let $\delta>0$, and note that \begin{equation*}
  \frac{a_{\eta}-a_0}{\eta}=\min\bigg\{\frac{\inf_{y\in U_{\delta}(x_0)}(F(y)-a_0+\eta G(y))}{\eta},
    \frac{\inf_{y\in U_{\delta}(x_0)^c}(F(y)-a_0+\eta G(y))}{\eta}\bigg\}.
 \end{equation*}
 Since $G\geq 0$, it holds that 
 \[
 \frac{\inf_{y\in U_{\delta}(x_0)^c}(F(y)-a_0+\eta G(y))}{\eta}\geq \frac{c(\delta)}{\eta},
 \]
 while 
 \[
 \frac{\inf_{y\in U_{\delta}(x_0)}(F(y)-a_0+\eta G(y))}{\eta}\geq \inf_{y\in U_{\delta}(x_0)}G(y).
 \]
 Since $c(\delta)>0$ for any $\delta>0$, we conclude that 
 \[
 \liminf_{\eta\to 0}\frac{a_{\eta}-a_0}{\eta}\geq \inf_{y\in U_{\delta}(x_0)}G(y).
 \]
 Letting $\delta\to 0$ and using the continuity of $G$, we conclude that 
 \[
 \lim_{\eta\to 0}{a_\eta-a_0\over \eta}= G(x_0).
 \qedhere
 \]
\end{proof}

\begin{lemma}\label{lemma-convexity}
If $\psi_0$ denotes the centered Pekar solution and $V$ is a function such that $\int_{\R^3} V(|x-y|) \psi_0^2(y) \d y$ is not identically zero, then the function $\eta\mapsto  g_\eta$ defined in \eqref{tilde_g_eps}
is strictly convex at $\eta=0$. 
\end{lemma}
\begin{proof}
When $\eta=0$, there is a unique (up to spatial translation) maximizer of $g_0$ which is the Pekar function $\psi_{0}(x)$. Let 
$$
F(\psi)= \iint_{\R^3\times\R^3} \frac{\psi^2(x)\psi^2(y)}{|x-y|}\d x\d y+\eta \iint_{\R^3\times\R^3} \psi^2(x)\psi^2(y) V(|x-y|) \d x\d y-\frac{1}{2}\int|\nabla\psi(x)|^2 \d x
$$
such that $g_\eta= \sup_{\|\psi\|_2=1}F(\psi)$. Then for $\eta\not=0$, the Euler-Lagrange equation is obtained by setting 
\[
\frac{\d}{\d\delta} F(\psi_\eta+ \delta \varphi)\bigg|_{\delta=0}=0, \qquad\varphi\in \mathcal C^\infty_c(\R^3),
\]
leading to 
$$
\begin{aligned}
&2 \iint \frac{\psi_\eta^2(x)\psi_\eta(y)\varphi(y)}{|x-y|} \d x\d y+2\eta \iint \psi_\eta^2(x)\psi_\eta(y)\varphi(y) V(|x-y|) \d x\d y \\
&\qquad\qquad -\int \langle\nabla\psi_\eta(x),\nabla\varphi(x)\rangle  \d x=0
\end{aligned}
$$
provided $\varphi\perp\psi_\eta$. Note that $\eta \mapsto g_\eta$ is convex. Now if $\eta\mapsto g_\eta$ is not strictly convex at $\eta=0$, then $\psi_\eta=\psi_{0}$ is a solution of the optimization problem $g_\eta= \sup_{\|\psi\|_2=1}F(\psi)$.\footnote{Indeed, let $a(\cdot)$ and $b(\cdot)$ be functions on a metric space $X$ and let $F(\eta)=\sup_{x\in X}[a(x) + \eta b(x)]$. Then $F(\cdot)$ is convex and if it is not strictly convex at $\eta=0$, then for some $c,d \in \R$, $F(\eta)= c + \eta d$ for $\eta$ sufficiently close to zero. Assume further that $\sup_{x\in X} a(x)$ is attained at a unique $x_0\in X$. Then $c=a(x_0)$ and $d=b(x_0)$ making $F(\eta)=a(x_0) + \eta b(x_0)$ for $\eta$ sufficiently close to zero. Hence, the supremum defining $F(\eta)$ is attained at $x_0$ and $a^\prime(x_0)+ \eta b^\prime(x_0)=0$.} But
\[
2 \iint \frac{\psi_{0}^2(x)\psi_{0}(y)\varphi(y)}{|x-y|} \d x\d y-\int \langle \nabla\psi_{0}(x),\nabla\varphi(x) \rangle \d x=0,
\]
which forces $\iint \psi_{0}^2(x)\psi_{0}(y)\varphi(y) V(|x-y|) \d x\d y=0$ whenever $\varphi\perp\psi_{0}$, leading to
$\int_{\R^3} \psi_{0}^2(y) V(|x-y|) \d y \equiv 0$, 
which is a contradiction. 
\end{proof}

\subsection{\bf Proof of Theorem \ref{thm-strong-coupling}.}\label{subsec-proof-thm-strong-coupling} 
Before proving Theorem \ref{thm-strong-coupling}, let us note down some properties of the variational problem for any fixed $\alpha>0$, 
\begin{align}
g(\alpha):=\lim_{T\to\infty}\frac 1 {2T}\log Z_{\alpha,T}&= \lim_{T\to\infty}\frac{1}{2T}\log \E^\P\bigg[\exp\bigg(\alpha\iint_{-T\le s\le t\le T}\frac{\e^{-|t-s|}}{|\omega(t)-\omega(s)|} \d t\d s\bigg)\bigg]\nonumber\\
&= \sup_{\mathbb Q}\bigg[   \E^{\mathbb Q}\bigg( \alpha\int_0^\infty \frac{\e^{-t}}{|\omega(t)-\omega(0)|} \d t \bigg) -H(\mathbb Q|\P) \bigg] \label{g-alpha},
\end{align} 
where the supremum is taken over all processes $\mathbb Q$ with stationary increments on $\R^3$ and $H(\mathbb Q|\P)$ is the specific relative entropy of $\mathbb Q$ w.r.t. the law $\P$ of the increments of three-dimensional Brownian paths. The above statement follows from a strong LDP for the empirical process of 3d-Brownian increments (\cite[Lemma 5.3]{MV18b}) and Varadhan's lemma. 
Next, as shown in \cite[Lemma 4.6]{MV18b}, for any fixed $\alpha>0$, the supremum in \eqref{g-alpha} is actually attained over the class of processes with stationary increments.\footnote{\eqref{g-alpha} was originally deduced in \cite{DV83} from a weak LDP for the empirical process for $3d$ Brownian {\it paths}, where the resulting supremum was taken over stationary processes $\mathbb Q$. However, in this case, the supremum may not be attained over this class, in contrast to processes over stationary {\it increments}, see \cite[Sec. 1.4, p. 2123]{MV18b}.} Furthermore, the infinite-volume limit $\widehat\P_\alpha=\lim_{T\to\infty}\widehat\P_{\alpha,T}$, which exists for any fixed $\alpha>0$ in total variation on finite intervals (\cite[Theorem 5.1]{MV18a}, see the text under \eqref{varT}) belongs to the set of maximizer(s) of the variational problem \eqref{g-alpha} (see \cite[Theorem 5.2]{MV18b}). 
Finally, for any $\alpha>0$, 
\begin{equation}\label{ergodic-P-alpha}
\widehat\P_\alpha=\lim_{T\to\infty}\widehat\P_{\alpha,T} \qquad\mbox{is also stationary and ergodic,}
\end{equation}
see \eqref{eq-polaron-mixture-rep} and the explanation that follows.

Let us now start with the proof of Theorem \ref{thm-strong-coupling}. We will prove \eqref{main} first assuming that $V(|\cdot|)$ is continuous and bounded on $[0,\infty)$. The remaining assertions will be subsequently deduced from this. 
As in \eqref{g-alpha}, for any $\alpha>0$, $\theta>0$ and $\eta>0$, 
\begin{align}
&g_\eta(\alpha,\theta):=\lim_{T\to\infty}\frac{1}{2T}\log \E^\P\bigg[\exp\bigg(\alpha\iint_{-T\le s\le t\le T}\frac{\e^{-|t-s|}}{|\omega(t)-\omega(s)|} \d t\d s \nonumber\\
&\qquad\qquad\qquad\qquad +\eta \alpha^2\iint_{-T\le s\le t\le T} \theta \e^{-\theta|t-s|} V(|\alpha(\omega(t)-\omega(s))|) \d t\d s\bigg)\bigg]\nonumber
\\
&= \sup_{\mathbb Q}\bigg[  \E^{\mathbb Q}\bigg(  \alpha\int_0^\infty \frac{\e^{-t}}{|\omega(t)-\omega(0)|} \d t+\eta \alpha^2\int_0^\infty \theta \e^{-\theta t} V(|\alpha(\omega(t)-\omega(0)|)) \d t \bigg)  -H(\mathbb Q|\P)\bigg].\label{g-eps-alpha}
\end{align}
The supremum defining $g_\eta(\alpha,\theta)$ is also taken over processes with stationary increments in $\R^3$. We will now handle, for any fixed $\eta>0$ and $\theta>0$, the rescaled asymptotic behavior of 
$g_\eta(\alpha,\theta)/\alpha^2$ as $\alpha\to\infty$:
\begin{align}
&g_\eta:=\lim_{\alpha\to\infty}\frac{1}{\alpha^2}g_\eta(\alpha,\theta)\nonumber\\
&=\lim_{\alpha\to\infty}\sup_{\mathbb Q}\bigg[ \E^{\mathbb Q}\bigg[\int_0^\infty\frac{\e^{-t}}{\alpha|\omega(t)-\omega(0)|}\d t+\eta \int_0^\infty \theta \e^{-\theta t}V(\alpha|\omega(t)-\omega(0)|) \d t   \bigg]  -\frac{1}{\alpha^2}H(\mathbb Q|\P)  \bigg] \nonumber\\
&=\lim_{\alpha\to\infty}\sup_{\mathbb Q}\bigg[ \E^{\mathbb Q}\bigg[ \int_0^\infty\frac{\frac 1 {\alpha^2}\e^{-\frac{t}{\alpha^2}}}{|\omega(t)-\omega(0)|} \d t+ \eta\int_0^\infty \big(\frac\theta{\alpha^2}\big) \e^{-(\frac\theta{\alpha^2}) t}V(|\omega(t)-\omega(0)|) \d t  \bigg]  -H(\mathbb Q|\P)  \bigg] \label{Br-scaling}\\
&=
\sup_{\psi:\|\psi\|_2=1}
\bigg[ \iint\limits_{\R^3\times\R^3}\frac{\psi^2(x)\psi^2(y)}{|x-y|}\d x\d y
+ 
\eta \iint\limits_{\R^3\times\R^3}
\psi^2(x)\psi^2(y)V(|x-y|)\d x\d y     -\frac{1}{2}  \int\limits_{\R^3} |\nabla\psi(x)|^2 \d x 
\bigg]. 
\label{DV-varfor}
\end{align}
In \eqref{Br-scaling}, we used the scaling property of Brownian increments, and in \eqref{DV-varfor}, the strong coupling limit of the free energy (see Remark \ref{rem-scaling} below for details). 
Also, note that for $g_\eta$, we used the notation from \eqref{tilde_g_eps}. In the above identity, we now differentiate left and right hand sides
with respect to $\eta$ at $\eta=0$, and obtain for every $\theta>0$,
\begin{align}
&\frac {\d}{\d\eta} g_\eta\bigg|_{\eta=0} =\iint\psi_{0}^2(x)\psi_{0}^2(y)V(|x-y|)\d x\d y\label{eq1}, \qquad\mbox{while}
\\
&\bigg(\frac{\d}{\d\eta}\frac{1}{\alpha^2}g_\eta(\alpha,\theta)\bigg)\bigg|_{\eta=0}=\lim_{T\to\infty} \E^{\widehat\P_{\alpha,T}}\bigg[\frac{1}{2T}\theta \iint_{-T\le s\le t\le T} \e^{-\theta |t-s|}V(\alpha|\omega(t)-\omega(s)|) \d s\d t \bigg]\label{eq2} \\
&\qquad\qquad\qquad\quad\qquad=\E^{\widehat\P_\alpha}\bigg[     \theta \int_0^\infty   \e^{-\theta t} V(\alpha|\omega(t)-\omega(0)|) \d t \bigg] .      \label{eq3}
\end{align}
In \eqref{eq1}, we used Lemma \ref{lemma-appendix-Pekar}, while in \eqref{eq2} we used the definition of $g_\eta(\alpha;\theta)$ and that of the Polaron measure 
$\widehat\P_{\alpha,T}$. Furthermore, in \eqref{eq3} we used the aforementioned convergence $\widehat\P_\alpha=\lim_{T\to\infty}\widehat\P_{\alpha,T}$ in total variation on finite intervals and the fact that $\widehat\P_\alpha$
is stationary, recall \eqref{ergodic-P-alpha}. Therefore, equating the two derivatives \eqref{eq1} and \eqref{eq3} we obtain, for any $\theta>0$, 
\begin{equation}\label{Laplace0}
\lim_{\alpha\to\infty}\E^{\widehat\P_\alpha}\bigg[   \int_0^\infty\theta \e^{-\theta t}V(\alpha|\omega(t)-\omega(0)|) \d t \bigg]=   \iint\psi_{0}^2(x)\psi_{0}^2(y)V(|x-y|)\d x\d y.      
\end{equation}
This shows \eqref{main}.  We now prove \eqref{main2}. By a standard density argument,  for any {continuous} $h\in L^1([0,\infty))$,
\begin{equation}\label{Laplace}
\begin{aligned}
&\lim_{\alpha\to\infty}\E^{\widehat\P_\alpha} \bigg[ \int_0^\infty h(t)V(\alpha|\omega(t)-\omega(0)|) \d t \bigg] \\
&= \bigg(\int_0^\infty h(t)\d t\bigg) \iint_{\R^3\times\R^3}\psi_{0}^2(x)\psi_{0}^2(y)V(|x-y|)\d x\d y.
\end{aligned}
\end{equation}
Indeed, by \eqref{Laplace0}, \eqref{Laplace} holds for functions in $\mathcal{A}:=\{f(\cdot)=\sum_{i=1}^{n}c_if_{\theta_i}(\cdot),n\in \N,c_i\in \R,\theta_i>0 \}$, where $f_\theta(t):=\e^{-\theta t}$. Observe that $\mathcal{A}$ is an algebra of continuous functions that separate points and vanishes nowhere, i.e., for each $t\geq 0$, there is some $f\in \mathcal{A}$ such that $f(t)\neq 0$. By Stone-Weierstrass theorem, $\mathcal{A}$ is dense in the set $C_0([0,\infty),\R)$ of continuous functions that vanish at infinity, with the topology of uniform convergence. In particular, $\mathcal A$ is dense in the set of smooth functions with compact support, which is furthermore dense in $L^1(\R)$. Since $V$ is assumed to be bounded at this stage, a dominated convergence argument implies \eqref{Laplace} for any { continuous} $h\in L^1(\R)$.

For any symmetric, continuous and integrable function $g(\cdot,\cdot)\in L^1((0,\infty)^2)$, let $h(u):=\int_0^\infty g(s,u+s) \d s$ and note that $2\int_0^\infty h(u) \d u=\int_0^\infty \int_0^\infty g(s,t) \d s\d t$.  For any such function $h$ and every function $k(\cdot)$, we have
\begin{equation}\label{g-h-k}
\iint_{(0,\infty)^2} g(s,t) k(t-s) \d t\d s=2\int_0^\infty h(u) k(u) \d u.
\end{equation}
Choosing $k(t-s)=\E^{\widehat\P_\alpha}[V(\alpha(\omega(t)-\omega(s)))]$, we have
\begin{align*}
&\lim_{\alpha\to\infty}\E^{\widehat\P_\alpha}\bigg[ \iint_{(0,\infty)^2} g(s,t) V(\alpha|\omega(t)-\omega(s)|) \d s\d t \bigg] \\
&\stackrel{\eqref{g-h-k}} = 2  \lim_{\alpha\to\infty}\E^{\widehat\P_\alpha}\bigg[   \int_0^\infty h(t)V(\alpha|\omega(t)-\omega(0)|) \d t \bigg] \\
&\stackrel{\eqref{Laplace}}=2\int_0^\infty h(t)\d t \iint\psi_{0}^2(x)\psi_{0}^2(y)V(|x-y|)\d x\d y \\
&=\bigg[\iint_{(0,\infty)^2}  g(s,t) \d s\d t \bigg]\iint\psi_{0}^2(x)\psi_{0}^2(y)V(|x-y|)\d x\d y,   
\end{align*}
which proves Theorem \ref{thm-strong-coupling} when $V$ is a continuous and bounded function. The proof of this result for unbounded $V$ (i.e., when $V(|x|)=\frac 1 {|x|}$ in $d=3$ or when $V(x)=|x|$) follows an approximation procedure, and we refer the reader to Section \ref{sec unbounded V}.
\begin{remark}\label{rem-scaling}
We deduced \eqref{Br-scaling} using Brownian scaling, which requires a remark. Let 
$$
\begin{aligned}
Z_{\alpha,T}(\lambda,\eta,\theta) := &\E^\P\bigg[\exp\bigg(\alpha\iint\limits_{-T\le s\le t\le T}\frac{\lambda \e^{-\lambda |t-s|}}{|\omega(t)-\omega(s)|}\d s\d t \\
&\qquad\qquad\qquad\qquad+\eta \alpha^2\iint\limits_{-T\le s\le t\le T} \theta \e^{-\theta (t-s)} V(\alpha|\omega(t)-\omega(s)|) \d t \d s\bigg ) \bigg]\,. 
\end{aligned}
$$
Then by Brownian scaling, for any $\tau>0$, $Z_{\alpha,T}(\lambda,\eta,\theta)= Z_{\alpha\sqrt\tau,\frac T\tau}(\lambda\tau,\eta,\theta\tau)$. Hence, 
$$
g_\eta(\alpha;\lambda,\theta):=\lim_{T\to\infty}\frac 1{2T}\log Z_{\alpha,T}(\lambda,\eta,\theta)= \frac 1\tau g_\eta(\alpha\sqrt\tau;\lambda\tau,\theta \tau).
$$
 In particular, by choosing $\tau=\frac 1 {\alpha^2}$ and $\lambda=1$, we have $g_\eta(\alpha;1,\theta)= \alpha^2 g_\eta(1;\frac 1{\alpha^2},\frac\theta{\alpha^2})$. But since $g_\eta(\alpha;1,\theta)= g_\eta(\alpha;\theta)$, which is defined in \eqref{g-eps-alpha}, we have $\frac{g_\eta(\alpha;\theta)}{\alpha^2}= g_\eta(1;\frac 1{\alpha^2},\frac\theta{\alpha^2})$ and $g_\eta(1;\frac 1{\alpha^2},\frac\theta{\alpha^2})$ is the supremum appearing in \eqref{Br-scaling}. This proves \eqref{Br-scaling}.
 To deduce \eqref{DV-varfor}, we used that (see \cite[Eq. (4.1)]{DV83}) for any $\eta, \theta>0$, 
 $$
 \begin{aligned}
 &\lim_{\lambda\to 0}\sup_{\mathbb Q}\bigg[ \E^{\mathbb Q}\bigg[ \int_0^\infty\frac{\lambda\e^{-\lambda t}}{|\omega(t)-\omega(0)|} \d t+ \eta\int_0^\infty (\theta \lambda) \e^{-(\theta \lambda) t}V(|\omega(t)-\omega(0)|)\d t   \bigg]  -H(\mathbb Q|\P)  \bigg] 
 \\
 &=\sup_{\psi:|\psi|_2=1}\bigg[ \iint\frac{\psi^2(x)\psi^2(y)}{|x-y|}\d x\d y+ \eta \iint\psi^2(x)\psi^2(y)V(|x-y|)\d x\d y     -\frac{1}{2}  \int |\nabla\psi|^2 \d x \bigg]. 
 \qed
 \end{aligned}
 $$
 \end{remark}

\subsection{Proof of Corollary \ref{cor-tightness}.} Fix any continuous and bounded function $V:[0,\infty)\to \R$. Recalling the definition of $\mu_\alpha(s,t,\cdot)$ and that of $\widehat\mu_\alpha(\cdot)$, we have 
\begin{align*}
\lim_{\alpha\to \infty} \int_0^\infty  V(\tau)\widehat\mu_\alpha(\d\tau)&=\lim_{\alpha\to\infty} \frac{1}{Z} \int_{a_1}^{b_1}\int_{a_2}^{b_2} \e^{-|s-t|}\int_0^\infty V(\tau) \mu_\alpha(s,t,\d\tau) \d s\d t\\
&=\lim_{\alpha\to\infty} \frac{1}{Z} \E^{\widehat\P_\alpha}\bigg[\int_{a_1}^{b_1}\int_{a_2}^{b_2} \e^{-|s-t|} V(\alpha|\omega(t)-\omega(s)|) \d s\d t\bigg]\\
&=\iint_{\R^3\times \R^3} V(|x-y|)\psi^2_{0}(x)\psi^2_{0}(y) \d x\d y
=\int V(\tau)\widehat\mu(\psi_0,\d\tau).
\end{align*}
The second equality follows from Theorem \ref{thm-strong-coupling}, and the third from the definition of $\widehat\mu(\psi_0,\cdot)$. \qed

\section{Duality between the point process $\widehat{\Theta}_{\alpha}$ and the Polaron measure $\widehat{\P}_{\alpha}$}\label{sec duality main}

\subsection{Duality between $\widehat{\Theta}_{\alpha}$ and $\widehat{\P}_{\alpha}$, part 1.}\label{sec-duality}
We recall some facts about the Gaussian representations of the Polaron measure $\widehat\P_{\alpha,T}$ and that of $\widehat\P_\alpha$ established in \cite{MV18a}.
Recall that $\Omega  = C\big((-\infty,\infty);\R^3)$ denotes the space of continuous functions $\omega$ taking values in $\R^3$ and $\mathcal F$ is the $\sigma$-algebra generated by the {\it increments} $\{\omega(t)-\omega(s)\}$. Recall that, if $\P$ denotes the law of $3$-dimensional Brownian increments on $\mathcal{F}$, then we have 
\begin{equation}\label{eq-duality-wiener}
  \mathrm{Var}^{\P}\Big[\frac{\omega(T)-\omega(-T)}{\sqrt{2T}}\Big]=3\sup_{f\in H_T}\bigg[2~\frac{f(T)-f(-T)}{\sqrt{2T}}-\int_{-T}^T f'(t)^2\d t\bigg],
\end{equation}
where \begin{equation}\label{eq-abs-cont-fun-space}
  H_T:=\Big\{f:[-T,T]\to \R: f\text{ is absolutely continuous and } f'\in L^2([-T,T])\Big\} 
\end{equation} is the Hilbert space of absolutely continuous functions with square-integrable derivatives. Indeed, $\P$ is the unique Gaussian measure such that \eqref{eq-duality-wiener} holds (see \cite[eq. (3.2)]{MV18a}).
More generally, given $T>0$ and $n\in \N$, if $\hat\xi:=\{[s_i,t_i]\}_{i=1}^n$ is a collection of intervals contained in $[-T,T]$ and $\hat u:=(u_1,\dots, u_n) \in (0,\infty)^n$, then for any
\begin{equation}\label{hat-xi-u}
(\hat{\xi},\hat{u})\in \widehat{\mathscr{Y}}_{n,T}:=\big\{(s_i,t_i,u_i):-T\leq s_i<t_i\leq T,u_i>0\big\}_{i=1}^n,
\end{equation}
 there is a unique Gaussian measure, denoted by $\mathbf{P}_{\hat{\xi},\hat{u}}$, such that 
\begin{equation}\label{eq-duality-tilted-wiener}
  \mathrm{Var}^{\mathbf{P}_{\hat{\xi},\hat{u}}}\Big[\frac{\omega(T)-\omega(-T)}{\sqrt{2T}}\Big]=3\sup_{f\in H_T}\bigg[2~\frac{f(T)-f(-T)}{\sqrt{2T}}-\int_{-T}^T f'(t)^2\d t-\sum_{i=1}^n u_i^2|f(t_i)-f(s_i)|^2\bigg], 
\end{equation}
see \cite[Eq. (3.3) and Eq. (3.4)]{MV18a}. Hence, for any probability measure $\widehat\Theta$ on $\widehat{\mathscr{Y}}_{T}:=\bigcup_{n=0}^{\infty}\widehat{\mathscr{Y}}_{n,T}$ (with the corresponding Borel $\sigma$-algebra), it holds that 
\begin{equation}\label{eq2-duality}
\begin{aligned}
  &\E^{\widehat\Theta}\bigg[\mathrm{Var}^{\mathbf{P}_{\hat{\xi},\hat{u}}}\Big[\frac{\omega(T)-\omega(-T)}{\sqrt{2T}}\Big]\bigg] \\
  &=3 \E^{\widehat\Theta} \bigg[\sup_{f\in H_T}\Big[2~\frac{f(T)-f(-T)}{\sqrt{2T}}-\int_{-T}^T f'(t)^2\d t-\sum_{i=1}^n u_i^2|f(t_i)-f(s_i)|^2\Big]\bigg].
  \end{aligned}
\end{equation}
In \cite{MV18a}, by writing the Coulomb potential as $\frac 1 {|x|}=\sqrt{\frac 2\pi} \int_0^\infty \e^{-\frac{u^2 x^2}2} \d u$ and by expanding the exponential weight in \eqref{def-polaronmeas} in a power series
\begin{equation}
\label{power}
\begin{aligned}
&\sum_{n=0}^\infty \frac{\alpha^n}{n!} \bigg[\iint_{-T\leq s \leq t\leq T} \frac{\e^{-|t-s|}\,\d t \, \d s}{|\omega(t)-\omega(s)|}\bigg]^n \\
&= \sum_{n=0}^\infty \frac {1}{n!} \prod_{i=1}^n \bigg[\bigg(\iint_{-T\leq s_i < t_i\leq T} \big(\alpha\,\e^{-(t_i-s_i)}\,\d s_i\, \d t_i\big)\bigg)\,\,  \bigg(\sqrt{\frac 2 \pi} \int_0^\infty \, \d u_i \e^{-\frac 12 u_i^2 |\omega(t_i)-\omega(s_i)|^2}\bigg)\bigg], 
\end{aligned}
\end{equation}
for any $\alpha>0$ and $T>0$ the Polaron measure was represented in \cite[Theorem 3.1]{MV18a} as a mixture 
\begin{equation}\label{Gauss-rep}
\widehat\P_{\alpha,T}(\d\omega)= \int \mathbf P_{\hat\xi,\hat u}(\d\omega) \widehat\Theta_{\alpha,T}(\d\hat\xi\d\hat u), \quad \bP_{\hat\xi,\hat u}(\d\omega)
= 
\frac 1 {\mathbf\Phi(\hat\xi,\hat u)} 
e^{-\frac 12\sum_{i=1}^{n_T(\hat\xi)} u_i^2 |\omega(t_i)- \omega(s_i)|^2}
\P(\d\omega),
\end{equation}  
of centered Gaussian measures $\mathbf P_{\hat\xi,\hat u}$, where 
$\mathbf \Phi(\hat\xi,\hat u)= \E^{\P_T}\big[\exp\{-\frac 12 \sum_{i=1}^{n_T(\xi)} u_i^2 |\omega(t_i)-\omega(s_i)|^2\}\big]$ is
 the normalizing weight of the Gaussian measure $\mathbf{P}_{\hat{\xi},\hat{u}}$ (\cite[Eq. (3.5)]{MV18a}). Indeed, in the second display in \eqref{power}, the term 
 $$
 \gamma_{\alpha,T}(\d s, \d t)= \gamma_\alpha(\d s\,\d t):=\alpha \e^{-(t-s)} \1_{-T \leq s < t \leq T} \d s\d t
 $$
 represents the intensity of a Poisson point process 
with total weight
\begin{equation}\label{cdef2}
\begin{aligned}
\alpha c(T)= \iint \gamma_{\alpha,T}(\d s \d t)= \alpha \iint_{-T\leq s < t \leq T} \e^{-(t-s)} \d s \d t &=\alpha\int_{-T}^T  (1-\e^{-(T-s)}) \d s \\
&=2\alpha T+o(T) 
\end{aligned}
\end{equation}
as $T\to\infty$. Let $\Gamma_{\alpha,T}$ be the law of this Poisson process which takes values on the space of (possibly overlapping) intervals $\hat\xi=\{[s_1,t_1],\dots,[s_n,t_n]\}_{n\geq 0}$ contained in $[-T,T]$. 
Thus, if $\hat u=(u_1,\dots, u_n)\in (0,\infty)^n$ is a string of positive numbers 
(each $u_i$ being linked to the interval $[s_i,t_i]$ and being sampled according to Lebesgue measure), then for any collection $(\hat\xi,\hat u)$, $\mathbf P_{\hat\xi,\hat u}$ is the unique centered Gaussian measure with variance 
\eqref{eq2-duality} and  the mixing measure 
\begin{equation}\label{hatQ}
\widehat{\Theta}_{\alpha,T}(\d\hat\xi\d\hat u)= \frac{\e^{\alpha c(T)}}{Z_{\alpha,T}} \bigg(\sqrt{\frac2\pi}\bigg)^{n_T(\hat\xi)} \mathbf\Phi(\hat\xi,\hat u) \Gamma_{\alpha,T}(\d\hat\xi)\d\hat u
\end{equation}
is the tilted probability measure on the space of collections $(\hat\xi,\hat u)\in\widehat{\mathscr{Y}}_{T}$.

\begin{remark}\label{remark stu}
  In the sequel, we will often abuse notation by writing sequences of intervals $(s_i,t_i)$ instead of the full triple $(s_i,t_i,u_i)\in \widehat{\mathscr{Y}}_T$. Similarly, we will write only $\hat\xi\in \wh{\mathscr Y}_T$ instead of $(\hat\xi,\hat u)\in \wh{\mathscr Y}_T$.
  \end{remark} 
 
 Returning to \eqref{eq2-duality}, \eqref{Gauss-rep} implies then that for any $\alpha>0$ and $T>0$,  
\begin{equation}\label{varT}
\begin{aligned}
  &\mathrm{Var}^{\widehat{\P}_{\alpha,T}}\Big[\frac{\omega(T)-\omega(-T)}{\sqrt{2T}}\Big]\\
  &=3\E^{\widehat{\Theta}_{\alpha,T}}\bigg[\sup_{f\in H_T}\bigg(2~\frac{f(T)-f(-T)}{\sqrt{2 T}}-\int_{-T}^T f'(t)^2\d t-\sum_{i=1}^n u_i^2|f(t_i)-f(s_i)|^2\bigg)\bigg].
  \end{aligned}
\end{equation}
Now, the collections $(\hat\xi,\hat u)\in\widehat{\mathscr{Y}}_{T}$
form an alternating sequence of {\it clusters} or {\it active periods} (constituted by overlapping intervals) and {\it dormant periods} (formed by ``gaps" left between the consecutive clusters) in $[-T,T]$, leading to a renewal structure 
for $\widehat\Theta_{\alpha,T}$. As a consequence of the ergodic theorem, $\widehat{\Theta}_{\alpha}:=\lim_{T\to\infty}\widehat{\Theta}_{\alpha,T}$ exists, can be characterized explicitly as a renewal process and $\widehat\Theta_\alpha$ can be assumed to be stationary and ergodic \cite[Theorem 5.8]{MV18a}.\footnote{In \cite[Theorem 5.8]{MV18a}, $\widehat\Theta_\alpha$ is denoted to be the law of the renewal process on $[0,\infty)$ obtained by alternating the law 
$\widehat\mu_\alpha$ of the tilted exponential distribution (defined in \cite[Eq. (5.4)]{MV18a}) on a single dormant period and the law
$\widehat\Pi_\alpha$ (defined in \cite[Eq. (5.3)]{MV18a}) of the
tilted birth-death process on a single active period. In \cite[Theorem 5.8]{MV18a}, the stationary version of $\widehat\Theta_\alpha$ is denoted by $\widehat{\mathbb Q}_\alpha$
and it is shown that the total variation $|\widehat\Theta_{\alpha,T} - \widehat{\mathbb Q}_\alpha|\to 0$ 
on any interval $[T_1,T_2]$ as $T\to\infty$ 
(in the sense that for any interval $[T_1,T_2]\subset [0,T]$ with $T_1\to\infty$ and $T- T_2\to\infty$). Currently, we will deviate slightly from this notation and continue to write $\widehat\Theta_\alpha$ also for the stationary version of $\widehat\Theta_\alpha$.}
Moreover, by \cite[Theorem 5.1]{MV18a}, the infinite-volume measure $\widehat{\P}_{\alpha}:=\lim_{T\to\infty}\widehat{\P}_{\alpha,T}$ exists in the sense that for any $A>0$, the restriction of $\widehat{\P}_{\alpha,T}$ to  the sigma algebra $\mathcal{F}_A$ generated by $\{\omega(t)-\omega(s):-A\leq s<t\leq A\}$ converges in total variation to the restriction of $\widehat{\P}_{\alpha}$ to the same $\sigma$-algebra. 
Moreover, analogous to \eqref{Gauss-rep}, the measure $\widehat{\P}_\alpha$ has the Gaussian representation 
\begin{equation}\label{eq-polaron-mixture-rep}
  \widehat{\P}_{\alpha}(\cdot)=\int_{\widehat{\mathscr{Y}}_{\infty}}\mathbf{P}_{\hat{\xi},\hat{u} }(\cdot)\widehat{\Theta}_{\alpha}(\d \hat{\xi}\d \hat{u}).
\end{equation}
Since for any $\alpha>0$, $\widehat\Theta_\alpha$ is stationary and ergodic, the above representation implies that $\widehat\P_\alpha$ is also stationary and ergodic. Moreover, \eqref{eq-polaron-mixture-rep} also implies that 
\begin{equation}\label{eq-var-P-alpha}
\begin{aligned}
  &\mathrm{Var}^{\widehat{\P}_\alpha}\Big[\frac{\omega(T)-\omega(-T)}{\sqrt{2T}}\Big] \\
  &=3\E^{\widehat{\Theta}_{\alpha}}\bigg[\sup_{f\in H_T}\bigg(2~\frac{f(T)-f(-T)}{\sqrt{2T}}-\int_{-T}^T {f^\prime}(t)^2 \d t-\sum_{-T\leq s_i<t_i\leq T} u_i^2|f(t_i)-f(s_i)|^2\bigg)\bigg].
  \end{aligned}
\end{equation}
As a consequence of the Gaussian representations of $\widehat\P_{\alpha,T}$ and $\widehat\P_\alpha$, and the renewal theorem, the rescaled distributions of $\frac{\omega(T)-\omega(-T)}{\sqrt{2T}}$, both under $\widehat\P_{\alpha,T}$ and under $\widehat\P_{\alpha}$, converge as $T\to\infty$ to a centered Gaussian law with the 
same variance $\sigma^2(\alpha)$ (\cite[Theorem 5.2]{MV18a}):
\begin{equation}\label{eq-lim-var-formula1}
  \sigma^2(\alpha)=\lim_{T\to\infty}\mathrm{Var}^{\widehat{\P}_{\alpha,T}}\Big[\frac{\omega(T)-\omega(-T)}{\sqrt{2T}}\Big]=\lim_{T\to\infty}\mathrm{Var}^{\widehat{\P}_{\alpha}}\Big[\frac{\omega(T)-\omega(-T)}{\sqrt{2T}}\Big].
\end{equation}
By \eqref{varT}-\eqref{eq-lim-var-formula1}, we obtain the following representation of the limiting variance:

\begin{lemma}\label{lemma variance formula}
  For any $\alpha>0$, the limiting variance $\sigma^2(\alpha)$ can be represented as \begin{equation*}
    \sigma^2(\alpha)=\lim_{T\to\infty}3\E^{\widehat{\Theta}_{\alpha}}\bigg[\sup_{f\in H_T}\Big[2~\frac{f(T)-f(-T)}{\sqrt{2T}}-\int_{-T}^T f'(t)^2\d t-\sum_{-T\leq s_i<t_i\leq T} u_i^2|f(t_i)-f(s_i)|^2\Big]\bigg].
  \end{equation*}
Thus, $\sigma^2(\alpha)$ is the $L^1(\widehat{\Theta}_\alpha)$-limit (as $T\to\infty$) of 
\begin{equation}\label{eq-sigma-alpha-T}
 \sigma_{\alpha,T}^2(\hat{\xi},\hat u):=3\sup_{f\in H_T}\bigg[2~\frac{f(T)-f(-T)}{\sqrt{2T}}-\int_{-T}^T f'(t)^2\d t-\sum_{-T\leq s_i<t_i\leq T} u_i^2|f(t_i)-f(s_i)|^2\bigg].
 \end{equation}
Moreover, due to the ergodic theorem used in the proof of \cite[Theorem 5.2]{MV18a}, $\sigma^2(\alpha)$ is also the $\widehat{\Theta}_\alpha$-almost sure limit of $\sigma^2_{\alpha,T}(\hat\xi,\hat u)$ as $T\to\infty$. 
\end{lemma}

\subsection{Duality between $\widehat\Theta_{\alpha}$ and $\widehat\P_{\alpha}$, part 2, and consequences.}\label{sec-duality-2}

In this subsection we prove the following duality result and deduce some important consequences. 


\begin{theorem}\label{thm duality}
Fix an interval $[-A,A]$. Then conditional on the Brownian increments $\{\omega(t)- \omega(s)\}_{-A\leq s < t \leq A}$ sampled according to $\wh\P_\alpha$, the law of the 
$\{(s_i,t_i,u_i): -A \leq s_i < t_i \leq A\}$ under $\wh\Theta_\alpha$, is a stationary ergodic Poisson point process with random intensity 
\begin{equation}\label{def Lambda}
\Lambda(\alpha,\omega, \, \d s \d t\d u):= \alpha \sqrt{\frac 2\pi} \e^{-(t-s)} \e^{-\frac{u^2 |\omega(t) - \omega(s)|^2}{2}} \,\, \1_{-A \leq s < t \leq A} \, \d s \d t\, \d u.
\end{equation}
Consequently, for any $C, C^\prime>0$, the above conditional law of $\{(s_i,t_i,u_i): -A \leq s_i < t_i \leq A, C\alpha \leq u \leq C^\prime \alpha\}$ under $\wh\Theta_\alpha$ is Poissonian with random intensity 
\begin{equation}\label{def Lambda integrated}
\begin{aligned}
\Lambda(\alpha,\omega, \, \d s \d t)&:= \alpha \sqrt{\frac 2\pi} \e^{-(t-s)} \bigg(\int_{C\alpha}^{C^\prime\alpha} \d u\,\, \e^{-\frac{u^2 |\omega(t) - \omega(s)|^2}{2}}\bigg) \,\, \1_{-A \leq s < t \leq A} \, \d s \d t \\
&= \alpha^2 \sqrt{\frac 2\pi} \e^{-(t-s)} \bigg(\int_{C}^{C^\prime} \d u\,\, \e^{-\frac{u^2 \alpha^2 |\omega(t) - \omega(s)|^2}{2}}\bigg) \,\, \1_{-A \leq s < t \leq A} \, \d s \d t 
\end{aligned}
\end{equation}
We will denote the above conditional law of $\wh\Theta_\alpha$ given $\omega$ (sampled according to $\wh\P_\alpha$) by $\wh\Theta_{\alpha,\omega}$. 
\end{theorem}
The proof of this theorem will follow from two lemmas stated below. 
\begin{lemma}\label{lemma-N}
For any $A, B \subset [-T,T]$, let
 \begin{equation}\label{def-N-Lambda}
N_{A,B}(\alpha,\hat\xi,\hat u)= \sum_{i=1}^{n_T(\hat\xi)} \1\big\{s_i\in A, \, t_i\in B, \, u_i \geq \alpha\big\}.
\end{equation} 
Then for any $\lambda>0$, 
\begin{equation}\label{eq:Laplace-transform-Theta-alpha-T}
  \E^{\widehat\Theta_{\alpha,T}}\big[\e^{-\lambda N_{A,B}(\alpha,\cdot,\cdot)}\big]= \E^{\widehat\P_{\alpha,T}}\big[\exp\big(  (\e^{-\lambda} -1) \Lambda_{A,B}(\alpha,\cdot)\big)\big], 
\end{equation}
where 
\begin{equation}\label{def-Lambda}
\begin{aligned}
&\Lambda_{A,B}(\alpha,\omega)= \alpha^2 \iint_{A\times B} \d s \d t\, \e^{-|t-s|} \frac{\Phi(\alpha|\omega(t)-\omega(s)|)}{\alpha|\omega(t)-\omega(s)|}, \quad\mbox{and}\\
&\Phi(z)=\sqrt{\frac2\pi}\int_z^\infty \e^{-\frac{u^2}2} \d u, \quad z>0.
\end{aligned}
\end{equation}
In other words, conditional on the realization of the Brownian increments $\{\omega(\cdot)- \omega(\cdot)\}$ sampled according to the Polaron measure $\widehat\P_{\alpha,T}$,
 the random variable $N_{A,B}(\alpha)$ under $\widehat\Theta_{\alpha,T}$ is Poisson-distributed with a (random) intensity 
$\alpha^2 \Lambda_{A,B}(\alpha,\cdot)$. 
Consequently, for any $\alpha>0$,  and bounded, measurable $A, B \subset \R$:
\begin{equation}\label{eq:laplace-transform-theta-alpha}
  \E^{\widehat\Theta_{\alpha}}\big[\e^{-\lambda N_{A,B}(\alpha)}\big]= \E^{\widehat\P_{\alpha}}\big[\exp\big( \alpha^2 (\e^{-\lambda}- 1) \Lambda_{A,B}(\alpha,\cdot)\big)\big]. 
\end{equation}
\end{lemma}

The above result will follow from the next.
Set $[-T,T]^2_{\leq}\equiv\{(s,t)\in [-T,T]^2:s\leq t\}$.

\begin{lemma}\label{thm-Theta-P}
Fix $\alpha,T>0$. Then for any measurable function $\mathrm f: [-T,T]^2_{\leq}\times(0,\infty)\to\R$,
\begin{equation}\label{eq1-thm-Theta-P}
\begin{aligned}
&\E^{\widehat\Theta_{\alpha,T}}\big[\e^{-\lambda \sum_{i=1}^{n_T(\hat\xi)} \mathrm f(s_i,t_i,u_i)}\big] 
\\
&= \frac{1}{Z_{\alpha,T}}\E^\P\bigg[ \exp\bigg(\alpha \iint_{-T\le s<t\le T} \e^{-|t-s|} g_\lambda(s,t,|\omega(t)-\omega(s)|) \d t\d s\bigg)\bigg],
\end{aligned}
\end{equation} 
where, for any $z>0$, we denote by 
\begin{equation}\label{def-g-lambda}
\begin{aligned}
g_\lambda(s,t,z) &: = \sqrt{\frac2\pi} \int_0^\infty \e^{-\lambda \mathrm f(s,t,u) - \frac{u^2 z^2} 2} \d u. 
\end{aligned}
\end{equation}
Moreover, for any $\lambda>0$ 
\begin{equation}\label{eq0-thm-Theta-P}
\E^{\widehat\Theta_{\alpha,T}}\big[\e^{-\lambda \sum_{i=1}^{n_T(\hat\xi)} \mathrm f(s_i,t_i,u_i)}\big]= \E^{\widehat\P_{\alpha,T}}\bigg[ \exp\bigg(\alpha \iint_{-T\le s<t\le T} \d s \d t \e^{-|t-s|} \,\, \widehat g_\lambda(s,t,|\omega(t) - \omega(s)|) \bigg)\bigg], 
\end{equation}
where, for any $z>0$, 
\begin{equation}\label{def-hat-g-lambda}
\begin{aligned}
\widehat g_\lambda(s,t,z)&= \sqrt{\frac 2\pi} \int_0^\infty \d u \, \big[\e^{-\lambda \mathrm f(s,t,u)}-1\big] \e^{- \frac{u^2 z^2}2} \\
&=  g_\lambda(s,t,z) - \sqrt{\frac 2\pi} \int_0^\infty \d u \e^{- \frac{u^2 z^2}2} = g_\lambda(s,t,z)- \frac 1z. 
\end{aligned}
\end{equation}
Finally, for any $\alpha>0$ and $A>0$, 
\begin{equation}\label{eq0.5-thm-Theta-P}
\begin{aligned}
&\E^{\widehat\Theta_{\alpha}}\big[\e^{-\lambda \sum_{i:[s_i,t_i]\subset [-A,A]} \mathrm f(s_i,t_i,u_i)}\big] \\
&= \E^{\widehat\P_{\alpha}}\bigg[ \exp\bigg(\alpha \iint\limits_{-A\le s<t\le A} \d s \d t \e^{-|t-s|} \,\, \widehat g_\lambda(s,t,|\omega(t) - \omega(s)|) \bigg)\bigg].
\end{aligned}
\end{equation}

\end{lemma}
\begin{proof}
Let us fix any $\lambda>0$. We will show first \eqref{eq1-thm-Theta-P}. Indeed, 
\begin{align*}
&\E^{\widehat{\Theta}_{\alpha,T}}\big[         \exp\big[-\lambda \sum_{i=1}^{n_T(\hat\xi)} \mathrm f(s_i,t_i,u_i )\big]   \big]\\
&=\frac{\e^{\alpha c(T)}}{Z_{\alpha,T}} \E^{\Gamma_{\alpha,T}}\bigg[ \E^\P\bigg( \int_{(0,\infty)^{n_T(\hat\xi)}} \e^{-\lambda \sum_i \mathrm f(s_i,t_i,u_i )} \bigg(\frac2\pi\bigg)^{{n_T(\hat\xi)\over 2}} \\
&\qquad\qquad\qquad\qquad\qquad  \times \e^{-\frac{1}{2}\sum_i  u_i^2 |\omega(t_i)-\omega(s_i))|^2} \d u_1 \d u_2 \cdots \d u_{n_T(\hat\xi)} \bigg)\bigg] \\
&=\frac{\e^{\alpha c(T)}}{Z_{\alpha,T}} \E^{\Gamma_{\alpha,T}} \bigg[\E^\P\bigg(\prod_{i=1}^{n_T(\hat\xi)} g_\lambda(s_i,t_i, |\omega(t_i)-\omega(s_i)|)  \bigg)\bigg]\\
&=\frac{1}{Z_{\alpha,T}} \E^\P\bigg[ \e^{\alpha c(T)}\E^{\Gamma_{\alpha,T}}\bigg(\prod_{i=1}^{n_T(\hat\xi)} g_\lambda(s_i,t_i, |\omega(t_i)-\omega(s_i)|)  \bigg)\bigg]\\
&=\frac{1}{Z_{\alpha,T}}\E^\P\bigg[ \exp\bigg(\alpha \iint_{-T\le s<t\le T} \e^{-|t-s|} g_\lambda(s,t,|\omega(t)-\omega(s)|) \d t\d s\bigg)\bigg].
\end{align*}
In the first identity above, we used the definition of $\widehat\Theta_{\alpha,T}$ from \eqref{hatQ}, in the second identity we plugged in the definition of $g_\lambda$ from \eqref{def-hat-g-lambda},
in the third identity we used Fubini's theorem and in the fourth identity we used the definition of the Poisson point process $\Gamma_{\alpha,T}$ with intensity $\alpha \e^{-(t-s)}\1_{-T\leq s < t \leq T} \d s \d t$ from 
\eqref{cdef2}. The above identity proves \eqref{eq1-thm-Theta-P}. Then by using the definition of $\widehat\P_{\alpha,T}$ and by plugging in the identity \eqref{def-hat-g-lambda}, we also obtain \eqref{eq0-thm-Theta-P}. 
The identity \eqref{eq0.5-thm-Theta-P} follows from \eqref{eq0-thm-Theta-P} if we let $T\to\infty$ on both sides, since the limits $\widehat\Theta_\alpha=\lim_{T\to\infty}\widehat\Theta_{\alpha,T}$ and $\widehat\P_\alpha=\lim_{T\to\infty}\widehat\P_{\alpha,T}$ exist and are stationary (recall \eqref{eq-polaron-mixture-rep}). 
\end{proof}

\noindent{\bf Proof of Lemma \ref{lemma-N} and Theorem \ref{thm duality}.}
By Lemma \ref{thm-Theta-P},  for any $\lambda>0$, 
\begin{align*}
&\E^{\widehat{\Theta}_{\alpha,T}}
\big[         \exp\big[-\lambda \sum_{i=1}^{n_T(\hat\xi)} \mathrm f(s_i,t_i,u_i )\big]   \big] \\
&=
\frac{1}{Z_{\alpha,T}}\E^\P\bigg[ \exp\bigg(\alpha \iint_{-T\le s<t\le T} \e^{-|t-s|} g_\lambda(s,t,|\omega(t)-\omega(s)|) \d t\d s\bigg)\bigg].
\end{align*}
Let $E=\{s\in A, \, t \in B, \, u\geq \alpha\}$ and $\mathrm f(s,t,u)=\1_E(s,t,u)$. Then 
$$
\begin{aligned}
g_\lambda(s,t,z) &= \sqrt{\frac2\pi} \int_0^\infty \e^{-\lambda \mathrm f(s,t,u) - \frac{u^2 |z|^2} 2} \d u \\
&= \sqrt{\frac2\pi} \bigg(\int_0^\infty [\e^{-\lambda} \e^{-\frac{u^2 |z|^2}2}]\1_E+ \int_0^\infty \e^{-\frac{u^2 |z|^2}2}[1-\1_E]\bigg) \\
&= \sqrt{\frac2\pi} \int_0^\infty \e^{-\frac{u^2 |z|^2}2}\d u +    (\e^{-\lambda} -1) \1\{s\in A, t \in B\} \sqrt{\frac2\pi} \int_\alpha^\infty \e^{-\frac{u^2 |z|^2}2}\d u \\
&= \frac 1 {|z|}+ \frac 1 {|z|} (\e^{-\lambda} -1) \1\{s\in A,t\in B\} \sqrt{\frac 2\pi}\int_{\alpha|z|}^\infty \e^{-\frac{u^2}2} \d u \\
&=\frac{1}{|z|} \bigg[1+ (\e^{-\lambda} -1) \1\{s\in A,t\in B\} \Phi(\alpha|z|)\bigg],
\end{aligned}
$$
with $\Phi$ defined in \eqref{def-Lambda}. Combining the previous two displays implies that
\begin{align*}
&\E^{\widehat{\Theta}_{\alpha,T}}\bigg[\exp\big(-\lambda \sum_{i=1}^{n_T(\hat\xi)} \mathrm f(s_i,t_i,u_i )\big)\bigg] \\
&= \frac{1}{Z_{\alpha,T}}\E^\P\bigg[ \exp\bigg(\alpha {} \iint_{-T\le s<t\le T} \frac{\e^{-|t-s|}}{|\omega(t)- \omega(s)|} \d s \d t 
\\
&\qquad\qquad\qquad\qquad + \alpha^2 (\e^{-\lambda}-1) \iint_{A\times B} \d s \d t \e^{-|t-s|} \frac{\Phi(\alpha |\omega(t)- \omega(s)|)}{\alpha|\omega(t)-\omega(s)|}\bigg)\bigg] \\
&= \E^{\widehat\P_{\alpha,T}}\bigg[\exp\bigg(\alpha^2 {(\e^{-\lambda}-1)}  \iint_{A\times B} \d s \d t \e^{-|t-s|} \frac{\Phi(\alpha |\omega(t)- \omega(s)|)}{\alpha|\omega(t)-\omega(s)|}\bigg)\bigg],
\end{align*}
as required. The proof of \eqref{eq:laplace-transform-theta-alpha} follows by taking the limit $T\to\infty$ from the previous part, concluding the proof of Lemma \ref{lemma-N}. Then Theorem \ref{thm duality} also follows from \eqref{eq-Laplace-fun-cox} together with \eqref{eq:laplace-transform-theta-alpha} which implies that, conditional on $\omega$ sampled according to $\wh\P_\alpha$, the law of  
$\{(s_i,t_i,u_i): -A \leq s_i < t_i \leq A\}$ is that of a Poisson point process with random intensity measure $\Lambda(\alpha,\omega,\d s \d t \d u)$. Since $\widehat{\P}_{\alpha}$ is stationary and ergodic, then the same properties are inherited by the point process (cf. Lemma \ref{lemma:stat-erg-pp}). \qed

\subsubsection{Consequences of Theorem \ref{thm duality}}
Following the proof from Theorem \ref{thm duality}, 
a number of interesting cases are made explicit in the following two corollaries:



\begin{cor}\label{lemma-xi-prime}
  The point process $\xi'=\{(s_i,t_i): 0<t_i-s_i<A\}$ is a stationary and ergodic Poisson point process with random intensity (under $\widehat{\P}_\alpha$) \begin{equation}\label{eq-Lambda-intensity-def}
 \Lambda(\alpha,\omega, \d s \d t)= \alpha^2 \e^{-(t-s)}\frac{1}{\alpha|\omega(t)-\omega(s)|} \1_{-A\leq s < t \leq A} \d s \d t.
  \end{equation}
    
  Moreover, the projections $\xi'_1:=\{s_i: (s_i,t_i)\in \xi'\}$, $\xi'_2:=\{t_i: (s_i,t_i)\in \xi'\}$ are stationary and ergodic Poisson point processes with random intensities \begin{equation}\label{eq-beta-def}
    \beta_1(\alpha,\omega,s) \d s :=\bigg(\alpha^2\int_{s-A}^{s+A}\e^{-(t-s)}\frac{1}{\alpha|\omega(t)-\omega(s)|}\d t\bigg)\d s 
  \end{equation}
  and \begin{equation}\label{eq-beta'-def}
    \beta_2(\alpha,\omega, t) \d t :=\bigg(\alpha^2\int_{t-A}^{t+A}\e^{-(t-s)}\frac{1}{\alpha|\omega(t)-\omega(s)|}\d s \bigg)\d t
  \end{equation}respectively. Finally, by Lemma \ref{lemma:ergodic-averages-pp}, it holds $\widehat{\Theta}_\alpha$-a.s.
    \begin{equation}
    \lim_{T\to\infty}\frac{\xi_1'((-T,T])}{2T}=\lim_{T\to\infty}\frac{\xi_2'((-T,T])}{2T}  
     =\alpha^2 \E^{\widehat{\P}_\alpha}\bigg[\int_1^2 \e^{-u}\frac{\d u}{\alpha|\omega(u)-\omega(0)|}\bigg],\label{eq-xi1-prime-number-of-points}
    \end{equation}
    \begin{equation} 
    \lim_{T\to\infty}\frac{1}{\xi_1'((-T,T])}
    \sum_{i=-\xi_1'((-T,0])}^{\xi_1'((0,T])-1} 
    \big(s_i-s_{i-1}\big) =
    \bigg(\alpha^2 
    \E^{\widehat{\P}_\alpha}\Big[\int_1^2 \e^{-u}\frac{\d u}{\alpha|\omega(u)-\omega(0)|}\Big]\bigg)^{-1}\label{eq-xi1-prime-sum-of-si},\\
    \end{equation}
    and for any $c>0$, 
    \begin{equation}
    \begin{aligned}
         \lim_{T\to\infty}\frac{1}{\xi_1'((-T,T])}
         \sum_{i=-\xi_1'((-T,0])}^{\xi_1'((0,T])-1}\big(s_i-s_{i-1}\big)\1\big\{s_i-s_{i-1}> c\big \} 
    =
    \frac{\widehat{\Theta}_{\alpha}(s_1-s_0>c)}{\alpha^2 \E^{\widehat{\P}_\alpha}\left[\int_1^2 \e^{-u}\frac{\d u}{\alpha|\omega(u)-\omega(0)|}\right]}\label{eq-xi1-prime-sum-of-si-conditioned}.
         \end{aligned}
    \end{equation}

\end{cor}

\begin{cor}\label{cor:consequences-number-intervals} 
  \begin{enumerate}
    \item {\normalfont (\textbf{Number of intervals grows like $\alpha^2$})} Let $n_T(\hat\xi)$ denote the number of all the intervals $\{[s_i,t_i]\}$ present in the time horizon $[-T,T]$. Then 
\begin{equation}\label{eq0-lemma-n}
\lim_{\alpha\to\infty}\frac{1}{\alpha^2}  \lim_{T\to\infty} \E^{\widehat\Theta_{\alpha, T}}\bigg[\frac{n_T(\hat\xi)}{2T}\bigg]= 2g_0=\int_{\R^3} |\nabla\psi_0(x)|^2\d x  >0.
\end{equation}
Also, for any $\eps>0$, 
\begin{equation}\label{eq0.5-lemma-n}
\lim_{\alpha\to\infty}\lim_{T\to\infty}\widehat\Theta_\alpha\bigg[(\hat\xi,\hat u)\in\widehat{\mathscr{Y}}_{T}\colon \bigg|\frac{n_T(\hat\xi)}{2T\alpha^2} - 2g_0\bigg| >\eps\bigg]=0.
\end{equation}
\item {\normalfont (\textbf{Lengths of intervals remain exponentially distributed})}  For any $a>0$, let $n_T^{\ssup a}(\hat\xi)=\#\{(\hat\xi,\hat u)\in\widehat{\mathscr{Y}}_{T}: (t_i-s_i) \leq a \}$. Then we have 
\begin{equation}\label{eq0-lemma-length}
\lim_{\alpha\to\infty}\frac1{\alpha^2}\lim_{T\to\infty}{1\over 2T}\E^{{\widehat \Theta}^{\alpha, T}}[n_T^{\ssup a}(\hat\xi)]= [1-e^{-a}](2g_0)=[1- \e^{-a}]\int_{\R^3} |\nabla \psi_0(x)|^2 \d x. 
\end{equation}
Also, for any $\eps>0$,
\begin{equation}\label{eq0.5-lemma-length}
\lim_{\alpha\to\infty}\lim_{T\to\infty}\widehat\Theta_\alpha\bigg[(\hat\xi,\hat u)\in\widehat{\mathscr{Y}}_{T}\colon \bigg|\frac{n_T^{\ssup a}(\hat\xi)}{2T\alpha^2} - 2g_0[1-\e^{-a}]\bigg| >\eps\bigg]=0.
\end{equation}
\item {\normalfont (\textbf{Size of $u$ grows like $\alpha$})} For any $A, B >0$, let $n_T^{\ssup{A,B}}(\hat\xi,\hat u)=\#\{(\hat\xi,\hat u)\in\widehat{\mathscr{Y}}_{T}:  A\alpha \leq u_i \leq B \alpha\}$. Then we have 
\begin{equation}\label{eq0-lemma-u}
\begin{aligned}
\lim_{\alpha\to\infty} \frac 1 {\alpha^2} \lim_{T\to\infty} \frac 1 {2T} \E^{\widehat{\Theta}_{\alpha,T}}[n_T^{\ssup{A,B}}]=
  \sqrt{\frac 2 \pi} \int_A^B \d z \iint\limits_{\R^3\times\R^3} \d x \d y \psi_0^2(x)\psi_0^2(y)   \e^{-\frac{z^2|x-y|^2}2}=:\widetilde g_0(A,B)
\end{aligned}
\end{equation}
Moreover, for any $\eps>0$,
\begin{equation}\label{eq0.5-lemma-u}
\lim_{\alpha\to\infty}\lim_{T\to\infty}\widehat\Theta_\alpha\bigg[(\hat\xi,\hat u)\in\widehat{\mathscr{Y}}_{T}\colon \bigg|\frac{n_T^{\ssup{A,B}}(\hat\xi,\hat u)}{2T\alpha^2} - \widetilde g_0(A,B)\bigg| >\eps\bigg]=0.
\end{equation}

  \end{enumerate}
\end{cor}

\begin{remark}\label{remark2-lemma-n}
We note that, under the base Poisson process $\Gamma_{\alpha,T}$, we have 
$\E^{\Gamma_{\alpha, T}} \big[\frac{n_T(\hat\xi)}{2T}\big]\simeq \alpha$, while under the tilted measure $\widehat\Theta_{\alpha,T}$, $\E^{\widehat\Theta_{\alpha, T}} \big[\frac{n_T(\hat\xi)}{2T}\big]\simeq 2 g_0 \alpha^2$ -- 
in other words, the tilt in $\widehat\Theta_{\alpha,T}$ increases the Poisson intensity from $\alpha$ to $\alpha^2$. In contrast, tilting in $\widehat\Theta_{\alpha,T}$ does not change the distribution of the length of the intervals, which, as under the base measure $\Gamma_{\alpha,T}$, still remains exponential with mean $1$. Moreover, the expectations in \eqref{eq0-lemma-n}, \eqref{eq0-lemma-length} and \eqref{eq0-lemma-u} could also be deduced directly from the Laplace transform in \eqref{eq:Laplace-transform-Theta-alpha-T} by taking the derivative at $\lambda=0$.
\end{remark}

\begin{proof}[{\bf Proof of Corollary \ref{cor:consequences-number-intervals}}]
The corollary follows directly from Corollary \ref{lemma-xi-prime} and Theorem \ref{thm-strong-coupling}. Alternatively, a direct proof can also be given using Lemma \ref{thm-Theta-P} and Theorem \ref{thm-strong-coupling}. Indeed, for Part (i), we can choose $f(s,t,u)\equiv -1$, so that according to \eqref{def-g-lambda}, $g_\lambda(s,t,z)=\sqrt{\frac 2\pi}\int_0^\infty \e^{-\lambda f(s,t,u) - \frac{u^2 z^2}2} \d u= \frac{\e^\lambda}{|z|}$. Then \eqref{eq1-thm-Theta-P} implies
\begin{equation}\label{eq1-lemma-n}
\E^{\widehat{\Theta}_{\alpha, T}}\big[ \e^{\lambda n_T(\hat\xi)}\big]=\frac{Z_{\alpha e^{\lambda},T}}{Z_{\alpha, T}} \quad\mbox{so that}\quad \frac {\d}{\d\lambda} \big[\e^{\lambda n_T(\hat\xi)}\big]\bigg|_{\lambda=0}= \big(\e^{\lambda n_T(\hat\xi)} n_T(\hat\xi)\big)\bigg|_{\lambda=0} = n_T(\hat\xi)
\end{equation}
Here $Z_{\beta,T}$ is the Polaron partition function with coupling parameter $\beta>0$. 
Therefore, 
$$
\begin{aligned}
\E^{\widehat\Theta_{\alpha,T}}\big[ n_T(\hat\xi)\big]= \frac{\d}{\d\lambda} \bigg(\E^{\widehat\Theta_{\alpha,T}}\big[\e^{\lambda n_T(\hat\xi)}\big]\bigg)\bigg|_{\lambda=0}
=\frac 1 {Z_{\alpha,T}} \frac {\d}{\d\lambda} \big(Z_{\alpha\e^\lambda,T}\big) \bigg|_{\lambda=0} 
= \frac{\d}{\d\lambda} \big(\log Z_{\alpha\e^\lambda,T}\big) \bigg|_{\lambda=0},
\end{aligned}
$$
We divide both sides by $2T$, pass to the limit $T\to\infty$ and use that, for any $\alpha>0$ and $\lambda>0$, $\lim_{T\to\infty} \frac 1 {2T} \log Z_{\alpha\e^\lambda,T}= g(\alpha \e^\lambda)$, to obtain
$$
\lim_{T\to\infty}\frac 1{2T}\E^{\widehat\Theta_{\alpha,T}}\big[ n_T(\hat\xi)\big]= \big(g^\prime(\alpha \e^\lambda) \alpha\e^\lambda\big)\big|_{\lambda=0}= g^\prime(\alpha) \alpha.
$$
Since, $\lim_{\alpha\to\infty} \frac{g(\alpha)}{\alpha^2}= g_0$, 
we arrive at \eqref{eq0-lemma-n}. {The second part of assertion (i) then also follows from the ergodicity of $\widehat\Theta_\alpha$.}

Similarly for part (ii), we can choose $f(s,t,u)={\1}_{(t- s)\le a}$ in Lemma \ref{thm-Theta-P} so that $g_\lambda(s, t, z)= \frac 1 {z} [\e^{-\lambda{\1}_{(t-s) \le a}}] = \frac 1 {z}[ {\e^{-\lambda}{\1}_{(t-s) \le a}+{\1}_{(t-s) \ge a}}]$. 
Proceeding exactly as above and using Theorem \ref{thm-strong-coupling}, we deduce \eqref{eq0-lemma-length}:
$$
\begin{aligned}
\lim_{T\to\infty}{1\over 2T}\E^{{\widehat \Theta}^{\alpha, T}}\big[\sum_i{\1}_{(t_i-s_i)\le a}\big] 
&=\E^{\widehat\P_\alpha }\bigg[ \alpha^2\int_0^a \frac{\e^{-t} \d t}{\alpha |\omega(t)-\omega(0)|}\d t\bigg] \\
&\stackrel{\alpha\to\infty}\simeq \alpha^2 \big(\int_0^a\e^{-t} \d t) \iint_{\R^3\times\R^3} {\psi_0^2(x)\psi_0^2(y)\over  |x-y|} \d x \d y= 2g_0(1- \e^{-a}).
\end{aligned}
$$
This shows \eqref{eq0-lemma-length} in (ii). {The second assertion of (ii) then follows again from the ergodicity of $\widehat\Theta_\alpha$.} For part (iii), we can proceed analogously by choosing $f(s,t,u)= \1_{A\alpha \leq u \leq B\alpha}$. 
\end{proof}

\subsection{FKG inequality on point processes}\label{subsec FKG}

The purpose of this subsection is to develop an FKG correlation inequality for the space $\wh{\mathscr Y}_T$, equipped with its Poisson point process measure.
This inequality and its application in Section \ref{subsec FKG comparison pot} will be used in Section \ref{sec high density} for the purpose of completing the proof of Theorem \ref{thm}. Additionally, it will also yield monotonicity and sub-additivity properties of independent interest (see Corollary \ref{cor:m-eff-monotone} and Corollary \ref{cor:subadditivity}).


First recall the notation from \eqref{hat-xi-u} and that 
$\wh{\mathscr Y}_T=\cup_{n=0}^\infty \wh{\mathscr Y}_{n,T}$ denotes the space of finite sets of weighted intervals $(\hat\xi,\hat u)$ contained within $[-T,T]$. Also recall the measure 
$\Gamma_{\alpha,T}(\d\hat\xi) \d\hat u$ appearing in \eqref{hatQ}, where $\Gamma_{\alpha,T}$ denotes the law of the Poisson point process with intensity $\gamma_\alpha(\d s \d t)=\e^{-(t-s)} \1_{-T\leq s < t \leq T} \d s \d t$. The realizations sampled by $\Gamma_{\alpha,T}(\d\hat\xi) \d\hat u$ are denoted by $(\hat\xi,\hat u)\in \wh{\mathscr Y}_T$. For simplicity and with a slight abuse of notation, we will write $\Theta_0:=\Gamma_{\alpha,T}$.

In this section, we will need the following definitions. 

\begin{definition}
\begin{enumerate}
   \item We say a function $f:\wh{\mathscr Y}_T\to\bbR_+$ is {\it coordinate-wise increasing} if $f(\hat\xi\cup \{[s,t]\})\geq f(\hat\xi)$ for all $\hat\xi\in\wh{\mathscr Y}_T$ and $s,t\in\bbR$ (recall Remark \ref{remark stu}).

\item We say that a probability measure $\Theta_1$ on $\wh{\mathscr Y}_T$ \textit{stochastically dominates} another probability measure $\Theta_2$ and write $\Theta_1\succeq\Theta_2$ if there exists a coupling $(\hat\xi_1,\hat\xi_2)\in \Pi(\Theta_1,\Theta_2)$ such that $\hat\xi_1$ contains $\hat\xi_2$ almost surely. 

\item We say $f:\wh{\mathscr Y}_T\to \bbR_+$ is \textit{log-super-modular} if 
\[
f(\hat\xi_1\vee \hat\xi_2)
f(\hat\xi_1\wedge \hat\xi_2)
\geq 
f(\hat\xi_1)f(\hat\xi_2),
\quad
\forall \hat\xi_1,\hat\xi_2\in \wh{\mathscr Y}_T.
\]
Here $\vee,\wedge$ respectively denote union and intersection.
\end{enumerate}
\end{definition}

First, we recall a characterization of stochastic domination, which is well-known in the discrete case.

\begin{lemma}[{\cite[Corollary 24.7]{bhattacharya2022special}}]
\label{lem:monotone-coupling-basic}
    Let $\Theta_1,\Theta_2$ be probability measures on $\wh{\mathscr Y}_T$.
    Then $\Theta_1\succeq \Theta_2$ if and only if  for all coordinate-wise increasing functions $f:\wh{\mathscr Y}_T\to\bbR_+$,
    \[
    \bbE^{\Theta_1}[f(\cdot)]\geq \bbE^{\Theta_2}[f(\cdot)].
    \]
\end{lemma}

We can now state the version of the FKG inequality we will use. 

\begin{prop}[{\cite[Corollary 1.2]{GK97}}]
\label{prop:FKG}
    Let $F:\wh{\mathscr Y}_T\to\bbR_{>0}$ be log-super-modular and $\Theta_0$-integrable (recall $\Theta_0:=\Gamma_{\alpha,T}$).
    Define a probability measure $\Theta_F$ on $\wh{\mathscr Y}_T$ by
    \[
    \d\Theta_F:= \frac 1 {Z_F}
    F(\cdot)
    \d\Theta_0, \qquad Z_F:=\E^{\Theta_0}[F].
    \]
    Then $\Theta_F$ is positively associated in the sense that for any coordinate-wise increasing $f,g:\wh{\mathscr Y}_T\to\bbR_{>0}$,
    \begin{equation}
    \label{eq:FKG-main}
    \bbE^{\Theta_F}[fg]\geq \bbE^{\Theta_F}[f]\bbE^{\Theta_F}[g].
    \end{equation}
\end{prop}

\begin{lemma}
\label{lem:monotone-coupling-useful}
    For $\Theta_0$ defined previously, let
$\de\Theta_1 \propto f_1(\cdot)\de\Theta_0$ and
    $\de\Theta_2\propto f_1(\cdot) f_2(\cdot)\de\Theta_0
    $
    where $f_1$ is log-super-modular and $f_2$ is increasing.
    Then $\Theta_2\succeq \Theta_1$.
\end{lemma}

\begin{proof}
    For any bounded coordinate-wise increasing function $g:\wh{\mathscr Y}_T\to\bbR_+$, Proposition~\ref{prop:FKG} implies
    \[
    \bbE^{\Theta_2}[g]
    =
    \frac{\bbE^{\Theta_1}[f_2 g]}{\bbE^{\Theta_1}[f_2]}
    \geq 
    \bbE^{\Theta_1}[g].
    \]
    Then Lemma \ref{lem:monotone-coupling-basic} implies the result.
\end{proof}

\begin{lemma}
\label{lem:GCI}
    Let $Q_1,Q_2$ be positive semi-definite finite rank quadratic forms on a separable Banach space equipped with a centered Gaussian measure $\mu$. 
    Then 
    \[
    \bbE^{\mu}[\e^{-Q_1-Q_2}]
    \geq 
    \bbE^{\mu}[\e^{-Q_1}]
    \bbE^{\mu}[\e^{-Q_2}].
    \]
\end{lemma}

\begin{proof}
    This follows by direct covariance computation or by the functional Gaussian correlation inequality. (The finite rank condition immediately reduces us to the finite-dimensional setting.)
\end{proof}

\begin{lemma}
\label{lem:W-LSM}
    The function (recall Section \ref{sec-duality})
    \begin{equation}\label{eq:weight-factor} 
    W(\hat\xi,\hat u):= \E^\P[\e^{-Q_{\hat\xi,\hat u}(\omega)}], 
\qquad Q_{\hat\xi,\hat u}(\omega):= \frac 12 \sum_{\{[s,t],u\}\in (\hat\xi,\hat u)} u^2 |\omega(t) - \omega(s)|^2,
\end{equation}
is log-super-modular. 
\end{lemma}

\begin{proof}
    For convenience, we write $\hat\xi$ for $(\hat\xi,\hat u)$. Then, given $\hat\xi,\hat\xi'\in \wh{\mathscr Y}_T$ we let
    \[
    Q_0=Q_{\hat\xi\cap \hat\xi'},
    \quad
    Q_1=Q_{\hat\xi}-Q_0 = Q_{\hat\xi\backslash \hat\xi'},
    \quad
    Q_2=Q_{\hat\xi'}-Q_0 = Q_{\hat\xi'\backslash \hat\xi}.
    \]
    Then it suffices to show that $\e^{-Q_1},\e^{-Q_2}$ are positively correlated under $\e^{-Q_0(\omega)} \P(\d\omega)$.
    This is a consequence of Lemma~\ref{lem:GCI} since the latter measure is Gaussian.
\end{proof}

\subsection{Duality and stochastic domination of $\wh\Theta_\alpha$ by FKG comparison of potentials}\label{subsec FKG comparison pot}

We also recall from Section \ref{sec-duality} that the key idea in \cite{MV18a} leading to the interval spaces introduced before was to expand $V(r)=1/|r|$ as a positive combination of Gaussian densities:
\begin{equation}
\label{eq:MV-identity}
    V(x)
    =
    \frac{1}{|x|}
    =
    \sqrt{\frac{2}{\pi}}
    \int_0^{\infty}
    \exp\lt(-\frac{u^2 x^2}{2}\rt)~\de u.
\end{equation}
This led to a coupling of the Polaron path measure to a tilted Poisson point process on intervals. 
We will use an extension of this coupling which allows general positive combinations of Gaussian densities, and argue stochastic monotonicity by applying FKG to the interval process.

Let $\gamma(\cdot)$ on $\wh{\mathscr Y}_T$ be a positive Borel measure, and $\{\gamma_{[s,t]}(\cdot)\}_{0\leq s<t\leq T}$ a choice of regular conditional distributions on $(0,\infty)$.
We will assume that the latter family is uniformly locally finite on $\bbR$, i.e.
\begin{equation}
\label{eq:locally-uniformly-finite}
    \sup_{0\leq s<t\leq T} \gamma_{[s,t]}([-A,A])<\infty,\quad\forall A\in (0,\infty).
\end{equation}
Then let
\begin{equation}
\label{eq:def-interaction-general}
    V_{[s,t],\gamma}
    :=
    \int
    \exp\bigg(-\frac{u^2 x^2}{2}\bigg)~ \gamma_{[s,t]}(\d u).
\end{equation}
We write simply $V_{\gamma}$ if all the $\gamma_{[s,t]}$ are identical. 
We assume throughout that $\gamma$ is such that the corresponding partition function 
\begin{equation}
\label{eq:partition-function}
    Z_{\gamma(\cdot),T}
    =
    \bbE^{\bbP}\bigg[
    \exp\bigg(
    \int_0^T
    \int_0^T 
    \e^{-|s-t|}
    V_{[s,t],\gamma}(|\omega_t-\omega_s|)
    ~\de s\, \de t
    \bigg)\bigg]
\end{equation}
is finite.
Hence for finiteness of $Z_{\gamma(\cdot),T}$, it suffices for the measure $\gamma(\cdot)$ to have uniformly bounded density.
Assuming this, we denote by $\wh\bbP_{\gamma(\cdot),T}$ the associated Polaron path measure.
For example, the constant density $\gamma(\cdot) \equiv\alpha \sqrt{2/\pi}$ recovers the Polaron measure $\wh\P_{\alpha,T}$.

Also, given the measure $\gamma(\cdot)$, there is a point process of weighted intervals $([s,t],u)$ for $0\leq a<b\leq T$ and $u\geq 0$. Let $\Theta_{\gamma(\cdot)}$ be law the Poisson process on $\wh{\mathscr Y}_T$ with intensity $\gamma(\cdot) \e^{-(t-s)}$.
To a collection $(\hat\xi,\hat u)=\{([s_i,t_i],u_i)\}_{i=1}^n$, using the same weight $W(\hat\xi,\hat u)$ from \eqref{eq:weight-factor}, we have, similar to \eqref{hatQ}:
\begin{equation}
\label{eq:interval-process-general}
    \wh\Theta_{\gamma(\cdot)}(\d\hat\xi\d\hat u)
    \propto
    W(\hat\xi,\hat u) \, \, \Theta_{\gamma(\cdot)}(\d\hat\xi \d \hat u).
\end{equation}

\begin{lemma}
\label{lem:MV-general-correspondence}  
    There exists a coupling of the laws of the Brownian increments $\omega$ and pairs $(\hat\xi,\hat u)$ with:
    \begin{enumerate}[label=(\alph*)]
\item \label{it:general-path-conditional-law}
        Conditional on $(\hat\xi,\hat u)$ sampled according to $\wh\Theta_{\gamma(\cdot),T}$,        the conditional distribution of the Brownian increments $\omega$ is the corresponding Gaussian measure 
        $\bP_{\hat\xi,\hat u}$.

        \item \label{it:path-interval-conditional-law}
        Conditional on $\omega$  sampled according to
    $\wh\bbP_{\gamma(\cdot),T}$,        
        the law of $(\hat\xi,\hat u)$ under
         $\wh\Theta_{\gamma(\cdot),T}$        is Poissonian with intensity $\sqrt{\frac{2}{\pi}}\exp\lt(-\frac{u^2 |\omega_t-\omega_s|^2}{2}-|t-s|\rt) \gamma(\cdot)$. 
    \end{enumerate}
\end{lemma}

\begin{proof}
    The proof is identical to \cite{MV18a} (recall Section \ref{sec-duality} and Theorem \ref{thm duality}).
\end{proof}

\begin{lemma}
\label{lem:MV-FKG}  
    If $\gamma_1(A)\leq \gamma_2(A)$ for all Borel $A\subseteq \wh{\mathscr Y}_T$, then $\wh\Theta_{\gamma_1(\cdot)}\preceq \wh\Theta_{\gamma_2(\cdot)}$ as measures on $\wh{\mathscr Y}_T$. 
\end{lemma}

\begin{proof}
    The assumption implies the Radon--Nikodym derivative $\frac{\de \Theta_{\gamma_2}}{\de \Theta_{\gamma_1}}$
    is an increasing function on $\wh{\mathscr Y}_T$. 
    Combining Lemma~\ref{lem:W-LSM} and Lemma~\ref{lem:monotone-coupling-useful} yields the result.
\end{proof}

\begin{cor}
\label{cor:compare}
    Suppose $\gamma_1(A)\leq \gamma_2(A)$ for all Borel $A\subseteq (0,\infty)$, and let $F:C([0,T];\bbR^3)\to \bbR$ be a symmetric quasi-convex function, such as $F(\omega_{[0,T]})=|\omega_T - \omega_0|^2$. Then 
    \[
    \bbE^{\wh\bbP_{\gamma_1(\cdot),T}}[F]
    \geq 
    \bbE^{\wh\bbP_{\gamma_2(\cdot),T}}[F].
    \]
\end{cor}

\begin{proof}
    By the definition of stochastic domination, it suffices to show that if $\hat\xi_1\subseteq \hat\xi_2$, then 
    \[
    \bbE^{\bP_{\hat\xi_1}}[F]
    \geq 
    \bbE^{\bP_{\hat\xi_2}}[F].
    \]
    This follows by the Gaussian correlation inequality.
\end{proof}

Recall from Lemma \ref{lemma variance formula} the formula $\sigma^2(\alpha)=\lim_{T\to\infty} \sigma^2_{\alpha,T}(\hat\xi,\hat u)$ for the CLT variance from \cite{MV18a} and that the effective mass satisfies $m(\alpha)=\frac 1 {\sigma^2(\alpha)}$. The above ideas then immediately give the monotonicity and sub-additivity we mentioned previously.

\begin{cor}
\label{cor:m-eff-monotone}
    The effective mass $m(\alpha)$ is strictly increasing in $\alpha$.
\end{cor}

\begin{proof}
    The uniform intensity $\gamma(\cdot)\equiv \alpha$ on $\wh{\mathscr Y}_T$ is monotone in $\alpha$, so $\wh\Theta_{\alpha}$ is strictly stochastically increasing by Lemma~\ref{lem:MV-FKG}.
    Corollary~\ref{cor:compare} implies the claim.
\end{proof}

\begin{cor}
\label{cor:subadditivity}
    {For any $t, \alpha>0$, let $\tilde\sigma^2_{\alpha,t}(\hat\xi,\hat u): =t \sigma^2_{\alpha,t}(\hat\xi, \hat u)$ be the variance corresponding to the non-rescaled increment.} Then, if $T=T_1+T_2$, then for any $\alpha>0$, 
    \[
    \bbE^{\wh\Theta_\alpha}[\tilde\sigma_{\alpha,T_1}^2]
    +
    \bbE^{\wh\Theta_\alpha}[\tilde\sigma_{\alpha,T_2}^2]
    \geq 
    \bbE^{\wh\Theta_\alpha}[\tilde\sigma_{\alpha,T}^2].
    \]
\end{cor}

\begin{proof}
    Let $T=T_1+T_2$.
    Since $\alpha\geq \alpha\cdot (\1_{s,t\leq T_1}+\1_{s,t\geq T_1})$ for all $s,t$, the result follows similarly to Corollary~\ref{cor:m-eff-monotone}.
\end{proof}

\section{Estimating the CLT variance $\sigma^2(\alpha)$: Proof of Theorem \ref{thm}}\label{sec est variance} 
Recall from \eqref{eq-sigma-alpha-T} that 
\begin{equation}\label{variance a.s. ergodic}
\begin{aligned}
 &\sigma_{\alpha,T}^2(\hat{\xi},\hat u)=3\sup_{f\in H_T}\bigg[2~\frac{f(T)-f(-T)}{\sqrt{2T}}-\int_{-T}^T f'(t)^2\d t-\sum_{-T\leq s_i<t_i\leq T} u_i^2|f(t_i)-f(s_i)|^2\bigg], \\
 & \mbox{and}\qquad\qquad \sigma^2(\alpha)= \lim_{T\to\infty} \sigma^2_{\alpha,T}(\cdot,\cdot) \qquad \mbox{a.s. and $L^1(\wh\Theta_\alpha)$.}
 \end{aligned}
 \end{equation}

Our goal is to show the following result, which will imply Theorem \ref{thm}: 
\begin{theorem}\label{thm-main-estimate}
There is a constant $\overline C\in (0,\infty)$ 
such that for any $\alpha>1$, 
\begin{equation}\label{eq-alpha4-bound-pos-prob}
    \limsup_{T\to\infty}\alpha^4 \sigma_{\alpha,T}^2(\hat{\xi},\hat u)\leq \frac 1 {\overline C} \qquad\qquad \widehat{\Theta}_\alpha\text{-}a.s.
  \end{equation}
  Consequently, for any $\alpha>1$, 
    \begin{equation}\label{eq-limiting-var-up-bound}
     \sigma^2(\alpha) \leq \frac 1 {\overline C\alpha^4 }, \qquad\mbox{and equivalently,}\quad
   \frac{m(\alpha)}{\alpha^4}\geq \overline C.
  \end{equation}   
  \end{theorem}
Our first step is the following observation, stated as
\begin{lemma}\label{lemma-square-root}
  If for every $f\in H_T$ it holds that 
  \begin{equation}\label{eq-eq1}
    \frac{f(T)-f(-T)}{\sqrt{2T}}\leq \frac{1}{\alpha^2 \sqrt{\overline C}}\sqrt{\int_{-T}^Tf'^2(t)\d t+\sum_{-T\leq s_i<t_i\leq T}u_i^2|f(t_i)-f(s_i)|^2} \qquad\mbox{for some $\overline C>0$,}
  \end{equation}
   then 
   $$
   \sigma_{\alpha,T}^2(\hat{\xi},\hat u)\leq \frac{1}{\overline C\alpha^4}.
   $$
\end{lemma}
\begin{proof}
Let \begin{equation*}
  Q_T(f):=Q_T(f, \hat\xi,\hat u)= \int_{-T}^Tf'^2(t)\d t+\sum_{-T\leq s_i<t_i\leq T}u_i^2|f(t_i)-f(s_i)|^2.
\end{equation*}
If \eqref{eq-eq1} holds,   
then \begin{align*}
  &\sigma_T^2(\hat{\xi},\hat u)=\sup_{f\in H_T}\bigg[2~\frac{f(T)-f(-T)}{\sqrt{2T}}-\int_{-T}^T f'(t)^2\d t-\sum_{-T\leq s_i<t_i\leq T} u_i^2|f(t_i)-f(s_i)|^2\bigg]\\
  &\leq \sup_{f\in H_T}\Big[\frac{2}{\alpha^2\sqrt{\overline C}}\sqrt{Q_T(f)}-Q_T(f)\Big]
=
    \sup_{f\in H_T}\bigg[\frac{1}{\overline C\alpha^4}-\bigg(\sqrt{Q_T(f)}-\frac{1}{\alpha^2\sqrt{\overline C}}\bigg)^2\bigg]
  \leq \frac{1}{\overline C\alpha^4}.
 \qedhere
\end{align*}
\end{proof}

\subsection{A sufficient condition for \eqref{eq-eq1} to hold.}\label{sec sufficient}
We now provide a criterion for \eqref{eq-eq1} to hold. 
Let us first describe the main idea. For technical reasons, instead of working on $[-T,T]$, we will work on the time horizon $[0,T_0]$ for some $T_0=T_0(\alpha)\gg \alpha$ sufficiently large (see Remark \ref{remark T0} below).

First, recall from Theorem \ref{thm duality} that given a realization of Brownian increments $\omega$, the conditional law $\widehat\Theta_{\alpha,\omega}$ of the point process $(\hat\xi,\hat u)$ is Poissonian with random intensity $\Lambda(\alpha,\omega)$. That is, if we fix a realization $\omega$, this determines a collection of intervals $(\hat\xi,\hat u)$. 

Next, we recall from Corollary \ref{cor:consequences-number-intervals} that in the interval $[0,T_0]$, under $\widehat\Theta_\alpha$, on average there are $(\int_{\R^3}|\nabla\psi_0(x)|^2 \d x) \alpha^2 T_0$-many intervals $[s,t]$, the lengths of all these intervals are exponentially distributed (i.e., the lengths of all these intervals are of order one, i.e., these remain uniformly bounded in $\alpha$) and the corresponding $u$ (attached to each interval $[s,t]$) satisfies $A\alpha \leq u \leq B\alpha$. We will therefore refer to these intervals as {\it standard} and restrict only to these standard intervals from our point process $(\hat\xi,\hat u)$ realization. 

Using these standard intervals, we will construct $\delta \alpha^2$ {\it runs}, or {\it paths} of intervals (for some absolute constant $\delta\in (0,1)$ to be chosen later), where each run will consist of $T_0$ many standard disjoint intervals. More concretely, the $L$-th round, or path, will be defined by a collection 
\begin{align}
\notag 
&\eta_L:= \bigg\{\big[s_{1}^{\ssup L}, t_{1}^{\ssup L}\big], \big[s_{2}^{\ssup L}, t_{2}^{\ssup L}\big], \dots, \big[s_{T_0}^{\ssup L}, t_{T_0}^{\ssup L}\big]:    \\
\label{def rounds}
&\qquad\qquad\forall k=0,\dots, T_0 \,\,\,\,  2(k-1) + \frac 13 \leq s^{\ssup L}_k \leq 2(k-1)+1\,\,\mbox{and}\\
&\qquad\qquad 2k + \frac 13 \leq t^{\ssup L}_k \leq 2k + \frac 2 3, \quad A \alpha \leq u^{\ssup L}_k\leq B\alpha\bigg\} 
\qquad\mbox{for}\qquad  L=1,\dots, \delta \alpha^2.
\end{align}
 of {\it disjoint} intervals in $[0,T_0]$, and at each step $L$, we will work with a {\it new} set of intervals from $(\hat\xi,\hat u)$, i.e., $\eta_\ell \cap \eta_{\ell^\prime}=\emptyset$ for all $\ell \ne \ell^\prime$ -- we refer to Section \ref{section-long-paths} for details



The {\it vacant period} $V_L$ of the $L$-th round is denoted by the union of the ``gaps":
\begin{equation}\label{def Vj}
V_L= \big[0,s^{\ssup L}_{1}\big) \,\bigcup \, \big(t^{\ssup L}_{1}, s^{\ssup L}_{2}\big) \, \bigcup \, \big(t^{\ssup L}_{2}, s^{\ssup L}_{3}\big) \,\bigcup \dots \bigcup \big(t^{\ssup L}_{T_0}, T_0\big], \qquad L=1,\dots, \delta \alpha^2. 
\end{equation}

We choose a point $x \in [-T, T]$ and define the random variable 
\begin{equation}\label{def-U}
U_N(x)=U_N(x,\hat\xi,\alpha)= \sum_{L=1}^N \1_{x \in V_L}
\end{equation}
That is, $U_N(x)$ is the number of times the point $x\in [0, T_0]$ is caught in a vacant period $V_L$ if we consider $N$ rounds. Note that, for any $J=1,\dots, \delta \alpha^2$, conditional on the event $\{U_N(x)= J\}$, the count $U_{N+1}(x)$ at the immediate next step either goes up by one and becomes $J+1$, or the count remains the same at $J$. In the first part of Corollary \ref{cor-second-mom} below, we will show that if 
\[
\limsup_{\alpha\to\infty} 
\frac 1{T_0}\E^{\wh\Theta_{\alpha,\omega}}
\Big[\int_0^{T_0} U_{\delta \alpha^2}(x)^2 \d x\Big] < \infty
\]
then the criterion \eqref{eq-eq1} holds. In fact, in Lemma \ref{lem:alg-succeeds}, we will prove the following.

\begin{prop}\label{prop-second-mom}
Let $T_0=T_0(\alpha)=\alpha^{10}$. Then there is an event $E_\alpha$ of Brownian increments on $[0,T_0]$ with $\wh\P_\alpha[E_\alpha] \leq \e^{-\alpha}$, such that for any $\omega \in E_\alpha^c$ and some $\delta \in (0,1)$ not depending on $\alpha$, it is possible to construct the rounds of intervals according to \eqref{def rounds} so that, 
\begin{equation}\label{eq-prop-main}
\limsup_{\alpha\to\infty} \frac 1 {T_0} \E^{\wh\Theta_{\alpha,\omega}}\bigg[\int_{0}^{T_0} U_{\delta \alpha^2}(x)^2 \d x\bigg] =:\overline A<\infty.
\end{equation}
\end{prop}

Proposition \eqref{prop-second-mom} is proved in Section \ref{sec proof prop second moment} ($\delta$ is introduced in \eqref{eq:delta-def}, with $1/\delta\ll\alpha$). Assuming this fact, we can now conclude the proof of Theorem \ref{thm}: 

\begin{cor}\label{cor-second-mom}
With $T_0=\alpha^{10}$ as above and $\alpha$ sufficiently large, we have the bound 
\[
\mathbb E^{\widehat\Theta_\alpha}[\alpha^4 \sigma_{\alpha,T_0}^2(\hat\xi,\hat u)]
\leq 
2\tilde A
+
\e^{-\alpha/2}
\]
for $\tilde A=\frac{2 \overline A}{\delta^2}+\frac{2}{A^2 \delta}$.

\end{cor}
\begin{proof}
Fix $f \in H_{T_0}$ and a collection of intervals 
as in \eqref{def rounds}.
On the event $\omega\in E_{\alpha}$ from Proposition \eqref{prop-second-mom}, we have the trivial bound $\sigma_{\alpha,T}^2(\hat\xi,\hat u)\leq 1$, so this event contributes at most $\e^{-\alpha/2}\leq \alpha^4 \e^{-\alpha}$.

Below we focus on the case $\omega\notin E_{\alpha}$.
For any step $l=1,\dots, \delta \alpha^2$, we write (recalling $u^{\ssup L}_k \geq A \alpha$)
  \begin{equation*}
  \begin{aligned}
    f(T_0)-f(0) &=\int_{V_L}f'(t)\d t + \sum_{k=1}^{T_0}(f(t^{\ssup L}_{k})-f(s^{\ssup L}_{k})) \\
    &\leq \int_{V_L}f'(t)\d t+\frac{1}{A\alpha}\sum_{k=1}^{T_0}u_{k}^{\ssup L} |f(t^{\ssup L}_{k})-f(s^{\ssup L}_{k})|.
    \end{aligned}
  \end{equation*}
The above estimate holds for every collection $\eta_L$ of intervals, for $L=1,\dots, \delta \alpha^2$, and the collections 
$\{\eta_L\}_{L=1}^{\delta\alpha^2}$ are disjoint by definition. 
Let 
\begin{equation}\label{def Astar}
A^*
=
\frac{1}{T_0}
\int_0^{T_0}
U_{\delta\alpha^2}(x)^2 \de x
\end{equation}
so that $\overline{A}=\limsup_{\alpha\to\infty} \mathbb E[A^*]$ in Proposition \eqref{prop-second-mom}.

Thus, adding the above estimates for all $L=1,\dots, \delta\alpha^2$ we obtain

    \begin{align}
        &\delta \alpha^2 (f(T_0)-f(0))
        \leq \sum_{L=1}^{\delta\alpha^2 }\int_{V_L}f'(t)\d t + \frac{1}{A\alpha}\sum_{j=1}^{\delta \alpha^2 }\sum_{k=1}^{T_0}u_{k}^{\ssup L}|f(t^{\ssup L}_{k})-f(s^{\ssup L}_{k})| \nonumber\\
    &\leq \int_{0}^{T_0}|f'(t)|\Big(\sum_{L=1}^{\delta \alpha^2 }\mathbbm{1}_{V_L}(t)\Big)\d t+\frac{1}{A\alpha}\sum_{(s_i,t_i)\in \bigcup_{L=1}^{\delta\alpha^2 } \eta_L}u_i|f(t_i)-f(s_i)| \nonumber \\ 
    &\leq \bigg(\int_{0}^{T_0}|f'(t)|^2\d t\bigg)^{\frac{1}{2}}\bigg(\int_{0}^{T_0} \Big(\sum_{L=1}^{\delta \alpha^2 }\mathbbm{1}_{V_L}(t)\Big)^2\d t \bigg)^{\frac{1}{2}}
\\
    &\quad\quad
    +\frac{\sum_{L=1}^{\delta \alpha^2}T_0}{A\alpha}\frac{1}{\sum_{j=1}^{\delta \alpha^2 }T_0}\sum_{(s_i,t_i)\in \bigcup_{L=1}^{\delta \alpha^2 } \eta_L}u_i|f(t_i)-f(s_i)| \nonumber
    \\
    &\leq \bigg(\int_{0}^{T_0}|f'(t)|^2\d t\bigg)^{\frac{1}{2}}\bigg(\int_{0}^{T_0} \Big(\sum_{L=1}^{\delta \alpha^2 }\mathbbm{1}_{V_L}(t)\Big)^2\d t \bigg)^{\frac{1}{2}}
    \\
    &\quad\quad+\frac{\sqrt{\sum_{L=1}^{\delta\alpha^2 }T_0}}{A\alpha}\sqrt{\sum_{(s_i,t_i)\in \bigcup_{L=1}^{\delta\alpha^2 } \eta_L}u_i^2|f(t_i)-f(s_i)|^2} \nonumber
    \\
&\leq \sqrt{A^* T_0} \bigg(\int_{0}^{T_0}|f'(t)|^2\d t\bigg)^{\frac{1}{2}} +
\frac{\sqrt{ T_0 \delta \alpha^2}}{A\alpha}\sqrt{\sum_{ 0 < s_i < t_i < T }u_i^2|f(t_i)-f(s_i)|^2}  \label{fifth}
 \\
 &= \sqrt{A^* T_0} \bigg(\int_{0}^{T_0}|f'(t)|^2\d t\bigg)^{\frac{1}{2}} +\frac{\sqrt{ T_0 \delta }}{A}\sqrt{\sum_{ 0 < s_i < t_i < T }u_i^2|f(t_i)-f(s_i)|^2}\,. \label{sixth}
 \end{align}
  In the third upper bound, we applied the Cauchy-Schwarz inequality to the first term, while in the fourth bound, we applied Jensen's inequality to the second term. In \eqref{fifth}, we used the definition $A^\star$ from \eqref{def Astar} for the first summand.
    Continuing with \eqref{sixth}, we now apply Lemma~\ref{lemma-square-root} with 
    \[
    (\overline C)^{-1/2}:= 
    \frac{\sqrt{A^*}}{\delta}
    +\frac 1 {A\sqrt\delta}
    \]
    to obtain (via $(x+y)^2\leq 2(x^2+y^2)$):
    \[
    \alpha^4\sigma_{\alpha,T}^2(\hat\xi,\hat u) 
    \leq 
    1/\overline{C}
    \leq 
    \frac{2 A^*}{\delta^2}
    +
    \frac{2}{A^2 \delta}.
    \]
    Since $\overline{A}=\limsup_{\alpha\to\infty} \mathbb E[A^*]$ in Proposition \eqref{prop-second-mom}, taking expectation over $(\hat\xi,\hat u)$  concludes the proof of \eqref{eq-eq1} for $T=T_0$.
\end{proof}

We can now deduce the main result of this paper.

\begin{proof}[Proof of Theorem \ref{thm}]
 By subadditivity of variance (Corollary \ref{cor:subadditivity}), it suffices to prove a variance upper bound for $T_0:=\alpha^{10}$.
 This is what we have just achieved in the preceding proof, so we are done.
\end{proof}

\begin{remark}
\label{remark T0}
 It is clear that the above argument will work for any $T_0=\alpha^n$ for $n\geq 10$ (say), because the error coming from the $\wh\P_\alpha$-probability of ``failure" to construct the $\delta\alpha^2$ rounds obeying \eqref{def rounds} (i.e., the probability of $E_\alpha$ in Proposition \ref{prop-second-mom}) is exponentially small in $\alpha$, while by definition, $U_{\delta\alpha^2}$ is at most $\delta\alpha^2$ (recall \eqref{def-U}).
 As we will be explained in Remark \ref{remark T0 T}, it might also be possible to work directly on the entire time interval $[0,T]$ instead of iterating the argument for $[0,T_0]$ (with polynomially growing $T_0=T_0(\alpha)$) $T/T_0$-many times.  
 \end{remark}

\section{High interval density with very high probability}\label{sec high density}

The main goal of this section is to prove Proposition \ref{cor:polaron-many-intervals},
which quantifies the tightness statement from Corollary \ref{cor-tightness} of the rescaled increments $\alpha|\omega(t)- \omega(s)|$ under $\wh\P_\alpha$ in an averaged sense on order one time scale.  From this result, we will deduce Corollary \ref{cor:super-standard-intensity}, which will be subsequently used in Section \ref{sec long paths} to complete the proof of Proposition \ref{prop-second-mom}. 

To derive Proposition \ref{cor:polaron-many-intervals}, we will first derive the quantitative tightness in Lemma \ref{lem:confinement-after-truncation} for a modified interaction, then use Lemma \ref{lem:MV-FKG} and duality between $\wh\Theta_\alpha$ and
$\wh\P_\alpha$ again to transfer the result to the Polaron interaction. 




\subsection{Quantified tightness under a modified interaction.}We will apply the comparison results from Section \ref{subsec FKG comparison pot} with 
\begin{equation}
\label{eq:wh-gamma-def}
\wh\gamma_{\alpha}(\d s \d t \d u) := C_1\alpha\cdot \1_{u\leq \alpha} \1_{0\leq s<t\leq 3} \,\, \d s \d t \, \d u.
\end{equation}
Here $1\ll C_1\ll\alpha$, i.e., $C_1$ will be chosen to be a sufficiently large absolute constant. 
Since the density of $\wh\gamma_{\alpha}$ does not depend on $0\leq s<t\leq 3$, we write $V_{\wh\gamma_{\alpha}}$ instead of $V_{[s,t],\wh\gamma_{\alpha}}$, and implicitly restrict attention to $[0,3]$. 
By Lemma~\ref{lem:MV-FKG}, this measure is dominated by the original point process $\wh\Theta_\alpha$ (corresponding to Polaron) up to a constant factor (the $\sqrt{\frac{2}{\pi}}$ does not play a role here since the Polaron corresponds to $\gamma(\cdot) \equiv \alpha$): \begin{equation}
\label{eq:polaron-domination-concrete-basic}
    \wh\Theta_{\wh\gamma_{\alpha}}
    \preceq 
    \wh\Theta_{C_1\alpha}.
\end{equation}

It will be useful to define the smooth function
\[
    g(y)
    =
    \int_0^1
    \e^{-v^2 y^2/2}~\de v,\quad y\geq 0, \qquad\mbox{so that} \quad g(\cdot) \leq g(0)
\]
and the explicit constant
\begin{equation}
\label{eq:A*}
    A_*
    =
    \frac{1}{2}\int_0^3 
    \int_0^3 
    \e^{-|t-s|}~\de t~\de s = \int_0^3 
    \int_0^3 
    \1_{s<t}\e^{-(t-s)}~\de t~\de s \in [4,4.1].
\end{equation}

\begin{lemma}
\label{lem:truncated-coulomb}
    For $\wh\gamma_{\alpha}$ as in \eqref{eq:wh-gamma-def}, the associated $V_{\wh\gamma_{\alpha}}(x)=V_{\wh\gamma_{\alpha}}(|x|)$ is given by
    \[
    V_{\wh\gamma_{\alpha}}(x)
    =
    C_1\alpha^2 g(\alpha x).
    \]
\end{lemma}

\begin{proof}
    Substituting $v=u/\alpha$, we have 
    \[
    V_{\wh\gamma_{\alpha}}(x)
    =
    C_1\alpha\int_0^{\alpha}
    \e^{-u^2 x^2/2}~\de u
    =
    C_1\alpha^2
    \int_0^1
    \e^{-v^2 (\alpha x)^2/2}~\de v
    \equiv 
    C_1\alpha^2 
    g(\alpha x).
    \qedhere
    \]
\end{proof}

We also let $\eps,C_0$ be such that $1\ll 1/\eps\ll C_0\ll C_1\ll \alpha$, i.e. each value is sufficiently large depending on the previous ones (we think of everything except $\alpha$ as an absolute constant, so in particular, the other values are not sent to infinity at any point). 

\begin{lemma}
\label{lem:confined-BM-LDP}
    For $\bbP$ Wiener measure and any $\eta>0$,
    let $E_{\eta}$ be the event that $\sup_{t\in [0,3]}|\omega_t|\leq \eta$. Then
    \[
    \bbP[E_{\eta}]
    \geq 
    \e^{-C_0\eta^{-2}}
    .
    \]
\end{lemma}

\begin{proof}
    The probability lower bound holds even if we additionally require that $|\omega_{j\eta^2}|\leq \eta/2$ for all integers $0\leq j\leq \eta^{-2}$.
    Indeed, splitting up time at each time in $\eta^2\bbZ$, we just need a constant probability event to happen $\eta^{-2}$ consecutive times.
    Namely if one conditions on some $\omega_{[0,j\eta^2]}$-dependent event satisfying $|\omega_{j\eta^2}|\leq \eta/2$, then with constant $\bbP$-conditional probability:
    \begin{align*}
    |\omega_{(j+1)\eta^2}|&\leq \eta/2,\qquad
    \sup_{j\eta^2 \leq t\leq (j+1)\eta^2}
    |\omega_t|
    \leq \eta.
    \end{align*}
    (In particular, this probability does not depend on $\eta$ by Brownian scaling.) This gives the bound. 
\end{proof}

Since $C_1\gg C_0$ and $g:\bbR\to\bbR$ is just a fixed function, we may assume that 
\begin{align}
\label{eq:C-gg-C0}
g(C_1^{-1/3})&\geq g(0)\lt(1-\frac{1}{2C_0}\rt),
\\
\label{eq:g(C)-half}
g(C_1)&\leq g(0)/2
.
\end{align}

Using free energy estimates, we now obtain a key confinement result.

\begin{lemma}
\label{lem:confinement-after-truncation}
    For any $\eps>0$, $0 \leq j \leq T$ and $C_1>0$, we write 
    \begin{equation}
    E_\star(\alpha)
    :=
    \lt\{\omega: 
    \frac 1 {A_\star}\int_{j \leq s < t \leq j+3} \e^{-(t-s)} \1_{\alpha |\omega_t-\omega_s|\geq C_1 }
    ~\de t~\de s
    \geq \eps 
    \rt\}.
    \end{equation}
     Then for any $\eps>0$, and $C_0>\frac{4}{\eps} $ and $C_1\geq \frac{C_0}{g(0)A_{*}}$ satisfying \eqref{eq:C-gg-C0}-\eqref{eq:g(C)-half} and for any $\alpha>0$ it holds that 
     \[
    p_\star(\alpha)
    :=
    \widehat\bbP_{\wh\gamma_{\alpha}}[E_\star(\alpha)]
    \leq 
    \e^{-\alpha^2}.
    \]
\end{lemma}

\begin{proof} 
We recall that 
$$
\widehat\P_{\widehat\gamma_\alpha}(\d\omega)= \frac 1 {Z_{\wh\gamma_\alpha}} \exp\bigg( \int_0^3 \int_0^3 \1_{s<t} \e^{-(t-s) } V_{\wh\gamma_\alpha} (|\omega_t - \omega_s|) \d s \d t\bigg) \P(\d\omega)
$$
We start with a lower bound on the partition function: 
\begin{align*}
    Z_{\wh\gamma_{\alpha}}
    &=
    \bbE^{\bbP}
    \exp\lt(
    \int_0^3
    \int_0^3 
    \1_{s<t}e^{-|s-t|}
    V_{\wh\gamma_{\alpha}}(|\omega_t-\omega_s|)
    ~\de s\, \de t
    \rt)
    \\
    &\geq 
    \bbE^{\bbP}
    \lt[
    \1_{E_{\frac{1}{2\alpha C_1^{1/3}}}}
    \exp\lt(
    \int_0^3
    \int_0^3 
    \1_{s<t} e^{-|s-t|}
    V_{\wh\gamma_{\alpha}}(|\omega_t-\omega_s|)
    ~\de s\, \de t
    \rt)
    \rt]
    \end{align*}
    Using the definition of event $E_\eta$ for 
    $\eta={\frac 1{2\alpha C_1^{1/3}}}$ (recall Lemma \ref{lem:confined-BM-LDP})
    we estimate that, on this event, 
    \[
    V_{\wh\gamma_\alpha}(|\omega_t - \omega_s|) = C_1 \alpha^2 \int_0^1 \e^{- \frac{v^2 (\alpha|\omega(t-s) - \omega(0)|)^2}{2}} \d v = C_1 \alpha^2 g(\alpha |\omega_t - \omega_s|) \geq C\alpha^2 g(0) \big(1 - \frac 1 {2C_0}\big).
    \]
    In the last lower bound, we have used \eqref{eq:C-gg-C0}. Thus, applying Lemma \ref{lem:confined-BM-LDP}, we have the lower bound
\[
    \begin{aligned}
Z_{\wh\gamma_{\alpha}}
&\geq \exp\bigg(A_\star C_1 \alpha^2 g(0) \big(1- \frac 1 {2C_0}\big)\bigg) \,\, \P(E_{\frac 1{2\alpha C_1^{1/3}}}) 
\\
&\geq \exp\bigg(A_\star C_1 \alpha^2 g(0) \big(1- \frac 1 {2C_0}\big)\bigg) \e^{-C_0 4\alpha^2 C_1^{2/3}} 
\\
&= \exp\bigg(A_\star C_1 \alpha^2 g(0) - \frac {A_\star C_1 \alpha^2 g(0)}{C_0} + \frac{A_\star C_1 \alpha^2 g(0)}{2C_0} - 4\alpha^2 C_0 C_1^{2/3}\bigg) 
\\
&\geq 
    \exp\lt(
    A_* C_1 g(0) \alpha^2\lt(1-\frac{1}{C_0}\rt)
    \rt).
\end{aligned},
\]
where in the last lower bound, we used that, since $C_1 \gg C_0$, we have $C_1^{\frac 13} \gg \frac{8C_0^2}{A_\star g(0)}$. 
    It then follows from the definition of $\wh\P_{\widehat\gamma_\alpha}$ and the above lower bound on $Z_{\wh\gamma_\alpha}$ that 
    \begin{equation}
    \label{eq:Markov-confinement}
    \begin{aligned}
    &\wh\bbP_{\wh\gamma_{\alpha}}
    \lt[
    \int_{0\leq s < t \leq 3}e^{-(t-s)}
    V_{\wh\gamma_{\alpha}}(|\omega_t-\omega_s|)
    ~\de s\, \de t
    \leq 
     A_* C_1 g(0) \alpha^2\lt(1-\frac{2}{C_0}\rt)
    \rt] \\
    &\leq  \frac 1 {Z_{\wh\gamma_\alpha}} \,\, \exp\big(A_\star C_1 \alpha^2 g(0)\big) \exp\big( - A_\star C_1 \alpha^2 g(0) \frac 2 {C_0}\big) \\
    &\leq \exp\big(A_\star C_1 \alpha^2 g(0)\big) \exp\bigg( - A_\star C_1 \alpha^2 g(0) \frac 2 {C_0}\bigg) \,\, \exp\big( - A_\star C_1 \alpha^2 g(0)\big) \exp \bigg(\frac{A_\star C_1 \alpha^2 g(0)}{C_0}\bigg)\\
&\leq \e^{-A_* C_1 g(0)\alpha^2/C_0}
\leq  \e^{-\alpha^2}
    \end{aligned}
    \end{equation}
    if $C_1\geq \frac{C_0}{g(0)A_{*}}$.
    On the other hand, from \eqref{eq:g(C)-half} and unimodality of $g$, we have the implication 
    \[
    |\omega_t-\omega_s|\geq C_1/\alpha\implies 
    V_{\wh\gamma_{\alpha}}(\omega_t-\omega_s)\leq 
    C
    g(C_1)\alpha^2\leq 
    C_1g(0)\alpha^2/2
    \]
    for all $0\leq s<t\leq 3$. Using this pointwise bound for all $0 \leq s < t \leq 3$, we have 
    \begin{align}
    \notag
    &\int_{0\leq s < t \leq 3} \e^{-(t-s)} \1_{\alpha|\omega_t - \omega_s|>C_1} \d s \d t 
    \\
    &\leq \frac {2}{C_1 \alpha^2 g(0)} \int_{0 \leq s < t \leq 3} \e^{-(t-s)} \big[ C_1 \alpha^2 g(0) - V_{\wh\gamma_\alpha}(\omega_t - \omega_s)\big] \1_{\alpha |\omega_t - \omega_s| >C_1} \d s \d t 
    \\
    \label{uni-g}
    &\leq \frac {2}{C_1 \alpha^2 g(0)} \int_{0 \leq s < t \leq 3} \e^{-(t-s)} \big[ C_1 \alpha^2 g(0) - V_{\wh\gamma_\alpha}(\omega_t - \omega_s)\big] \d s \d t
    \end{align}
    where in the last inequality, we again used the unimodality of $g(\cdot)$ (which enforces that $V_{\wh\gamma_\alpha}(x)= C_1 \alpha^2 g(\alpha x) \geq C_1 \alpha^2 g(0)$, implying the non-negativity of the term inside the bracket in the integral above). 
    Using \eqref{uni-g}, we thus have 
    \begin{equation}\label{uni-g2}
    \begin{aligned}
    \wh\P_{\wh\gamma_\alpha}(E_\star(\alpha))   
    &=\wh\P_{\wh\gamma_\alpha}\bigg[
    \frac 1 {A_\star}\int_{j \leq s < t \leq j+3} \e^{-(t-s)} \1_{\alpha |\omega_t-\omega_s|\geq C_1 }
    ~\de t~\de s
    \geq \eps \bigg] 
    \\
    &\leq \wh\P_{\wh\gamma_\alpha}\bigg[\frac 2 {A_\star C_1 \alpha^2 g(0)} \int_{0\leq s < t \leq 3} \e^{-(t-s)} \big[C_1 \alpha^2 g(0) - V_{\wh\gamma_\alpha}(\omega_t - \omega_s)\big] \d s \d t \geq \eps\bigg] \\
    &= \wh\P_{\wh\gamma_\alpha}\bigg[\int_{0\leq s < t \leq 3} \e^{-(t-s)}  V_{\wh\gamma_\alpha}(\omega_t - \omega_s) \d s \d t \leq (1- \frac\eps 2 ) A_\star C_1 \alpha^2 g(0)\bigg].    
    \end{aligned}
    \end{equation}
    Now, given $\eps>0$ we choose $C_0=C_0(\eps) > \frac 4 \eps$ so that $1- \frac \eps 2 < 1- \frac 2 {C_0}$ so that combining \eqref{uni-g2} and \eqref{eq:Markov-confinement} we have that $\wh\P_{\wh\gamma_\alpha}(E_\star(\alpha)) \leq \e^{-\alpha^2}$, as desired.   
\end{proof}

\subsection{Quantified tightness under $\wh\P_\alpha$.} 
Next, we transfer this result to the original Polaron measure. A bit of care is needed here since we cannot directly compare the integrands in, e.g., \eqref{eq:most-points-close}, but have to compare the associated interval processes instead.
The idea of the proof below is to define a set of \emph{standard} intervals in \eqref{eq:count-interval-to-compare} below, the number $n_{\std}$ of which serves as a proxy for the scale of fluctuations of $\omega$. Stochastic domination of point processes lets us directly compare the number of standard intervals, allowing us to transfer, using duality again, the preceding result to $\wh\bbP_{\alpha}$ at the cost of some constant factors.

\begin{prop}
\label{cor:polaron-many-intervals}
    For any $\eps>0$ and integer $0\leq j\leq T$ and $\tilde C>0$, let
    \begin{equation}
    \label{eq:most-points-close}
    \begin{aligned}
    &E_{\star\star}(\alpha):= \bigg\{\omega\colon \frac 1 {A_\star} \int_{j \leq s < t \leq j+3} \d s \d t
    \e^{-(t-s)} \1_{\alpha |\omega_t-\omega_s|\geq {\tilde{C}}} 
    \geq \eps \bigg\}.
   \end{aligned} 
    \end{equation}
    Then for any $0<\eps<1$, there are constants $\tilde{C}=\tilde{C}(\eps)$ and $\tilde{C}_1=\tilde{C}_1(\eps)$ such that  
    for $\alpha>0$ large enough, 
    \[
    p_{\star\star}(\alpha) :=
    \wh\bbP_{\alpha}[E_{\star\star}(\alpha)]    
    \leq \e^{-\tilde{C}_1(\eps) \alpha^2}, \quad\mbox{thus in particular}\quad p_{\star\star}(\alpha)\leq \e^{-2\alpha}.
    \]
\end{prop}


\begin{proof}
We fix $\eps>0$. We choose $\tilde{C}=\frac{C_1^3}{\eps}$, where $C_1=C_1(\eps^2)$ is defined in Lemma \ref{lem:confinement-after-truncation}. By the comparison \eqref{eq:polaron-domination-concrete-basic}, we will analyse $E_{\star\star}(C_1\alpha)$ under the measure $\wh\bbP_{C_1\alpha}$.

Recall from Theorem \ref{thm duality} (and also Lemma~\ref{lem:MV-general-correspondence}\ref{it:path-interval-conditional-law}) that 
$\widehat\Theta_{\alpha, \omega}$ (resp. $\widehat\Theta_{\wh\gamma_\alpha,\omega}$) denotes the law of $(\hat\xi,\hat u)$ conditional on $\omega|_{[0,T]}$ sampled according to $\wh\P_\alpha$ (resp. $\wh\P_{\wh\gamma_\alpha}$) -- this law is Poissonian with $\omega$-dependent intensity. Now suppose that the event $E_{\star\star}(C_1\alpha)$ holds for $\omega|_{[0,T]}$. We recall Corollary \ref{cor:consequences-number-intervals} and set the \textit{standard} intervals $([s,t],u)\in (\hat\xi,\hat u)$ defined to satisfy:
    \begin{equation}
    \label{eq:count-interval-to-compare}
    j\leq s<t\leq j+3,\quad\quad
    u\in [\alpha/C_1^2, 2\alpha/C_1^2].
    \end{equation}
    We let $n_{\std}$ be the number of standard intervals, which is a $\bbZ_{\geq 0}$-valued random variable. If the event $E_{\star\star}(C_1\alpha)$ holds, then conditionally on $\omega_{[0,T]}\in E_{\star\star}(C_1\alpha)$, the law of $n_{\std}$ is Poisson with mean:

\begin{align}
    \notag
    \bbE^{\wh\Theta_{C_1\alpha,\omega}}[n_{\std}~|~\omega_{[0,T]}]
    &= C_1\alpha \int_{0\leq s\leq t<3}\int_{\alpha/C_1^2}^{2\alpha/C_1^2}\sqrt{\frac{2}{\pi}}
    \exp\lt(-\frac{u^2 |\omega_s-\omega_t|^2}{2}-(t-s)\rt)\de u ~\d s \d t 
  \\
    \notag
    &\leq \sqrt{\frac{2}{\pi}}\frac{\alpha^2}{C_1}\int\limits_{\substack{0\leq s\leq t\leq 3\\
    |\omega_s-\omega_t|\geq C_1^2/(\alpha\eps )}} \e^{-(t-s)}\e^{-1/2\eps^2} \d s \d t
    \\
    &\quad+ \sqrt{\frac{2}{\pi}}\frac{\alpha^2}{C_1^2}\int\limits_{\substack{0\leq s\leq t\leq 3\\
    |\omega_s-\omega_t|\leq C_1^2/(\alpha\eps )}}\e^{-(t-s)}\d s \d t 
    \\
    \notag
    &\leq \sqrt{\frac{2}{\pi}}\frac{\alpha^2}{C_1}\big(\e^{-1/2\eps^2}A_{*}+(1-\eps )A_{*}\big)
    \\
    \label{eq:poisson-mean}
    &\leq 
   \sqrt{\frac{2}{\pi}}\frac{\alpha^2}{C_1}A_{*}\big(1-\frac{\eps}{2}\big).
    \end{align}

    Since $\alpha\gg C_1\gg 1/\eps$, if a random variable $n$ is Poisson with mean $\bbE[n]\leq \sqrt{\frac{2}{\pi}}\frac{\alpha^2}{C_1}A_{*}\Big(1-\frac{\eps}{2}\Big)$, then
    \[
    \bbP\bigg[n\leq \sqrt{\frac{2}{\pi}}\frac{\alpha^2}{C_1}A_{*}\Big(1-\frac{\eps}{4}\Big)\bigg]\geq 1/2.
    \]
    In particular combining with \eqref{eq:poisson-mean} yields:
    \begin{equation}
    \label{eq:given-omega-few-intervals}
    \begin{aligned}
    \bbP^{\wh\Theta_{C_1\alpha,\omega}}\bigg[n_{\std}\leq \sqrt{\frac{2}{\pi}}\frac{\alpha^2}{C_1}A_{*}\Big(1-\frac{\eps}{4}\Big) \Big)~\bigg|~E_{\star\star}\bigg]
    &\geq 
    1/2
    \\
    \implies 
    \bbP^{\wh\Theta_{C_1\alpha}}\bigg[n_{\std}\leq \sqrt{\frac{2}{\pi}}\frac{\alpha^2}{C_1}A_{*}\Big(1-\frac{\eps}{4}\Big)\bigg]
    &\geq 
    p_{\star\star}(C_1\alpha)/2.
    \end{aligned}
    \end{equation}   

    Next we return to $\wh\gamma_{\alpha}$ and $\wh\P_{\wh\gamma_\alpha}$. We can use Lemma~\ref{lem:confinement-after-truncation} to similarly lower bound the number of intervals obeying \eqref{eq:count-interval-to-compare}.
    Recall that this lemma ensures $p_{\star}(\alpha)\leq \e^{-\alpha^2}$.
    Indeed if $|\omega_s-\omega_t|\leq C_1/\alpha$, then since $C_1\gg 1/\eps$ the intensity of standard intervals $([s,t],u)$ with $\alpha/C_1^2\leq u\leq 2\alpha/C_1^2$ is 
    \begin{equation}
    \label{eq:standard-interval-intensity-usual}
    \sqrt{\frac{2}{\pi}}
    \e^{-(t-s)}
    \e^{-\frac{u^2 |\omega_s-\omega_t|^2}{2}}
    C_1\alpha
    \geq 
    \sqrt{\frac{2}{\pi}}\,\e^{-(t-s)}\e^{-2/C_1^2}C_1\alpha.
    \end{equation}
    Conditioned on the complementary event $E_{\star}(\alpha)^c$ (defined in Lemma \ref{lem:confinement-after-truncation})  and integrating out $u$, an estimate similar to \eqref{eq:poisson-mean} shows the total intensity of standard intervals is at least 
    \begin{align*}
      \bbE^{\wh\P_{\wh\gamma_{\alpha},\omega}}[n_{\std}~|~E_{\star}(\alpha)^c]
      &\geq \e^{-4/C_1^2}\sqrt{\frac{2}{\pi}}\frac{\alpha^2}{C_1}\int\limits_{\substack{0\leq s\leq t\leq 3\\
    |\omega_t-\omega_s|\leq C_1/\alpha}}\e^{-(t-s)}\d s \d t 
    \\
    &\geq \sqrt{\frac{2}{\pi}}\e^{-4/C_1^2}\frac{\alpha^2}{C_1}A_{*}(1-\eps^2)
    \geq \sqrt{\frac{2}{\pi}}\frac{\alpha^2}{C_1}A_{*}(1-\eps/8)
    \end{align*}
    if $C_1$ is large enough.
    Since $\alpha\gg C_1$, a Poisson variable $n$ with $\bbE[n]\geq \sqrt{\frac{2}{\pi}}\frac{\alpha^2}{C_1}A_{*}(1-\eps/8)$ satisfies 
   \begin{align*}
    \bbP\Big[n\leq \sqrt{\frac{2}{\pi}}\frac{\alpha^2}{C_1}A_{*}(1-\eps/4)\Big]&\leq \exp\Bigg(\sqrt{\frac{2}{\pi}}\frac{\alpha^2}{C_1}A_{*}\Big((1-\eps/4)\log \bigg(\frac{1-\eps/8}{1-\eps/4}\bigg)-\eps/8\Big)\Bigg)
    \\
    &=\e^{-\alpha^2 C_2(\eps)}, 
   \end{align*}
   where $C_2(\eps):=-\sqrt{\frac{2}{\pi}}\frac{A_{*}}{C_1}\Big((1-\eps/4)\log \bigg(\frac{1-\eps/8}{1-\eps/4}\bigg)-\eps/8\Big)\in (0,1)$ for $0<\eps <1$.
    Combined with Lemma~\ref{lem:confinement-after-truncation} and taking $n=n_{\std}$ on the complementary event $E_{\star}(\alpha)^c$, we conclude that 
    \begin{align*}
    \bbP^{\wh\Theta_{\wh\gamma_{\alpha}}}\bigg[n_{\std}\geq \sqrt{\frac{2}{\pi}}\frac{\alpha^2}{C_1}A_{*}(1-\eps/4)\bigg]
    &\geq 
    (1-\e^{-\alpha^2})
    (1-\e^{-C_2(\eps)\alpha^{2}})
    \\
    &\geq 1-2\e^{-\alpha^2 C_2(\eps)}
    \\
    &\geq 1-\e^{-\alpha^2 C_2(\eps)/2}
    \end{align*}
    for $\alpha$ large enough, since $C_2(\eps)\in (0,1)$. By the comparison \eqref{eq:polaron-domination-concrete-basic}, we deduce
    \begin{align*}
    \e^{-\alpha^2 C_2(\eps)/2}
    &\geq 
    \bbP^{\wh\Theta_{\wh\gamma_{\alpha}}}\Bigg[n_{\std}\leq \sqrt{\frac{2}{\pi}}\frac{\alpha^2}{C_1}A_{*}(1-\eps/4)\Bigg]
    \\
    &\stackrel{\eqref{eq:polaron-domination-concrete-basic}}{\geq}
    \bbP^{\wh\Theta_{C_1\alpha}}\Bigg[ 
    n_{\std}\leq \sqrt{\frac{2}{\pi}}\frac{\alpha^2}{C_1}A_{*}(1-\eps/4)\Bigg]
    \stackrel{\eqref{eq:given-omega-few-intervals}}{\geq}
    p_{\star\star}(C_1\alpha)/2.
    \end{align*}
    The result of the previous line completes the proof by setting $\tilde{C}_1(\eps):=\frac{C_2(\eps)}{4C_1(\eps)^2}$,  as it gives:
    \[
    p_{\star\star}(\alpha)
    \leq 2\e^{-\frac{\alpha^2 C_2(\eps)}{2C_1(\eps)^2}}
    \leq \e^{-\alpha^2 \tilde{C}_1(\eps)}.
    \qedhere
    \]
    
\end{proof}

To apply Proposition~\ref{cor:polaron-many-intervals} in the next section, we need to reformulate it in terms of Poissonian intensities of intervals (conditionally on $\omega$) rather than empirical interval counts.
We say the interval $([s,t],u) \in (\hat\xi,\hat u)$ is \textit{super-standard} if $k\leq s<t\leq k+3$ for some $k\in\bbZ$, and $u\in [\alpha/C_1^4, 2\alpha/C_1^4]$. 
Let $\mu_{\bbR}$ and $\mu_{\bbR^2}$ denote Lebesgue measure on $\bbR$ and $\bbR^2$, and let $\de\Lambda_{\sstd}([a,b])$ be the $\omega$-dependent intensity of super-standard intervals (recall \eqref{def Lambda}), viewed as a density on $\bbR^2$, that is 
\begin{equation}\label{def-rho-sstd}
\begin{aligned}
\Lambda_{\sstd}(\d s \d t) &= \alpha \sqrt{\frac 2\pi} \e^{-(t-s)} 
\bigg(\int_{\frac\alpha{C_1^4}}^{\frac{2\alpha}{C_1^4}} \e^{-\frac{u^2|\omega(t)- \omega(s)|^2}{2}} \d u\bigg) \,\, \1_{k \leq s < t \leq k+3} \,\, \d s \d t \\
&= \alpha^2 \sqrt{\frac 2\pi} \e^{-(t-s)} 
\bigg(\int_{\frac 1 {C_1^4}}^{\frac{2}{C_1^4}} \e^{-\frac{u^2 \alpha^2 |\omega(t)- \omega(s)|^2}{2}} \d u\bigg) \,\, \1_{k \leq s < t \leq k+3} \,\, \d s \d t.
\end{aligned}
\end{equation}
For each $k\in \{0,1,\dots,\alpha^{10}\}$, let 
\begin{equation}
\label{eq:def-S-k}
S_k
:=
\bigg\{
(a,b)~:~
\Lambda_{\sstd}([a,b])\leq \alpha^2/C_1^5
\bigg\}
\subseteq [2k,2k+1]\times [2k+2,2k+3].
\end{equation}

\begin{cor}
\label{cor:super-standard-intensity}
    With parameters $1\ll 1/\eps\ll C_0\ll C_1\ll \alpha$ as in Lemma \ref{lem:confinement-after-truncation},
    \[
    \wh\bbP_{\alpha}\big[\mu_{\bbR^2}(S_k)\leq \eps~~~\forall ~0\leq k\leq \alpha^{10}\big]
    \geq 1-\e^{-\alpha}.
    \]
\end{cor}

\begin{proof}
    By Proposition~\ref{cor:polaron-many-intervals} and the bound $(1+\alpha^{10})\e^{-2\alpha}\leq \e^{-\alpha}$, it suffices to show that if $|\omega_a-\omega_b|\leq C_1^3/(\eps\alpha)$, then $(a,b)\notin S_k$.
    This deterministic estimate is similar to those within Proposition~\ref{cor:polaron-many-intervals}. In particular, one mimics \eqref{eq:standard-interval-intensity-usual} to argue that since $u$ is small, the interval intensity on $([a,b],u)$ is near-maximal.
\end{proof}

\begin{remark}
\label{rem:general-p}
    For general $p\in (0,2)$, scale invariance implies the Gaussian mixture representation
    \[
    \frac{1}{x^p}
    =
    A_p
    \int_{0}^{\infty}
    u^{p-1}
    e^{-u^2 x^2/2}
    \de u.
    \]
    As a result, our methods apply to this more general class of potentials (we require $0<p<2$ so that $Z_{\alpha,T}<\infty$, see \cite[Lemma 3.7]{DV83}). 
    Namely, if one sets
    \[
    \wh\gamma_{\alpha}=C\alpha\cdot \1_{u\leq \alpha^{\frac{1}{2-p}}}
    \]
    and takes $p$-standard intervals to have $u\in [\alpha^{\frac{1}{2-p}}/C^{O(1)}, 2\alpha^{\frac{1}{2-p}}/C^{O(1)}]$, then the result and proofs above remain essentially unchanged.
    The arguments of the next section then imply 
    \[
    \bbE^{\wh\bbP_\alpha}[|\omega_T-\omega_0|^2]\leq O(T\alpha^{-\frac{4}{2-p}}).
    \]
    (Here, the implicit constant may depend on $0<p<2$.)
    This is expected to be optimal based on a natural generalization of the Pekar conjecture explained in \cite{MS22}.
    The generalization discussed in Section~\ref{subsec:universality} similarly applies to this setting.
\end{remark}

\section{Proof of Proposition \ref{prop-second-mom}.}\label{sec long paths}




In this section, we will conclude the proof of Proposition \ref{prop-second-mom} (see Section \ref{sec proof prop second moment}). For this purpose, Section \ref{sec-good-verygood} and Section \ref{section-long-paths} will provide the recipe for defining the rounds of intervals from \eqref{def rounds} with the desired properties. 

With this aim, we now fix $\omega$ such that Corollary~\ref{cor:super-standard-intensity} holds with $\eps = 10^{-10}$, so in particular the sets $S_k=S_k(\omega)$ are fixed and 
\[
    \max_{k=0,\dots,\alpha^{10}} \mu_{\bbR^2}(S_k)\leq \eps\leq 10^{-10}.
\]
We write $\vec S=\{S_k\}_{k=0}^{\alpha^{10}}$.
In order to give ourselves some slack room (as will be important later), we will argue below using only the weaker bound $\max_k \mu_{\bbR^2}(S_k)\leq 10^{-7}$.
For convenience, we will write

\begin{equation}
\label{eq:g-S-k}
f_{S_k}(a,b): =\1_{(a,b)\in S_k}, \qquad\mbox{and}\qquad
g_{S_k}(a):= \int_{2k+2}^{2k+3} f_{S_k}(a,b)\de b.
\end{equation}

\subsection{Good and very good points.}\label{sec-good-verygood}
For simplicity, let us begin with $k=0$. Given the collection of intervals 
$S_0=\big\{(a,b): \Lambda_{\sstd}([a,b])\leq \frac{\alpha^2}{C_1^5}\big\} \subset [0,1]\times [2,3]$, we say that a point $a \in [0,1]$ 
is {\it good} w.r.t. $S_0$, written $a\in \good(S_0)$, if the corresponding right endpoints $b\in [2,3]$ of the intervals $[a,b]$ belonging to $S_0$ have measure at most $10^{-1}$, i.e., we write $a\in \good(S_0)$ if 
$$
g_{S_0}(a)= \int_2^3 \1_{(a,b)\in S_0} \d b \leq 10^{-1},
$$
else, we say $a\in [0,1]$ is {\it bad} for $S_0$. Extending this definition, we say a point $a \in [\frac 13, \frac 23]$ (chosen from the middle third) is {\it very good} w.r.t. the collection $S_0$, written $a\in \verygood(S_0)$, if for all $x \in (0,\frac 13]$, the density of left end-points $y \in [a,a+x]$ for which the measure $g_{S_0}(y)$ of the corresponding right endpoints $b$ of the intervals $(y,b)$ belonging to $S_0$ is at most $10^{-3}$, i.e., we declare $[\frac 13, \frac 23]\ni a \in \verygood(S_0)$ if \begin{equation}
\label{eq:def-very-good}
\frac{1}{x}\int_a^{a+x} g_{S_0}(y) \d y= \frac 1 x \int_a^{a+x}
\int_{2}^{3}  \1_{(y,b)\in S_0}\de b ~\de y\leq 10^{-3},\quad\forall x\in (0,1/3].
\end{equation}

For any $k=0,\dots, \alpha^{10}$, the definitions $\good(S_0)$ and $\verygood(S_0)$ extend in an obvious manner to $\good(S_k)$ and $\verygood(S_k)$ for the collections of intervals $S_k \subset [2k, 2k+1] \times [2k+2, 2k + 3]$. Moreover, applying the Lebesgue density theorem to 
$g_{S_k}(a)$ shows that $\verygood(S_k)\subseteq\good(S_k)\subseteq [2k,2k+1]$, up to sets of measure zero. By removing a measure zero set from $\verygood(S_k)$, we can and do assume this containment always holds. (We emphasize that although $S_k$ ``involves'' the interval $[2k+2,2k+3]$, only points in $[2k,2k+1]$ can be in $\good(S_k)$.)

Our purpose in Section \ref{section-long-paths} will be to construct long paths 
$
\{[s_1^{\ssup L}, t_1^{\ssup L}], [s_2^{\ssup L}, t_2^{\ssup L}], \dots, [s_{\alpha^{10}}^{\ssup L}, t_{\alpha^{10}}^{\ssup L}]\}_{L=1}^{\delta\alpha^2} \in \hat\xi
$
of intervals with {\it good left endpoints} $s^{\ssup L}_j$ and {\it very good right endpoints} $t^{\ssup L}_j$. For this purpose, 
we will need the notions (of good and very good points) for larger $S_k^{\dagger}\supseteq S_k$, which still obey $S_k^{\dagger}\subseteq [2k,2k+1]\times [2k+2,2k+3]$.
Below we give some simple preparatory results, taking $k=0$ for convenience. Importantly, these results do not rely on the precise definition \eqref{eq:def-S-k}.

\begin{prop}
\label{prop:very-good}
If $a\in\verygood(S_0^{\dagger})$ then
\[
\frac{1}{x}
\mu_{\bbR}\big([a,a+x]\cap \good(S_0^{\dagger})\big)
\geq 1-10^{-1},
\quad\forall x\in (0,1/3].
\]
\end{prop}

\begin{proof}
    Given $x\in (0,1/3]$, the definition of $\good(S_0^{\dagger})$ yields
    \begin{align*}
    10
    \int_a^{a+x}
    \int_{2}^{3} 
    f_{S_0^{\dagger}}(y,b)\de b ~\de y
    \geq 
    \mu_{\bbR}([a,a+x]\backslash \good(S_0^{\dagger})).
    \end{align*}
    If the claim to be proved is false, then the right-hand side is at least $x/10$, and rearranging contradicts the definition \eqref{eq:def-very-good}.
\end{proof}

\begin{prop}
\label{prop:mostly-very-good}
    If $\mu_{\bbR^2}(S_0^{\dagger})\leq 10^{-7}$, then 
    $\mu_{\bbR}\big(\verygood(S_0^{\dagger})
    \big)\geq 0.33$. 
\end{prop}

\begin{proof}
    We apply the one-dimensional Hardy--Littlewood maximal inequality to $g=g_{S_0^{\dagger}}$ (recall \eqref{eq:g-S-k}).
    For concreteness, we define $g(a)=0$ for $a\in\bbR\backslash [0,1]$. Then by the hypothesis $|g|_{L^1}\leq 10^{-7}$, hence 
    \[
    Mg(a):= \sup_{x\geq 0} \frac{1}{2x}\int_{a-x}^{a+x} g(y)~\de y
    \]
    satisfies 
    \[
    \mu_{\bbR}\lt(\bigg\{a:Mg(a)\geq 10^{-3}/2\bigg\}\rt)\leq \frac{6\cdot 10^{-7}}{10^{-3}}\leq \frac{1}{1000}\leq \frac{1}{3} - 0.33.
    \]
    Since $\{a:Mg(a)\leq 10^{-3}/2\}\subseteq \verygood(S_0^{\dagger})$ the conclusion follows.
\end{proof}

\begin{prop}
\label{prop:good-points-connected}
    If $\mu_{\bbR^2}(S_0^{\dagger})\leq 10^{-7}$ and $\mu_{\bbR^2}(S_1^{\dagger})\leq 10^{-7}$, then for all $a\in \good(S^\dagger_0)$,
    \[
    \mu_{\bbR}\Big(
    \verygood(S_1^{\dagger})\cap
    \big\{
    b:
    (a,b)\notin S_0^{\dagger}
    \big\}\Big)
    \geq 0.2.
    \]
\end{prop}

\begin{proof}
    This lemma follows by the definition of $\good(S_0^{\dagger})$ combined with Proposition~\ref{prop:mostly-very-good}, since $0.33-0.1\geq 0.2$.
\end{proof}

\subsection{Finding long paths of intervals.}\label{section-long-paths} We set a final parameter $\delta>0$ so that
\begin{equation}
\label{eq:delta-def}
C_1\ll 1/\delta\ll\alpha.
\end{equation}
We will find $\delta\alpha^2$ long paths of intervals connecting $[0,1]\to [2,3]\to\dots\to [2\alpha^{10},2\alpha^{10}+1]$, with small gaps of {\it typical} order $O_{C_1}(\alpha^{-2})$. We will continue to {\it treat $\omega$ fixed} and restrict to super-standard intervals (recall \eqref{eq:count-interval-to-compare}). Hence, the collection $\vec S=\{S_k\}_{k=0}^{\alpha^{10}}$ defined in \eqref{eq:def-S-k} will also be fixed. 


Our purpose will now be to prove an algorithm that constructs intervals iteratively using very good points (recall Section \ref{sec-good-verygood}). To start with, we will construct 
$$
\begin{aligned}
&\mbox{the first path:}\qquad
\eta_1:=\bigg\{(s_1^{\ssup 1},t_1^{\ssup 1}),(s_2^{\ssup 1},t_2^{\ssup 1}),\dots,(s_{\alpha^{10}}^{\ssup 1},t_{\alpha^{10}}^{\ssup 1})\bigg\}, \quad\mbox{and then}\\ 
&\mbox{the second path:}\qquad
\eta_2:=\bigg\{(s_1^{\ssup 2},t_1^{\ssup 2}),(s_2^{\ssup 2},t_2^{\ssup 2}),\dots,(s_{\alpha^{10}}^{\ssup 2},t_{\alpha^{10}}^{\ssup 2})\bigg\} 
\end{aligned}
$$
and so on to construct the paths for all steps $L=1,2,\dots,\delta\alpha^2$. 
Our construction will always ensure:
$$
2k+ \frac 13\leq t_k^{\ssup L}\leq 2k+ \frac 2 3\qquad\mbox{and}\qquad 2k+ \frac 13\leq s_{k+1}^{\ssup L}\leq 2k+1.
$$

At the beginning of the first step (i.e., for $L=1$), we set 
\begin{equation}\label{def-S-aug1}
S^{\dagger}_0(1):= S_0, \quad S^{\dagger}_1(1):= S_1, \quad S^{\dagger}_2(1)= S_2, \dots, S^{\dagger}_{\alpha^{10}}(1):=S_{\alpha^{10}}
\end{equation}
to be the set of super-standard intervals $S_0$ in $[0,1]\times[2,3]$, the set of super-standard intervals $S_1$ in $[2,3]\times[4,5]$ and so on, and 
\begin{equation}\label{def-UV1}
 U_{j,1}:= \good(S_j^{\dagger}(1))\subseteq [2j,2j+1],
    \quad\quad
    V_{j,1}:= \verygood(S_j^{\dagger}(1))\subseteq [2j,2j+1]  \quad \forall j=0,\dots,\alpha^{10}, 
\end{equation}
for the set of good and very good points in $[2j,2j+1]$, relative to $S^{\dagger}_j(1)$, for all $j=0,\dots,\alpha^{10}$. 

Assuming \textit{failure has not yet been declared} (see below for the definition of ``declaring failure"), step $L$ will construct intervals $[s_j^{\ssup L},t_j^{\ssup L}]$ with \textit{good left endpoint} $s_j^{\ssup L}\in U_{j-1,L}$ and \textit{very good right endpoint} $t_j^{\ssup L}\in V_{j,L}$. The first step of the algorithm (i.e., the case $L=1$) is defined as follows:
\begin{itemize}
    \item We choose some $t_0^{\ssup 1}\in V_{0,1} \subset [\frac 13, \frac 23]$.
    \item Let $s_1^{\ssup 1}$ be the smallest value satisfying the following:
    \begin{itemize}
        \item $s_1^{\ssup 1} \in U_{0,1} \cap [t_0^{\ssup 1}, 1]$ and 
\item The collection of point process intervals  $(\hat\xi,\hat u)\in \wh{\mathscr Y}_{\alpha^{10}}$ contains a super-standard interval $[s_1^{\ssup 1}, t_1^{\ssup 1}]$ with $t_1^{\ssup 1} \in V_{1,1} \subset [\frac 73, \frac 83]$.
\end{itemize}
\item Similarly, for $j=2,\dots,\alpha^{10}$, let $s_{j}^{\ssup 1}$ be the smallest value satisfying 
\begin{itemize}
\item $s_j^{\ssup 1} \in U_{j-1,1} \cap [t_{j-1}^{\ssup 1}, 2j-1]$ and 
\item The collection of intervals $\hat\xi$ contains a super-standard interval $[s_j^{\ssup 1}, t_j^{\ssup 1}]$ with $t_j^{\ssup 1} \in V_{j,1} \subset [ 2j+\frac{1}3,  2j+\frac 23]$.
\end{itemize}
\item If such $s_j^{\ssup 1}$ ever fails to exist for any $j$, we set $s_j^{\ssup 1}=\infty$, {\it declare failure} and stop the algorithm. 
\item If for any $j=0,\dots,\alpha^{10}$, 
$\mu_{\R^2}(S^{\dagger}_j(1))\geq 10^{-7}$, {\it declare failure} and stop the algorithm.
\end{itemize}
The steps above define the first step $L=1$. For $L=2,\dots, \delta \alpha^2$, we proceed very similarly, except that at the beginning of step $L\geq 2$, we augment the set of bad pairs (compare \eqref{def-S-aug1}) to 
    \begin{equation}
    \label{eq:augmented-bad-pairs}
    S^{\dagger}_j
    =
    S^{\dagger}_j(L)
    :=
    S_j
    \cup 
    \bigg(
    \bigg(
    \bigcup_{\ell=1}^{L-1}
    [t_j^{\ssup \ell},s_{j+1}^{\ssup \ell}]
    \bigg)
    \times 
    [2j+2,2j+3]\bigg).
    \end{equation}
Hence $S^{\dagger}_k(L)$ includes all intervals whose left-endpoint is within a ``gap'' from a previous step. Similar to \eqref{def-UV1}, we set 
    \begin{equation}
    \label{eq:U-V-j-L}
    U_{j,L}:=\good(S_j^{\dagger}(L))\subseteq [2j,2j+1],
    \quad\quad
    V_{j,L}:=\verygood(S_j^{\dagger}(L))\subseteq [2j,2j+1]
    \end{equation}
for the sets of good and very good points in $[2j,2j+1]$ relative to $S_j^{\dagger}(L)$. We also let 
\begin{equation}\label{def-xiL}
\xi_L := 
\bigcup_{\ell=1}^{L-1}
\Big\{[s_j^{\ssup\ell)},t_{j}^{\ssup\ell}]\Big\}_{j=0}^{\alpha^{10}}
\end{equation}
consist of all intervals used before step $L$; hence $\emptyset=\xi_1\subseteq \xi_2\subseteq \dots\subseteq \xi_{\delta\alpha^2}\subseteq \xi$. To define the intervals $\{[s_j^{\ssup L}, t_j^{\ssup L}]\}_{L=2}^{\delta\alpha^2}$ we now proceed as above:
\begin{enumerate}[label = {(\alph*)}]
\item 
We choose $t_0^{\ssup L}\in V_{0,L}\subseteq [1/3,2/3]$ and let $s_j^{\ssup L}\in U_{j-1,L}\cap [t_{j-1}^{\ssup L},2j-1]$ be the smallest value such that $\hat\xi\setminus \xi_L$ (recall \eqref{def-xiL}) contains a super-standard interval $[s_j^{\ssup L},t_j^{\ssup L}]$ with $t_j^{\ssup L}\in V_{j,L}\subseteq [2j+1/3,2j+2/3]$.
    \item 
    \label{it:failure-tail}
    If such $s_{j}^{\ssup L}$ ever fails to exist, set $s_{j}^{\ssup L}=\infty$, declare failure and stop the algorithm.
    \item  
    \label{it:augment-S-failure}
    If for some $1\leq \ell \leq L$ and $j=0,\dots, \alpha^{10}$, it holds that
    $\mu_{\bbR^2}(S^{\dagger}_j(\ell))\geq 10^{-7}$, then declare failure and stop the algorithm.
   \end{enumerate}
Given a fixed $\omega$, we write
\begin{equation}
\label{eq:algorithm-filtration-def}
\cF_{s,k,L}
=
\sigma\big(t_0^{\ssup 1},s_1^{\ssup 1},\dots, s_k^{\ssup L}
\big)
\end{equation}
to denote the $\sigma$-algebra generated by the information until $s_k^{(L)}$ defined in the above algorithm.(Here, we treat $\omega$ as a deterministic continuous function since it has been fixed; one could alternatively include $\omega$ as another random variable generating each $\sigma$-algebra above.)
The $\sigma$-algebra $\cF_{t,k,L}$ is defined similarly based on information up to $t_k^{(L)}$. 
For any $0\leq x\leq 2\alpha^{10}+1$, let (recall \eqref{def-U})
\[
U_{\delta\alpha^2}(x)
=
\begin{cases}
    
\sum_{L=1}^{\delta\alpha^2} 
\sum_{j=0}^{\alpha^{10}} \1_{x\in [t_j^{\ssup L},s_{j+1}^{\ssup L}]},\quad \text{if failure is \emph{never} declared},
\\
\delta\alpha^2,\quad\quad\quad\quad\quad\qquad\quad\qquad \text{if failure is \emph{ever} declared.}
\end{cases}
.
\]
In the first (main) case that failure does not occur, $U_{\delta\alpha^2}(x)$ is the number of gaps containing $x$.

We will now prove Proposition \ref{prop-second-mom}. Recall from Corollary \ref{cor:super-standard-intensity} that the event 
$$
E_\alpha^c:=\bigg\{\omega: \max_{k=0,\dots,\alpha^{10}}\mu_{\R^2}(S_k) \leq 10^{-9}\bigg\}
\qquad\mbox{satisfies}\qquad\wh\P_\alpha[E_\alpha^c]\geq 1- \e^{-\alpha}
\quad\mbox{for large $\alpha$.}
$$

\begin{lemma}
\label{lem:alg-succeeds}
    On the event $E_\alpha^c$ and with the prescribed constant $C_1$ defining $S_k$ in \eqref{eq:def-S-k}, there is an absolute constant $\overline A=\overline A_{C_1}\in (0,\infty)$ such that 
    \begin{equation}
    \label{eq:expected-QV-bounded}
    \limsup_{\alpha\to\infty}\frac{1}{\alpha^{10}}
    \bbE^{\widehat\Theta_{\alpha,\omega}}\Big[
    \int_0^{2\alpha^{10}+1}U_{\delta\alpha^2}(x)^2~\de x
    \Big]
    \leq \overline A.
    \end{equation}
\end{lemma}


\begin{remark}\label{remark T0 T}
 As mentioned in Remark \ref{remark T0}, it might also be possible to construct the algorithm (a) - (c) (defined below \eqref{def-xiL}) on the entire interval $[0,T]$ by declaring ``restart" instead of ``failure." Then applying Lemma \ref{lemma-square-root} and Corollary \ref{cor-second-mom} on $[0,T]$ directly would imply the proof of Theorem \ref{thm} (without requiring sub-additivity from Corollary \ref{cor:subadditivity}).

\end{remark}

\subsection{Proof of Lemma~\ref{lem:alg-succeeds}.}\label{sec proof prop second moment}

We will use the following terminology: we say that a $\bbR_{\geq 0}$-valued random variable $X$ is $\oC$-subexponential if $\bbP[X\geq x]\leq \e^{-x/\oC}$ for all $x\geq 0$. For the proof of Lemma~\ref{lem:alg-succeeds}, we will need the following estimate for sums of independent exponential random variables, which follows by Bernstein's inequality after centering: 

\begin{lemma}
\label{lem:subexponential-sum}
    Given $\oC>0$, there exists $\delta>0$ such that the following holds for large enough $\alpha$.
    Let $X_1,\dots,X_{\delta\alpha^2}$ be independent almost surely positive $\oC$-subexponential random variables.
    Then 
    \[
    \bbP\Big[
    \frac{1}{\alpha^2}
    \sum_{i=1}^{\delta\alpha^2} X_i 
    \geq 
    10^{-10}
    \Big]
    \leq 
    \e^{-\delta\alpha}.
    \]
\end{lemma}

    Let us turn to the proof of Lemma \ref{lem:alg-succeeds} for which we recall the notation $U_{j,L},V_{j,L}$ from \eqref{eq:U-V-j-L} and $\cF_{s,j,L},\cF_{t,j,L}$ from \eqref{eq:algorithm-filtration-def}.
    Below we will apply Propositions~\ref{prop:very-good}, \ref{prop:mostly-very-good}, \ref{prop:good-points-connected} for $S^{\dagger}$ in \eqref{eq:augmented-bad-pairs} \textit{on the event that failure has not yet been declared}; this is possible due to part \ref{it:augment-S-failure} of the algorithm.  
    We decompose the proof into three steps.
  
    
    \paragraph{\textit {Step $1$ (Gap lengths are $O_{C_1}(\alpha^{-2})$-subexponential):}}
    For a gap $[t_j^{\ssup L},s_{j+1}^{\ssup L}]$, define the truncated waiting time $\oq_j^{\ssup L}=\min(1,s_{j+1}^{\ssup L}-t_j^{\ssup L})$.
    (Thus $\oq_j^{\ssup L}=1$ if and only if $s_{j+1}^{\ssup L}=\infty$.)
    Also, let
    \begin{equation}
    \label{eq:q-j-L}
    q_j^{\ssup L}
    :=
    \mu_{\bbR}\big([t_j^{\ssup L},s_{j+1}^{\ssup L}]\cap U_{j,L}\big)
    \leq \oq_j^{\ssup L}
    \end{equation}
    be the Lebesgue measure of good points inside $[t_j^{\ssup L},s_{j+1}^{\ssup L}]$. Assuming no failure has occurred so far, we claim that $\oq_j^{\ssup L}$ is $\frac{C_1^7}{\alpha^2}$-subexponential, even conditionally on $\cF_{t,j,L}$.

    To prove the claim, the first important step is to observe that since $t_j^{\ssup L}\in V_{j,L}$ is very good, we have by Proposition~\ref{prop:very-good} the almost sure bound:
    \[
    \oq_j^{\ssup L}
    \leq 
    10q_j^{\ssup L}.
    \]
    Indeed even if $\oq_j^{\ssup L}=1$, this follows by taking $x=1/3$ in Proposition~\ref{prop:very-good} since $t_j^{\ssup L}\leq 2j+2/3$ by definition of $V_{j,L}$.
    Hence it suffices to show $q_j^{\ssup L}$ is $C_1^6/\alpha^2$-subexponential under the same conditioning. 
    For this, we observe that for any $X\geq t_j^{\ssup L}$, the conditional probability to have $s_{j+1}^{\ssup L}\geq X$ is exactly
    \begin{align}
    \label{eq:gap-tail-abstract}
    &\bbP^{\wh\Theta_{\alpha,\omega}}\bigg[s_{j+1}^{\ssup L}\geq X~|~\cF_{t,j,L}\bigg]
    =
    \exp\big(-f_{j,L}(X)\big)
    \end{align}
    where
    \[
    f_{j,L}(X)
    :=
    \int_{t_j^{\ssup L}}^{X} \d x 
    \1_{x\in U_{j,L}}
    \int_{V_{j+1,L}}
    \Lambda_{\sstd}([x,y])
    \de y.
    \]
    Indeed, the above identity follows from the fact that conditional on $\cF_{t,j,L}$, the (super-standard) intervals used in our algorithm are Poissonian and $f_{j,L}$ is the integrated intensity for the left endpoint of an interval. Here the augmentation \eqref{eq:augmented-bad-pairs} importantly accounts for the negative information accrued from gaps in previous iterations $\ell<L$. Next, we lower bound $f_{j,L}(X)$. Note that, if $(x,y)\notin S^\dagger_j$ (with $x \in [2j,2j+1]$ and $y\in [2j+2,2j+3]$), then $(x,y)\notin S_j$, and then by definition of the set $S_j$ (recall \eqref{eq:def-S-k}), for such $(x,y)$,
    \[
    \Lambda_{\sstd}([x,y])\geq \alpha^2/C_1^5,
    \]
    Hence, for $X\in [t_j^{\ssup L},2j+1]$:
    \begin{align*}
    f_{j,L}(X)
    \geq 
    \frac{\alpha^2}{C_1^5}
    \int\limits_{[t_j^{\ssup L},X]\cap U_{j,L}}
    \int_{2j+7/3}^{2j+8/3}
    \1_{y\in V_{j+1,L}}
    \1_{(x,y)\notin S_j^{\dagger}}
    \de y
    ~\de x
    \geq 
    \frac{\alpha^2}{C_1^6}
    \mu_{\bbR}\big([t_j^{(L)},X]\cap U_{j,L}\big),
    \end{align*}
    where we used 
    Proposition~\ref{prop:good-points-connected} in the second lower bound above. Now by definition, $q^{\ssup L}_j= \mu_\R([t^{\ssup L}_j, s^{\ssup L}_{j+1}]\cap U_{j,L})$, and therefore \eqref{eq:gap-tail-abstract} shows that conditionally on $\cF_{t,j,L}$, $q_j^{\ssup L}$ is $C_1^6/\alpha^2$-subexponential, and therefore, $\overline{q}_j^{\ssup L}$ is $C_1^7/\alpha^2$-subexponential, as claimed.



    \noindent{\textit {Step $2$ (Failure is exponentially unlikely):}} Again conditioning on $\cF_{t,j,L}$ with the assumption that failure has not yet occurred, the probability of failure from criterion~\ref{it:failure-tail} is 
    \[
    \bbP^{\wh\Theta_{\alpha,\omega}}\big[\oq_j^{\ssup L}=1~|~\cF_{t,j,L}\big]
    \leq
    \e^{-\alpha^2/C_1^7}
    \]
    by the previous step. Let $E_{j,L}$ denote the event that failure has not occurred by the time $t_j^{\ssup L}$ is chosen. We claim that Lemma~\ref{lem:subexponential-sum} and the previous step imply that for each $(j,L)$, without any conditioning,
    \begin{equation}
    \label{eq:subexponential-bound-within-algorithm}
    \bbP^{\wh\Theta_{\alpha,\omega}}\big[
    \sum_{\ell=1}^L
    \1_{E_{j,\ell}} \,\, 
    \oq_j^{\ssup \ell}
    \geq 10^{-10}
    \big]
    \leq \e^{-\delta\alpha}.
    \end{equation}
    Indeed, in Step 1 above we have just shown $\1_{E_{j,\ell}}\oq_j^{\ssup\ell}$ is $\cF_{t,j,\ell}$-conditionally $O_{C_1}(\alpha^{-2})$-subexponential for each $1\leq\ell\leq \delta\alpha^2$.
    By definition, such conditional distributions are stochastically dominated by exponential random variables with mean $O_{C_1}(\alpha^{-2})$. 
    Hence for each fixed $j$, we find by direct construction a Markovian coupling with an independent sequence $X_1,\dots,X_{\delta\alpha^2}$ as in Lemma~\ref{lem:subexponential-sum} such that almost surely for all $1\leq\ell\leq \delta\alpha^2$:
    \[
    \1_{E_{j,L}}
    \oq_j^{\ssup \ell}
    \leq 
    X_{\ell}/\alpha^2.
    \]
    Using this coupling,\footnote{Here we use the following fact: Let $\lambda$ and $\mu$ be two probability measures on $\R$ and $\big\{\{P^{\ssup \ell}(x,\cdot)\}_{x\in \R}\big\}_{\ell=1,\dots,L}$ and 
    $\{Q(x,\cdot)\}_{x\in \R}$ are transition probability kernels obeying the stochastic domination relations
    $\lambda \preceq\mu$ and $P^{\ssup\ell}(x,\cdot) \preceq Q(y,\cdot)$ for all $\ell=1,\dots,L$ and all $x\leq y$. Then $\lambda \star P^{\ssup 1} \star \dots \star P^{\ssup L} \preceq \mu \star Q^{\star L}$, for all $L\in\N$ (see \cite[p.133, Lemma 5.3]{Lindvall}). In the present context, we set $Z_{j,L}= \sum_{\ell=1}^L \1_{E_{j,\ell}} \overline{q}_j^{\ssup\ell}$ and for any independent sequence of random variables $X_1,\dots,X_{\delta\alpha^2}$ which are all uniformly sub-exponential, we define $Z_{L}^\prime:= \frac 1 {\alpha^2} \sum_{\ell=1}^L X_\ell$. Using that  $\1_{E_{j,\ell}}\oq_j^{\ssup\ell}$ is $\cF_{t,j,\ell}$-conditionally $O_{C_1}(\alpha^{-2})$-subexponential for each $1\leq\ell\leq \delta\alpha^2$, then the previous fact implies that the probability on the LHS of \eqref{eq:subexponential-bound-within-algorithm} is dominated by the tail probability $\P[Z_{L}^\prime \geq 10^{-10}]$. We then apply Lemma \ref{lem:subexponential-sum} to the latter probability.}
    \eqref{eq:subexponential-bound-within-algorithm} is a direct consequence of Lemma~\ref{lem:subexponential-sum}, which applies because of how $\delta$ is chosen.
Union bounding over 
    $1\leq {j}\leq\alpha^{10}$, it follows that with probability $\e^{-\delta\alpha/2}$, none of these events occurs, hence failure criterion~\ref{it:augment-S-failure} never holds.
    Combined with the previous paragraph, we conclude that the total probability of failure is at most $\e^{-\delta\alpha/3}$.

    \noindent{\textit{Step $3$ (Proof of the $L^2$ bound):}} We will show \eqref{eq:expected-QV-bounded}. The previous step shows the contribution from failure events is $O(1)$, so we restrict below to the event that failure never occurs.
    Similarly the contribution from $x\in [0,1]\cup [2\alpha^{10},2\alpha^{10}+1]$ is clearly $O(1)$. Fixing $x\in [2j,2j+1]$, we claim that if $s_{j}^{\ssup L}$ is chosen without failure having occurred yet, then 
    \[
    \bbP^{\wh\Theta_{\alpha,\omega}}\bigg[ x\in [t_{j}^{\ssup L},s_{j+1}^{\ssup L}] ~\big |~ \cF_{s,j,L}\bigg]
    \leq \frac{\overline A_{C_1}}{\alpha^2}.
    \]
    This claim implies $U_{\delta\alpha^2}(x)$ is stochastically dominated by a Poisson variable with mean $O_{C_1}(1)$ for all $x\notin [0,1]\cup [2\alpha^{10},2\alpha^{10}+1]$ (restricted to the event of no failure), thus implying the desired result. It now remains to prove the claim.

    The idea is that $t_j^{\ssup L}$ is \textit{anti-concentrated}, even conditionally on $\cF_{s,j,L}$.
    Indeed by definition of the algorithm, the law  of $t_j^{\ssup L}$ under $\wh\Theta_{\alpha,\omega}$ conditional on $\cF_{s,j,L}$ is supported in $V_{j,L}\subseteq [2j,2j+1]$ with density
    \[
    \Lambda_{\cL_{j,L}}(b)
    =
    \frac{\Lambda_{\sstd}([s_j^{\ssup L},b])}
    {\int_{V_{j,L}} \Lambda_{\sstd}([s_j^{\ssup L},y])~\de y}
    \,.
    \]
    Recalling \eqref{def-rho-sstd}, note that for all $[a,b]$:
   $$
    \Lambda_{\sstd}([a,b])\leq \int_{\alpha/C_1^4}^{2\alpha/C_1^4}
    \alpha
    ~\d u
    =
    \alpha^2/C_1^4.
    $$
    Meanwhile again using \eqref{eq:def-S-k}, we have $\Lambda_{\sstd}([s_j^{\ssup L},y]) \geq \Lambda_{\sstd}([s_j^{\ssup L},y]) \1_{y\notin S^\dagger_{j-1}} \geq \frac{\alpha^2}{C_1^5}\1_{y\notin S^\dagger_{j-1}}$. Hence, since $s_j^{\ssup L}\in U_{j-1,L}$, Proposition~\ref{prop:good-points-connected} implies 
    \[
    \int_{V_{j,L}} \Lambda_{\sstd}([s_j^{\ssup L},y])~\d y
    \geq 
    \frac{\alpha^2}{C_1^5}
    \mu_{\bbR}\big(V_{j,L}\cap \{b:(s_j^{\ssup L},b)\notin S_{j-1}^{\dagger}\}\big)
    \geq 
    \frac{\alpha^2}{10C_1^5}.
    \]
    We conclude that law  of $t_j^{\ssup L}$ under $\wh\Theta_{\alpha,\omega}$, conditional on $\cF_{s,j,L}$, has uniformly bounded density in $[0,10C_1]$. 
    Applying the first step of this proof (subexponential gap lengths) conditionally on $\cF_{t,j,L}$, we conclude via:
   \begin{align*}
    \bbP^{\wh\Theta_{\alpha,\omega}}\big[ x\in [t_{j}^{\ssup L},s_{j+1}^{\ssup L}] ~|~ \cF_{s,j,L}\big]
    &=
    \bbE^{\wh\Theta_{\alpha,\omega}}\bigg[ \bbP^{\wh\Theta_{\alpha,\omega}}\big[x\in [t_{j}^{\ssup L},s_{j+1}^{\ssup L}] ~\big |~ \cF_{t,j,L}\big]~\bigg |~\cF_{s,j,L}\bigg]
    \\
    &\leq 
    10C_1
    \int_{2j}^{2j+1}
    \exp\bigg(-\frac{\alpha^2 |t-x|}{10 C_1^5}\bigg)
    ~\d t
    \leq \frac{\overline A_{C_1}}{\alpha^2}.
    \quad\quad\qed
    \end{align*}

\appendix

\section{Proof of Theorem \ref{thm-strong-coupling} for unbounded $V$.}\label{sec unbounded V}
We need to show \eqref{main} and \eqref{main2} when $V(|x|)= \frac 1 {|x|}$ or $V(|x|)=|x|$. Note that, for this purpose, we only need to verify that \eqref{g-eps-alpha} holds for such $V$. 
It suffices to show this for $\theta=1$. 
Let us write 
$$
\widetilde V(x)= \frac 1 {|x|}, \qquad \widetilde V_M= \frac 1 {\sqrt{{M^{-2}+ |x|}}}, \qquad\mbox{and}\quad \widetilde Y_M= \widetilde V - \widetilde V_M.
$$
Likewise, we write   
$$
|x|=\widehat{V}(|x|)= \widehat{V}_M(x)+ \widehat{Y}_{M}(x),\qquad \mbox{with }\widehat{V}_M(x)= |x| \wedge M.
$$

By  H\"older's inequality (with $\frac 1 p+ \frac 1 q=1$), the expectation in \eqref{g-eps-alpha} with $V=\widehat{V}$ (resp. $V=\widetilde{V}$) is bounded by $\widehat{A}_\eta(\alpha,T,p)\times \widehat{B}_\eta(\alpha,T,q)$ (resp. $\widetilde{A}_\eta(\alpha,T,p)\times \widetilde{B}_\eta(\alpha,T,q)$), where 
$$
\begin{aligned}
&\widehat{A}_\eta(\alpha,T,p)=
\E^\P\bigg[\exp\bigg(p\alpha\iint_{-T\le s\le t\le T} {\e^{-|t-s|}} \widetilde{V}_M(|\omega(t)-\omega(s)|) \d t\d s\\
&\qquad\qquad\qquad\qquad\qquad\qquad\qquad+p\eta \alpha^2\iint_{-T\le s\le t\le T} \e^{-(t-s)} \widehat{V}_M(\alpha|\omega(t)-\omega(s))|)\d t\d s \bigg)\bigg]^{\frac 1p}, \\
& \widehat{B}_\eta(\alpha,T,q)= \E^\P\bigg[\exp\bigg(q\alpha\iint_{-T\le s\le t\le T} {\e^{-|t-s|}} \widetilde{Y}_M(\omega_t-\omega_s) \d t\d s \\
&\qquad\qquad\qquad\qquad\qquad\qquad\qquad +q\eta \alpha^2\iint_{-T\le s\le t\le T} \e^{-(t-s)} \widehat{Y}_{M}(\alpha|\omega(t)-\omega(s)|)\d t \d s\bigg)\bigg]^{\frac 1q},
\end{aligned}
$$
and $\widetilde{A}_\eta(\alpha,T,p), \widetilde{B}_\eta(\alpha,T,q)$ are defined similarly by replacing $\widehat{V}_M$, $\widehat{Y}_M$ with $\widetilde{V}_M$ and $\widetilde{Y}_M$.
For any fixed $M$, the modified potentials $\widetilde V_{M}(x)=\frac1 {\sqrt{M^{-2}+ |x|^2}}$  and $\widehat{V}_M= |x| \wedge M$ are continuous and bounded. {Now recall that the statement \eqref{g-alpha} follows from a strong LDP for the empirical process of 3d-Brownian increments (\cite[Lemma 5.3]{MV18b}) and Varadhan's lemma. Applying the same LDP and Varadhan's lemma (now used for the continuous and bounded functions $\widetilde V_M$, resp. $\widehat V_M$),  we have, for any fixed $M>0$ and $\alpha>0$,}  
$$
\begin{aligned}
&\limsup_{T\to\infty}\frac 1 {2T}\log A_\eta(\alpha,T,p)
=g_\eta(p,\alpha,M) \\
&=\sup_{\mathbb Q}\bigg[  p  \alpha\int_0^\infty \e^{-t} \widetilde V_M(|\omega(t)-\omega(0)|) \d t+ p\eta \alpha^2\int_0^\infty  \e^{- t} V_M(|\alpha(\omega(t)-\omega(0)|)\d t  -H(\mathbb Q|\P) \bigg]
\end{aligned}
$$
for $A_\eta\in \{\widehat{A}_\eta,\widetilde{A}_\eta\}$ and $V_M\in \{\widehat{V}_M,\widetilde{V}_M\}$. Again, similar to \eqref{g-alpha}, the supremum above is taken over the space of stationary increments.  
 On the other hand, by Proposition \ref{lemma-Coulomb} (see below), for any $q>1$ and $B_\eta\in \{\widehat{B}_\eta,\widetilde{B}_\eta\}$, 
 $$
 \limsup_{M\uparrow\infty}\limsup_{T\to\infty}\frac 1 {2T}\log B_\eta(\alpha,T,q)=0.
 $$
 But
 $$
\lim_{p\downarrow 1}  \lim_{M\uparrow \infty} g_\eta(p,\alpha,M) =g_\eta(\alpha),
$$
where $g_\eta(\alpha)$ is $g_\eta(\alpha,\theta)$ for $\theta=1$ defined in \eqref{g-eps-alpha}. 
This proves Theorem \ref{thm-strong-coupling}. \qed

  \begin{prop}\label{lemma-Coulomb}
Let $V(x)=|x|$, $V_{M}(x)= V(x) \wedge M$, and $Y_{M}(x)= V(x)- V_{M}(x)$. Then for any $\lambda>0$ and $\alpha>0$,
\begin{equation}\label{est-lin}
\limsup_{M\to\infty}\limsup_{T\to\infty} \frac1{2T}\log\E^\P\bigg[\exp\bigg(\alpha\lambda\iint_{-T\leq s < t\leq T} \e^{-(t-s)} Y_{M}(\omega(s)-\omega(t)) \d s \d t\bigg)\bigg]=0.
\end{equation}
For $V(|x|)= \frac 1 {|x|}$ we have a similar statement for $V_M(x)= \frac 1 {\sqrt{|x|^2 + \frac 1 {M^2}}}$ and $Y_M= V- V_M$. 
\end{prop}

\subsection{Proof of Proposition \ref{lemma-Coulomb}.}
For the estimate relevant for $V(|x|)=\frac 1 {|x|}$, we refer to \cite[Lemma 4.3]{MV18b}. It remains to prove \eqref{est-lin} for $V(x)=|x|$.
In the following, we will write $\P_x$ for the law of a three-dimensional Brownian motion starting at $x\in \R^3$; while $\E_x$ will stand for the corresponding expectation, while $\P$ denotes the law of three-dimensional Brownian increments $(\omega(t)-\omega(s))_{s<t}$. For $T>0$, set   $\mathcal{F}_T:=\sigma(\{\omega(t)-\omega(s): - T\leq s < t \leq T\})$.

\begin{lemma}\label{exp.bound}
Let $G(\omega)$ be a $\mathcal{F}{_T}$-measurable function such that  $\sup_{x\in\R^3} \E^{\P_x}[\exp[G(\omega)]]\le \e^{\rho}$ for some $\rho>0$. Then for any $t>0$ and $x\in\R^3$,
$$
\E^{\P_x}\bigg[\exp\bigg(\frac{1}{T} \int_0^t G(\theta_s\omega) \d s\bigg)\bigg] \le \exp\bigg[\frac{\rho t}{T}\bigg].
$$
\end{lemma}

\begin{proof}
Since we can replace $G$ by $G-\rho  $, we can assume that $\rho=0$. For $s\le T$, let $k(s)=\sup\{k\in \N: s+kT\le t\}$, with $\theta$ being the canonical shift (i.e., $(\theta_s\omega)(\cdot)=\omega(s+\cdot)$) and 
$$\widehat{G} (s,\omega):=G(\theta_s \omega)+G(\theta_{s+T}  \omega)+\cdots+ G(\theta_{s+k(s)T}\omega).$$
Then 
$$\int_0^t G(\theta_s\omega) \d s= \int_0^T \widehat{G}(s,\omega) \d s.$$
By the assumption of the lemma, and by successive conditioning together with the Markov property, for every $s\le T$ we have 
$\E^{\P_x}[ \exp [\widehat{G}(s,\omega)]]\le 1$. 
Therefore
\begin{align*}
&\E^{\P_x}\bigg[ \exp \bigg[\frac{1}{T} \int_0^t G(\theta_s\omega) \d s\bigg]\bigg]= \E^{\P_x}\bigg[ \exp \bigg[\frac{1}{T} \int_0^T \widehat{G}(s,\omega)\d s\bigg]\bigg]\\
&\le \frac{1}{T}\int_0^T \E^{\P_x} \big[\exp\big[ \widehat{G}(s,\omega)\big]\big] \d s\le 1,
\end{align*}
which proves the lemma.
\end{proof}

\medskip
We recall that $\P$ denotes the law of three dimensional Brownian increments $\omega=(\omega(t)-\omega(s))_{s<t}$. If we set 
\begin{equation}\label{def F}
F(T, \omega)= \iint_{-T\le s <t\le T} \e^{-(t-s)} |\omega(t)-\omega(s)| \d s\d t, 
\end{equation}
our goal is to estimate, in the lemma below, 
$\frac {1}{2T}\log \E^\P[\exp[\alpha F(T,\omega)]]$:


\begin{lemma}\label{lemma-F}
We have for any $\alpha>0$
\[
\limsup_{T\to\infty}\frac{1}{2T}\log \E^\P\big[\exp[ \alpha F(T,\omega)\big]\big]\le C(\alpha)<\infty.
\]
\end{lemma}
\begin{proof}
Let \begin{equation}\label{eq-Gn-def}
  G_n(\omega):=\int_{n}^{n+1}|\omega(u)-\omega(0)|\d u.
\end{equation}
Then observe that \begin{align*}
  F(T,\omega)&=\int_{-T}^T \int_{0}^{T-s}\e^{-u}|\omega(s+u)-\omega(s)|\d u \d s \\
  &\leq \int_{-T}^T \int_{0}^{\infty}\e^{-u}|\omega(s+u)-\omega(s)|\d u \d s\\
  &=\int_{-T}^T\sum_{n=0}^{\infty}\int_{n}^{n+1}\e^{-u}|\omega(s+u)-\omega(s)|\d u \d s \\
  &\leq \sum_{n=0}^{\infty}\int_{-T}^T \int_{n}^{n+1} \e^{-n}|\omega(s+u)-\omega(s)|\d u \d s\\
  &=\sum_{n=0}^{\infty}\int_{-T}^T \e^{-n}G_n(\theta_s \omega)\d s.
  \end{align*}
By H\"{o}lder's inequality, we deduce that \begin{equation*}
  \log \E^\P\big[\exp[ \alpha F(T,\omega)\big]\big]\le \sum_{n=0}^{\infty}\frac{1}{2^{n+1}}\log \E^{\P}\bigg[\exp\bigg[ \alpha 2^{n+1}\e^{-n}\int_{-T}^T G_{n}(\theta_s \omega)\d s\bigg]\bigg].
\end{equation*}
Let $c_n:=(n+1)2^{n+1}\e^{-n}$, so that the last expectation can be written as \begin{equation*}
  \E^{\P}\bigg[\exp\bigg[ \frac{1}{n+1}\alpha\int_{-T}^T c_n G_{n}(\theta_s \omega)\bigg]\bigg]\leq \exp\left(\frac{\rho_n(\alpha) T}{n+1}\right)\leq \exp\left(\rho_n(\alpha) T\right), 
\end{equation*}
where we write $\rho_n(\alpha):=\sup_{x}\log \E^{\P_x}\left[\exp\left(\alpha c_n G_n(\omega)\right)\right]$. 
Therefore, \begin{equation*}
  \limsup_{T\to\infty}\frac{1}{2T}\log \E^\P\big[\exp[ \alpha F(T,\omega)\big]\big]\le \sum_{n=0}^{\infty}\frac{1}{2^{n+1}}\sup_{x}\log \E^{\P_x}\left[\exp\left(\alpha c_n G_n(\omega)\right)\right].
\end{equation*}
It remains to show that the right hand side is finite. Indeed, recalling the definition of $G_n$ and using Jensen's inequality, \begin{equation*}
  \E^{\P_x}\left[\exp\left(\alpha c_n G_n(\omega)\right)\right]\leq \int_{n}^{n+1}\E^{\P_0}\left[\e^{\alpha c_n |\omega(u)|}\right]\d u,
\end{equation*}
where we used that, for a fixed $u$, $\omega(u)-\omega(0)$ under $\P_x$ has the same distribution as $\omega(u)$ under $\P_0$. Noting that $|\omega(u)|\leq \sum_{i=1}^3 |\omega^i(u)|$ and the independence of the coordinates, we deduce that \begin{equation*}
  \int_{n}^{n+1}\E^{\P_0}\left[\e^{\alpha c_n |\omega(u)|}\right]\d u\leq \int_{n}^{n+1}\E\left[\e^{3\alpha c_n |X(u)|}\right]\d u,
\end{equation*}
where $X(u)\sim N(0,u)$. In particular, $\E\left[\e^{3\alpha c_n |X(u)|}\right] \leq 2 \e^{9\alpha^2 c_n^2 u}$. A crude bound gives us 
\begin{equation*}
  \int_{n}^{n+1}\E\left[\e^{3\alpha c_n |X(u)|}\right]\d u\leq 2\e^{9\alpha^2 c_n^2 (n+1)},
  \end{equation*}
so that \begin{equation*}
  \sum_{n=0}^{\infty}\frac{1}{2^{n+1}}\sup_{x}\log \E^{\P_x}\left[\exp\left(\alpha c_n G_n(\omega)\right)\right]\leq \sum_{n=0}^{\infty}\left(\frac{\log(2)}{2^{n+1}}+9\frac{\alpha^2 c_n^2 (n+1)}{2^{n+1}}\right),
\end{equation*}
which is clearly summable since $c_n:=(n+1)2^{n+1}\e^{-n}$.
\end{proof}

\noindent{\bf Completing the proof of Proposition \ref{lemma-Coulomb}:}

From Lemma \ref{lemma-F} it follows that, for any $\alpha,\lambda>0$, 
$$
\limsup_{T\to\infty}\frac{1}{2T}\log \E^\P\bigg[  \exp\bigg[  \alpha \lambda  \iint_{-T\le s<t\le T}  \e^{-(t-s)}  |(\omega(t)-\omega(s))| \d s\d t \bigg]\bigg]\leq C(\alpha,\lambda)<\infty.
$$ 
Thus, with $V(x)=|x|$, $V_M= V \wedge M$ and $Y_M= V-V_M$, we have for any $M>0$, 

$$
\limsup_{T\to\infty}\frac{1}{2T}\log \E^\P[  \exp[  \alpha \lambda  \iint_{-T\le s<t\le T}  \e^{-(t-s)} Y_M(|(\omega(t)-\omega(s))|) \d s\d t ] ]\leq C_M(\alpha,\lambda),
$$
 so that for any $\alpha,\lambda>0$, 
$$
\lim_{M\uparrow \infty} C_M(\alpha,\lambda)=0,
$$
 which proves Proposition \ref{lemma-Coulomb}. 
\qed

\section{Universality of Confinement Estimates}
\label{subsec:universality}

Using duality and FKG inequality on point processes, in Section \ref{subsec FKG comparison pot}, we established a form of stochastic monotonicity for path measures associated with potentials of the ``Gaussian mixture'' form \eqref{eq:def-interaction-general}. 
In fact, only the ``dominating'' measure needs to be a Gaussian mixture.
This allows our confinement results to extend to general potentials, which are ``at least as spatially unimodal'' as the Polaron interaction $1/|x|$.
The arguments below are, in some sense more general than that section. 
However, they require a slightly tedious finite-dimensional approximation scheme for $C([0,T];\bbR^3)$, and these arguments are not necessary to prove our main results for the Polaron itself. Similar ideas are also presented in the concurrent work \cite{S23} by one of us, where they play a more central role.

Suppose the even functions $V_{[a,b]}=V_{[a,b],\gamma}$ take the form \eqref{eq:def-interaction-general} for $0\leq a\leq b\leq T$, and let $\big(\wt V_{[a,b]}\big)_{0\leq a\leq b\leq T}$ be another family of even functions $\bbR\to\bbR$ (\textit{without} any associated parametrization $\gamma$) such that 
\[
    \wt V_{[a,b]}(x)-V_{[a,b],\gamma}(x)
\]
is decreasing on $\bbR_+$ and the partition function associated to $\wt V_{[a,b]}(x)$ (analogously to \eqref{eq:partition-function}) is finite.
We let $\wh\bbP_{V},\wh\bbP_{\wt V}$ be the associated probability measures on path in $C([0,T];\bbR^3)$.
Then combining the FKG and Gaussian correlation inequalities shows domination of $\wh\bbP_{\wt V}$ by $\wh\bbP_{V}$.
Recall that the (functional) Gaussian correlation inequality states the following.

\begin{prop}
\label{prop:GCI}
    Let $\mu$ be a centered Gaussian measure on $\bbR^d$. For any symmetric quasi-concave $f_1,f_2,\dots,f_n$:
    \begin{equation}
    \label{eq:GCI}
    \bbE^{\mu}\lt[\prod_{j=1}^m f_j(x)\rt]\cdot
    \bbE^{ \mu}\lt[\prod_{k=m+1}^{n} f_k(x)\rt]
    \leq
    \bbE^{ \mu}\lt[\prod_{i=1}^n f_i(x)\rt].
    \end{equation}
\end{prop}

\begin{proof}
    By Fubini's theorem, we are immediately reduced to the case that each $f_i=\1_{K_i}$ is the indicator of a symmetric convex set. 
    As intersections of such sets take the same form, it remains to handle to the case $(m,n)=(1,2)$, which is the usual statement of the Gaussian correlation inequality \cite{royen2014simple}.
\end{proof}

The main result of this subsection is as follows.

\begin{prop}
\label{prop:FKG-GCI-polaron}
    In the setting described above, for any $t,s\in [0,T]$ and continuous increasing $f:\bbR_{\geq 0}\to\bbR_{\geq 0}$ (in particular $f(x)=x^2$), we have
    \[
    \bbE^{\wh\bbP_{\wt V}}[f(|\omega_t-\omega_s|)]
    \leq 
    \bbE^{\wh\bbP_{V}}[f(|\omega_t-\omega_s|)].
    \]
\end{prop}

This result implies that our confinement estimates are universal, in the sense that any $\wt V$ which is ``more unimodal'' than the Polaron where $V(x)=1/|x|$ will be more confined.
One can similarly take $V(x)=1/|x|^p$ as in Remark~\ref{rem:general-p}.

Proposition~\ref{prop:FKG-GCI-polaron} will be shown by combining the FKG and Gaussian correlation inequalities. 
Since the Gaussian correlation inequality Proposition~\ref{prop:GCI} is stated for finite-dimensional subspaces of $C([0,T];\bbR^3)$, we begin with an approximation scheme, following \cite[Section 2.3]{S22}. We note that \cite{BSS23} showed how to avoid such approximations when the potential $V$ is bounded.
Given $V$ as in \eqref{eq:def-interaction-general} and $A>0$, we define
\[
    V^{(A)}_{[a,b],\gamma}
    =
    \int
    \1_{u\in [-A,A]}
    \exp\bigg(-\frac{u^2 x^2}{2}\bigg)~\de \gamma_{[a,b]}(u).
\]
It is easy to see that such $V^{(A)}_{[a,b],\gamma}$ are uniformly bounded on $\bbR$ (for each fixed $A$) thanks to the assumption \eqref{eq:locally-uniformly-finite}.
We write $\wh \bbP^{(A)}_V$ for $\wh\bbP_{V^{(A)}}$.
Next, for $\eta>0$, let $\bbP^{(\eta)}_{[0,T]}$ be the law of the piecewise linear process which agrees with $3$-dimensional Brownian motion when restricted to times in $\eta\bbZ$. 
Define the corresponding measure on piecewise-linear paths:
\[
    \de\wh\bbP^{(A,\eta)}_{V}(\omega)
    =
    \frac{
    \exp\lt(
    \int_0^T
    \int_0^T 
    e^{-|b-a|}
    V_{[a,b],\gamma}(|\omega_b-\omega_a|)
    ~\de b\, \de a
    \rt)
    \,
    \de \bbP^{(\eta)}(\omega)}{
    \int 
    \exp\lt(
    \int_0^T
    \int_0^T 
    e^{-|b-a|}
    V_{[a,b],\gamma}(|\omega_b-\omega_a|)
    ~\de b\, \de a
    \rt)
    \de \bbP^{(\eta)}(\omega)
    }
    \,.
\]
We then have the following approximation result.

\begin{prop}
\label{prop:polaron-FD-approx}
    For uniformly bounded functions $f$ as in Proposition~\ref{prop:FKG-GCI-polaron}, and $t,s\in [0,T]$, we have:
    \begin{align*}
    \lim_{A\to\infty}
    \int f(|\omega_t-\omega_s|) \de \wh\bbP^{(A)}_V
    &=
    \int f(|\omega_t-\omega_s|)\de \wh\bbP_V
    ,
    \\
    \lim_{\eta\to 0}
    \int f(|\omega_t-\omega_s|) \de \wh\bbP^{(A,\eta)}_V
    &=
    \int f(|\omega_t-\omega_s|) \de \wh\bbP^{(A)}_V,\quad\quad\forall A\in (0,\infty).
    \end{align*}    
\end{prop}

\begin{proof}
    Given the finiteness of partition functions and monotone convergence of $V^{(A)}_{[a,b],\gamma}$ up to $V_{[a,b],\gamma}$ as $A\to\infty$, the first result follows by the monotone convergence theorem as in \cite[Proposition 2.5]{S22}.
    Similarly to \cite[Proposition 2.6]{S22}, the second follows by the continuous mapping theorem if $f$ is uniformly bounded.
\end{proof}

\begin{proof}[Proof of Proposition~\ref{prop:FKG-GCI-polaron}]
    We will show the result when $f$ is uniformly bounded, which suffices by monotone convergence in the original statement.
    Given Proposition~\ref{prop:polaron-FD-approx}, it suffices to establish the finite-dimensional analog:
    \[
    \bbE^{\wh\bbP_{\wt V^{(A,\eta)}}}[f(|\omega_t-\omega_s|)]
    \leq 
    \bbE^{\wh\bbP_{V^{(A,\eta)}}}[f(|\omega_t-\omega_s|)].
    \]
    We will use two crucial properties.
    Firstly, by approximating the integral over $[0,T]^2$ by Riemann sums, the Radon--Nikodym derivative $\de \wh\bbP_{\wt V^{(A,\eta)}}/\de \wh\bbP_{V^{(A,\eta)}}$ is the pointwise limit of uniformly bounded products of finitely many symmetric quasi-concave functions on the finite-dimensional space $C^{(\eta)}([0,T];\bbR^3)$.\footnote{One could also use a space of piece-wise \emph{constant} functions, in which case the Riemann sum approximations would not be needed. We chose not to, to avoid considering discontinuous functions of time.}
    Second, $\wh\bbP_{V}^{(A,\eta)}$ has a Gaussian mixture representation analogous to \eqref{Gauss-rep} and Lemma~\ref{lem:MV-general-correspondence}\ref{it:general-path-conditional-law}, with exactly the same proof.
    We define $\bP_{\xi}^{(A,\eta)}$ analogously to \eqref{Gauss-rep} but with base measure $\bbP^{(\eta)}_{[0,T]}$ rather than continuous-time Wiener measure, and similarly $W^{(A,\eta)}(\xi)$ from \eqref{eq:weight-factor}.
    Then
    \begin{equation}
    \label{eq:piecewise-linear-mixture-representation}
    \begin{aligned}
    \wh\bbP_{V}^{(A,\eta)}
    &=
    \int \bP_{\xi}^{(A,\eta)}\de \wh\Theta^{(A,\eta)}_{\gamma},
    \\
    \de\wh\Theta^{(A,\eta)}_{\gamma}
    &=
    W^{(A,\eta)}(\xi)\,\de\Theta_{\gamma|_{[-A,A]}}(\xi).
    \end{aligned}
    \end{equation}
    (Note we have replaced $\gamma$ by its restriction $\gamma|_{[-A,A]}$ to $[-A,A]$.)
    The resulting mixing measure $\wh\Theta_{\gamma}^{(A,\eta)}$ is hence another tilted Poisson point process; the log super modularity result of Lemma~\ref{lem:W-LSM} still applies with exactly the same proof.

    Next define $\wt\Theta^{(A,\eta)}$ to be the reweighting
    \[
    \frac{\de\wt\Theta^{(A,\eta)}(\xi)
    }{
    \de\wh\Theta_{\gamma}^{(A,\eta)}(\xi)
    }
    \propto 
    \int \frac{\de\wh\bbP_{\wt V}^{(A,\eta)}(\omega)}{\de \wh\bbP_V^{(A,\eta)}(\omega)} \de\bP^{(A,\eta)}_{\xi}(\omega).
    \]
    These are the weights for the modification of \eqref{eq:piecewise-linear-mixture-representation} associated with $\wt V$. Indeed one sees directly that 
    \begin{align*}
    \wh\bbP^{\wt V^{(A,\eta)}}(\omega)
    &=
    \int \wt\bP^{(A,\eta)}_{\xi}(\omega)
    \de \wt\Theta^{(A,\eta)}_{\gamma}(\xi),
    \\
    \de\wt\bP^{(A,\eta)}_{\xi}(\omega)
    &\propto
    \frac{\de\wh \bbP^{(A,\eta)}_{\wt V}(\omega)}{\de\wh \bbP^{(A,\eta)}_{V}(\omega)}
    \de\bP^{(A,\eta)}_{\xi}(\omega)
    \\
    &=
    \exp\lt(
    \int_0^T
    \int_0^T 
    e^{-|b-a|}
    \lt(
    \wt V_{[a,b],\gamma}(|\omega_b-\omega_a|)
    -
    V_{[a,b],\gamma}(|\omega_b-\omega_a|)
    \rt)
    ~\de b\, \de a
    \rt)
    \de\bP^{(A,\eta)}_{\xi}(\omega).
    \end{align*}
    (We emphasize that the notation $\wt \bP$ here hides implicit dependence on both $V$ and $\wt V$.)
The desired inequality now follows from:
    \begin{align*}
    \int f(|\omega_t-\omega_s|)\de  \wh\bbP^{\wt V^{(A,\eta)}}(\omega)
    &=
    \iint 
    f(|\omega_t-\omega_s|)
    \de\wt\bP^{(A,\eta)}_{\xi}(\omega)
    \de \wt\Theta^{(A,\eta)}_{\gamma}(\xi)
    \\
    &\leq 
    \iint 
    f(|\omega_t-\omega_s|)
    \de\bP^{(A,\eta)}_{\xi}(\omega)
    \de \wt\Theta^{(A,\eta)}(\xi)
    \\
    &\leq 
    \iint 
    f(|\omega_t-\omega_s|)
    \de\bP^{(A,\eta)}_{\xi}(\omega)
    \de \Theta^{(A,\eta)}_{\gamma}(\xi)
    \\
    &=\int f(|\omega_t-\omega_s|)\de  \wh\bbP^{(A,\eta)}_V(\omega).
    \end{align*}
    The first inequality above follows by Proposition~\ref{prop:GCI}. Indeed the Radon--Nikodym derivative
    \[
    \frac{\de\wt\bP^{(A,\eta)}_{\xi}(\omega)}{\de\bP^{(A,\eta)}_{\xi}(\omega)}
    \propto 
    \frac{\de\wh \bbP^{(A,\eta)}_{\wt V}(\omega)}{\de\wh \bbP^{(A,\eta)}_{V}(\omega)}
    \]
    is the pointwise limit of a uniformly bounded sequence of functions that are finite products of symmetric quasi-concave functions, while $\omega\mapsto f(|\omega_t-\omega_s|)$ is symmetric and \textbf{quasi-convex}.
    The second inequality holds because $\wt\Theta^{(A,\eta)}\succeq \wh\Theta^{(A,\eta)}$ by exactly the same proof as in Lemma~\ref{lem:MV-FKG}, so the argument from Corollary~\ref{cor:compare} applies. This concludes the proof.
\end{proof}

\section{Stationary point processes, Palm measures, and random intensities.}\label{sec-pointprocess}

In this section, we consider a generic simple point process $N$ in $\R$, i.e., random measures supported on atoms,  living in a probability space $(\Omega,\mathcal{F},P)$ such that $N(\{x\})\in \{0,1\}. $ We will usually refer to it as a quadruple $(\Omega,\mathcal{F},P,N)$. The point process can be characterized by its support, namely, $N(\cdot)=\sum_{i\in \Z}\delta_{r_i}(\cdot)$ -- more precisely, if $(r_i)_{i\in \Z}$ is a sequence of ordered (random) real numbers with the convention that 
\begin{equation}\label{eq-ordered-ints}
  \cdots<r_{-2}<r_{-1}<r_0\leq 0<r_1<r_2<\cdots, 
\end{equation}
then for every $\omega\in \Omega$ and Borel set $C\subset \R$, 
$N(\omega,C)= \#\{i \in \Z: r_i(\omega)\in C\}$ denotes the number of indices $i\in \Z$ such that 
$r_i=r_i(\omega) \in C$. On $\Omega$, the shifts $(\theta_t)_{t\in \R}$ act via $N(\theta_t \omega,C):=N(\omega,C+t)$ for a Borel set $C\subset \R$. We say that the point process is \textit{stationary} if $P\circ\theta_t=P$ for each $t\in \R$. We define a measure on $\R$  by 
the expected number of points 
$\lambda(C):=E^P[N(C)]$ on Borel sets $C\subset \R$. If the point process is stationary, then $\lambda$ is a multiple of the Lebesgue measure so that there is a constant $m>0$ such that $\lambda(C)=m|C|$ for each Borel set $C\subset \R$. We call $m$ the \textit{intensity} of $N$. 

\begin{definition}[Palm measure]
  Let $(\Omega,\mathcal{F},P,N)$ be a stationary point process with positive intensity $m>0$. Let $C$ be any Borel set of positive and finite Lebesgue measure $|C|$. Then we define the (normalized) Palm measure $P_0$ on $(\Omega,\mathcal{F})$ as \begin{equation}\label{eq-palm-meas-def}
    P_0(A):=\frac{1}{m|C|}E\Big[\sum_{n\in \Z}\mathbbm{1}_A(\theta_{r_n})\mathbbm{1}_C(r_n)\Big],\qquad A\in \mathcal{F}.
  \end{equation}
\end{definition}
We take note of the following consequences of the above definition: First, since the point process $(\Omega,\mathcal{F},P,N)$ is stationary, the above definition is independent of the set $C$. Moreover, from the definition we can see that $P_0$ is concentrated on $\Omega_0:=\{\omega: r_0(\omega)=0\}$ -- indeed, for each $n\in \Z$ and $\omega\in \Omega$, 
$\1_{\{r_0=0\}}(\theta_{r_n}\omega)=1$ if and only if $N(\theta_{r_n}\omega,\{0\})=N(\omega,\{r_n\})=1$, which is true by definition. Thus, under $P_0$, $N$ is concentrated on the set of the point processes with an atom at the origin. The following lemma justifies our interest in the Palm measure since it allows us to see the point process from the ``point of view of the particle":
\begin{lemma}\cite[Statement 1.2.16]{BB03}
  Let $\theta:\Omega_0\mapsto \Omega_0$ defined as $\theta:=\theta_{r_1}$ with inverse $\theta^{-1}:=\theta_{r_{-1}}$. Then $P_0$ is invariant under $\theta$.  In particular, $(r_n- r_{n-1})_{n\in \Z}$ is stationary under $P_0$. 
\end{lemma}

The following result allows us to express $P$ in terms of $P_0$, so that we can go back and forth between the two measures:
\begin{lemma}[Inversion formula]\cite[Eq.1.2.25]{BB03}
  For a stationary point process $(\Omega,\mathcal{F},P,N)$ with Palm measure $P_0$, the following holds for any nonnegative measurable function $f$:
  \begin{equation}\label{eq-inversion-formula}
  E[f]=mE_0\left[\int_{0}^{r_1}(f\circ\theta_t)\d t\right]  .
  \end{equation}
\end{lemma}
Setting $f=1$ in the previous Lemma, we deduce that \begin{equation}\label{eq-P0-moment-s1}
  E_0[r_1]=\frac{1}{m}.
\end{equation}
The previous facts can be summarized as follows:
\begin{theorem}\cite[Theorem 13.3.I]{DJ08}
  There is a one-to-one correspondence between stationary point process with intensity $m\in(0,\infty)$ and stationary sequences of nonnegative random variables $(\tau_n)_{n\in \Z}\in \mathcal{T}^+$ with mean $\frac{1}{m}$.\footnote{Here $\mathcal T^+$ denotes the space of doubly-infinite sequences with non-negative entries, and $\mathcal B(\mathcal T^+)$ denotes the Borel $\sigma$-algebra.} More precisely, for a sequence $N_0:=(r_{n})_{n\in \Z}$ satisfying \eqref{eq-ordered-ints} and $r_0=0$, define the mapping $\Psi(N_0):=(\Psi(N_0))_{n\in \Z}$, where  $\Psi(N_0)_n:=r_n-r_{n-1}$. Then the correspondence is given by 
  \begin{equation*}
    \begin{aligned}
      &\Psi: (\Omega,\mathcal{F},P,N)\longrightarrow (\mathcal{T}^+,\mathcal{B}(\mathcal{T}^+),P_0\circ \Psi^{-1})\\
      &\Psi^{-1}: (\mathcal{T}^+,\mathcal{B}(\mathcal{T}^+),\Pi)\longrightarrow (\Omega, \mathcal{F},\tilde{P}),
    \end{aligned}
  \end{equation*}
  where in the second direction, $\tilde{P}$ is defined as in \eqref{eq-inversion-formula} with replacing $P$ by  $\tilde{P}$ and $P_0$ by $\Pi\circ \Psi$.
  \end{theorem}

Next, we relate the notions of ergodicity under $P$ with the family $(\theta_t)_{t\in \R}$  and under $P_0$ with $\theta=\theta_{s_1}$. 
\begin{lemma}\cite[Properties 1.6.1-1.6.2]{BB03} The following holds:
  \begin{enumerate}
    \item Let $A\in \mathcal{F}$ be invariant under $(\theta_t)_{t\in \R}$. Then $P(A)=1$ if and only if $P_0(A)=1$.
    \item Let $A\in \mathcal{F}$ be invariant under $\theta$. Then $P(A)=1$ if and only if $P_0(A)=1$.
  \end{enumerate}
\end{lemma}

\begin{lemma}\cite[Property 1.6.3]{BB03}
  $(\Omega,\mathcal{F},P,(\theta_t)_{t\in \R})$ is ergodic if and only if $(\Omega,\mathcal{F},P_0,\theta)$ is ergodic. In that case, if \begin{equation}
    \begin{aligned}
      A&:=\Big\{\lim_{T\to\infty}\frac{1}{2T}\int_{-T}^T (f\circ\theta_t)\d t=E[f]\Big\},\quad f\in L^1(P)\\
      A'&:=\Big\{\lim_{n\to\infty}\frac{1}{2n}\sum_{i=-n}^n f\circ\theta_{r_i}=E_0[f]\Big\},\quad f\in L^1(P_0),
    \end{aligned}
  \end{equation}
  then \begin{equation*}
    P(A)=P_0(A)=P(A')=P_0(A')=1.
  \end{equation*}
\end{lemma}

Next, we will deduce some consequences from the previous results for stationary and ergodic point processes in $\R$. \begin{lemma}\label{lemma:ergodic-averages-pp}
  Let $(\Omega,\mathcal{F},P,N)$ be a stationary and ergodic point process on $\R$ with intensity $m$ and Palm measure $P_0$. Then the following holds $P$-a.s. (and hence also $P_0$-a.s.):\begin{enumerate}
    \item \begin{equation}\label{eq-ergodic-averages-pp-1}
      \lim_{T\to\infty}\frac{N([-T,T])}{2T}=m,
    \end{equation}
    \item \begin{equation}\label{eq-ergodic-averages-pp-2}
      \lim_{n\to\infty}\frac{1}{2n}\sum_{i=-n}^{n-1}(r_{i+1}-r_i)=\frac{1}{m},
    \end{equation}
    \item \begin{equation}\label{eq-ergodic-averages-pp-3}
      \lim_{T\to\infty}\frac{r_{N([0,T])}}{T}=1.
    \end{equation}
    \item \begin{equation}\label{eq-ergodic-averages-pp-4}
      \lim_{n\to\infty}\frac{1}{2n}\sum_{i=-n}^{n-1}(r_{i+1}-r_i)\mathbbm{1}\{r_{i+1}-r_i> c\}=\frac{1}{m}P(r_1-r_0>c),\quad c>0.
    \end{equation}
  \end{enumerate}
\end{lemma}
\begin{proof} We first prove Part (i). 
  By stationarity, it is enough to prove that $P$-a.s. \begin{equation*}
    \lim_{T\to\infty}\frac{N((0,T])}{T}=m.
  \end{equation*}
  To check it, note first that for $n\in \N$, $N((0,n])=\sum_{i=0}^{n-1}N((i,i+1])=\sum_{i=0}^{n-1}N((0,1])\circ \theta_{i}$, so that, by the ergodic theorem, \begin{equation*}
    \lim_{n\to\infty}\frac{N((0,n])}{n}=E[N(0,1]]=m.
  \end{equation*}
  Using that  \begin{equation*}
    \frac{N((0,n])}{n+1}\leq \frac{N((0,T])}{T}\leq \frac{N((0,n+1])}{n}
  \end{equation*}
  if $n<T\leq n+1$, we can extend the limit over $T\in \R$.

  We now prove Part (ii). By ergodicity with respect to $P_0$, and recalling that $r_1-r_0=r_1$ $P_0$-a.s., 
    \begin{equation*}
    \lim_{n\to\infty}\frac{1}{2n}\sum_{i=-n}^{n-1}(r_{i+1}-r_i)=\lim_{n\to\infty}\frac{1}{2n}\sum_{i=-n}^{n-1}(r_{1}-r_0)\circ\theta^i=E_0[r_1]=\frac{1}{m},
  \end{equation*}
  where in the last equality, we used \eqref{eq-P0-moment-s1}.
  
  Note that Part (iii) is a direct consequence of (i) and (ii), since $\frac{N((0,T])}{T}\to m$ and $\frac{r_n}{n}\to \frac{1}{m}$.
  We now prove Part (iv), for which we apply the ergodic theorem to conclude that $P_0$-a.s., \begin{align*}
    \lim_{n\to\infty}\frac{1}{2n}\sum_{i=-n}^{n-1}(r_{i+1}-r_i)\mathbbm{1}\{r_{i+1}-r_i> c\}&=
        \lim_{n\to\infty}\frac{1}{2n}\sum_{i=-n}^{n-1}(r_{1}-r_0)\mathbbm{1}\{r_1-r_0>c\}\circ\theta^i
        \\
        &=E_0[r_1,r_1> c].
  \end{align*}
  Finally, applying the inversion formula \eqref{eq-inversion-formula} to $f=\mathbbm{1}\{r_1>c\}$ leads to \begin{equation*}
    E_0[r_1,r_1\geq c]=\frac{1}{m}P(r_1>c).
  \qedhere
  \end{equation*}
\end{proof}
Let us also remark that the distribution of a point process can be identified uniquely by its \textit{Laplace functional}. More precisely, if $(\Omega,\mathcal{F},P,N)$ is a point process in $\R$, its Laplace functional $L_N$ is defined on nonnegative, measurable functions $u:\R\mapsto [0,\infty)$ by 
\begin{equation}\label{eq-Laplace-fun-def}
  L_N(u):=E\bigg[\exp\bigg(-\int u(x)N(dx)\bigg)\bigg]=E\Big[\exp\big(-\sum_{i\in \Z}u(r_i)\big)\Big].
\end{equation}

We will be interested in a particular class of stationary and ergodic point process on $\R$, the so-called {\it Poisson point process with random intensity}:
\begin{definition}
  Let $\mu$ be a random measure on $\R$, i.e., given a probability space $(\widehat\Omega, \widehat{\mathcal F}, \widehat P)$, $\mu:\widehat\Omega\to \Mcal_{\mathrm{loc}}(\R)$ is a random variable taking values on the space of locally finite measures on $\R$.  
  A point process N on $\R$ is called a Poisson process with random intensity $\mu$ (or a Poisson process directed by a random measure $\mu$), if, conditionally on the random measure $\mu$, $N$ is a Poisson point process with intensity measure $\mu$, that is, \begin{equation*}
    P\big(N(C|\mu(\widehat\omega,\cdot))=k\big)=\frac{\mu(\widehat\omega,C)^k \e ^{-\mu(\widehat\omega,C)}}{k!}, \qquad k\in \N\cup\{0\}, C\in \Bcal(\R). 
  \end{equation*}
\end{definition}
The Laplace functional of a Poisson process directed by random intensity $\mu$ defined on a probability space $(\widehat\Omega, \widehat{\mathcal F}, \widehat P)$ is given by 
\begin{equation}\label{eq-Laplace-fun-cox}
  L_N(u):=\E^{\widehat P}\bigg[\exp\Big(-\int_\R (1-\e^{u(x)})\mu(\cdot,\d x)\Big)\bigg]= \int_{\widehat\Omega} \widehat P(\d\widehat\omega) \exp\Big(-\int_\R (1-\e^{u(x)})\mu(\widehat\omega,\d x)\Big)
  \end{equation} 
and it also uniquely characterizes its distribution. 
Stationarity and ergodicity of this point process can be determined by its directing measure.

\begin{lemma} \label{lemma:stat-erg-pp}
Let $(\Omega,\mathcal{F},P,N)$ be a Poisson process directed by the random measure $\mu$ on the probability space $(\widehat\Omega,\widehat{\mathcal{F}},\widehat P)$. 
  \begin{enumerate}
    \item \cite[Proposition 6.1.I]{DJ03} $N$ is stationary if and only if its Laplace functional is stationary, i.e., $L_N(\theta_t u)=L_N(u)$ for each measurable $u:\R\mapsto [0,\infty)$.
    \item \cite[Proposition 12.3.VII]{DJ08}If $N$ is stationary, then it is also ergodic if and only if the distribution $\widehat P[\mu \in \cdot]$ of $\mu$ under $\widehat P$ is ergodic.
  \end{enumerate}
\end{lemma}

\appendix


\noindent{\bf Acknowledgement:} The second author would like to thank Herbert Spohn (Munich) for very helpful and inspiring comments on the previous version of the manuscript. The first two authors are supported by the Deutsche Forschungsgemeinschaft (DFG) under Germany's Excellence Strategy EXC 2044 - 390685587, Mathematics M\"unster: Dynamics - Geometry - Structure.




\end{document}